\documentclass[letterpaper,11pt]{article}
\usepackage{amsmath, amssymb}
\usepackage{color}
\usepackage{fullpage}
\usepackage[numbers]{natbib}
\usepackage[noend]{algorithmic}
\usepackage{tikz}
\usepackage{enumerate}
\usepackage{hyperref}
\usepackage{graphicx}
\usepackage{caption}
\usepackage{hyperref}
\usepackage{subcaption}
\usepackage{zerosum,csquotes}
\usepackage{enumitem,linegoal}
\usepackage[linesnumbered,ruled,vlined]{algorithm2e}
\usepackage{comment}	
\usepackage{ctable}

\newtheorem{observation}{Observation}[section]
\usepackage{mathtools}
\usepackage[flushleft]{threeparttable}
\usepackage[normalem]{ulem}

\usepackage{footnote}
\makesavenoteenv{tabular}
\makesavenoteenv{table}

\usepackage[margin=1in]{geometry}

 \newtheorem{theorem}{Theorem}[section]
 \newtheorem{lemma}[theorem]{Lemma}
 
 \newtheorem{corollary}[theorem]{Corollary}
 
 \newtheorem{definition}[theorem]{Definition}

 \newtheorem{remark}[theorem]{Remark}
\makeatletter
\def\GrabProofArgument[#1]{ #1: \egroup\ignorespaces}
\def\proof{\noindent\textbf\bgroup Proof%
	\@ifnextchar[{\GrabProofArgument}{. \egroup\ignorespaces}}

\makeatother

\newcommand{\agents}{\mathcal{N}}
\newcommand{\items}{\mathcal{M}}
\newcommand{\ite}{b}
\newcommand{\agent}{a}
\newcommand{\inc}{\mathsf{inc}}
\newcommand{\satagents}{\mathcal{S}}

\newcommand{\cone}{\mathcal{C}_1}
\newcommand{\ctwo}{\mathcal{C}_2}
\newcommand{\cthree}{\mathcal{C}_3}

\newcommand{\wcal}{\mathcal{W}}
\newcommand{\agentsv}{\mathcal{Y}}
\newcommand{\itemsv}{\mathcal{X}}
\newcommand{\itemv}{x}
\newcommand{\agentv}{y}

\newcommand{\parttwo}{\hat{\mathcal{Y}}}
\newcommand{\partone}{\hat{\mathcal{X}}}

\newcommand{\vone}{\hat{x}}
\newcommand{\vtwo}{\hat{y}}

\newcommand{\firstset}{f}
\newcommand{\secondset}{g}

\newcommand{\valu}{V}
\newcommand{\domp}{\mathsf{ground}}
\newcommand{\MMS}{\mathsf{MMS}}
\newcommand{\poly}{\mathsf{poly}}
\newcommand{\fitems}{\mathcal{F}}

\newcommand{\epsteinefficient}{Epstein and Levin}
\newcommand{\feigemaximizing}{Feige}
\newcommand{\chakrabartyallocating}{Chakrabarty, Chuzhoy,and Khanna}
\newcommand{\asadpourapproximation}{Asadpour and Saberi}
\newcommand{\procacciafirst}{Procaccia and Wang}
\newcommand{\steinhausfirst}{Steinhaus}
\newcommand{\foleyfirst}{Foley}
\newcommand{\dubinsfirst}{Dubins and Spainer}

\newcommand{\amanatidisapproximation}{Amanatidis \textit{et al.}} 
\newcommand{\goemansapproximating}{Goemans \textit{et al.}} 
\newcommand{\bansalsanta}{Bansal and Sviridenko} 
\newcommand{\amanatidisapproximationful}{Amanatidis, Markakis, Nikzad, and Saberi}   
\newcommand{\bagfilling}{\textsf{bag filling}}

\newcommand{\MCMWM}{\textsf{MCMWM}}
\newcommand{\ceil}[2]{#1^{#2}}

\newcommand*\samethanks[1][\value{footnote}]{\footnotemark[#1]}

\newcounter{proccnt}

\newcommand{\konote}[1]{}

\title{Fair Allocation of Indivisible Goods: Improvement and Generalization
}
\author{
	Mohammad Ghodsi \thanks{Sharif University of Technology} 
	\and MohammadTaghi HajiAghayi \thanks{University of Maryland}
	\thanks{Supported in part by NSF CAREER award CCF-1053605,  NSF BIGDATA grant IIS-1546108, NSF AF:Medium grant CCF-1161365, DARPA GRAPHS/AFOSR grant FA9550-12-1-0423, and another DARPA SIMPLEX grant.
		Portions of this research were completed while the first and the second authors were visitors at the Simons Institute for the Theory of Computing.}
	\and Masoud Seddighin \samethanks[1]
	\and Saeed Seddighin \samethanks[2] \samethanks[3]
	\and Hadi Yami \samethanks[2] \samethanks[3]
}

\begin{document}
	\newcommand{\ignore}[1]{}
\renewcommand{\theenumi}{(\roman{enumi}).}
\renewcommand{\labelenumi}{\theenumi}
\sloppy

%
%

\date{}

\maketitle
\textit{``Being good is easy, what is difficult is being just."} \cite{hugo2000miserables} (Victor Hugo, 1862)

\thispagestyle{empty}

\begin{abstract}
We study the problem of fair allocation for indivisible goods. We use the \textit{the maxmin share} paradigm introduced by Budish~\cite{Budish:first} as a measure for fairness. \procacciafirst ~\cite{Procaccia:first} (EC'14) were first to investigate this fundamental problem in the additive setting. In contrast to what real-world experiments suggest, they show that a maxmin guarantee (1-$\MMS$ allocation) is not always possible even when the number of agents is limited to 3. While the existence of an approximation solution (e.g. a $1/2$-$\MMS$ allocation) is quite straightforward, improving the guarantee becomes subtler for larger constants.
\procacciafirst ~\cite{Procaccia:first}\footnote{Recipient of the best student paper award at EC'14.} provide a proof for existence of a $2/3$-$\MMS$ allocation and leave the question open for better guarantees.

Our main contribution is an answer to the above question. We improve the result of \procacciafirst\! to a $3/4$
 factor in the additive setting. The main idea for our $3/4$-$\MMS$ allocation method is clustering the agents. To this end, we introduce three notions and techniques, namely \textit{reducibility}, \textit{matching allocation}, and \textit{cycle-envy-freeness}, and prove the approximation guarantee of our algorithm via non-trivial applications of these techniques. Our analysis involves coloring and double counting arguments that might be of independent interest.

One major shortcoming of the current studies on fair allocation is the additivity assumption on the valuations. We alleviate this by extending our results to the case of submodular, fractionally subadditive, and subadditive settings. More precisely, we give constant approximation guarantees for submodular and XOS agents, and a logarithmic approximation for the case of subadditive agents. Furthermore, we complement our results by providing close upper bounds for each class of valuation functions. Finally, we present algorithms to find such allocations for additive, submodular, and XOS settings in polynomial time. The reader can find a summary of our results in Tables \ref{resultstable} and \ref{resultstable2}. 

\end{abstract}
\section{Introduction}\label{introduction}
Suppose we have a set of $m$ indivisible items, and wish to distribute them among $n$ agents. Agents have  valuations for each set of items that are not necessarily identical. How hard is it to divide the items between the agents to make sure everyone receives a fair share?

\setcounter{page}{1}
Fair division problems have been vastly studied in the past 60 years, (see, e.g.~\cite{amanatidis2015approximation,asadpour2010approximation, Bouveret:first,Budish:first, Dubins:first,Procaccia:first,Steinhaus:first,alijani2017envy}). This line of research was initiated by the work of {\steinhausfirst}   in 1948~\cite{Steinhaus:first} in which the author introduced the \emph{cake cutting} problem as follows: given a heterogeneous cake and a set of agents with different valuation functions, the goal is to find a fair allocation of the cake to the agents.

In order to study this problem, several notions of fairness are proposed, the most famous of which are \emph{proportionality} and \emph{envy-freeness}, introduced by {\steinhausfirst} in 1948~\cite{Steinhaus:first} and {\foleyfirst} in 1967~\cite{Foley:first}. A division is called proportional, if the total value of the allocated pieces to each agent is at least $1/n$ fraction of his total value for the entire cake, where $n$ is the number of agents. In an envy-free division, no agent wishes to exchange his share with another agent, i.e., every agent's valuation for his share is at least as much as his valuation for the other agents' shares. Clearly, proportionality is implied by envy-freeness.

{\dubinsfirst} in 1961~\cite{Dubins:first} propose a simple {\emph{moving knife} procedure that can guarantee a proportional division of the cake. For envy-freeness, Selfridge and Conway design an algorithm that guarantees envy-freeness when the number of agents is limited to 3. Later, Brams and Taylor extend this guarantee to an arbitrary number of agents in the additive setting~\cite{brams}.

The problem becomes even more subtle when we assume the items are indivisible. It is not hard to show that for indivisible items, neither of proportionality nor envy-freeness can be guaranteed; for instance, when the number of items is smaller than the number of agents, at least one agent receives no items. 


From a theoretical standpoint, proportionality and envy-freeness are too strong to be delivered in the case of indivisible goods. Therefore, Budish~\cite{Budish:first} proposed a newer notion of fairness for indivisible goods, namely \textit{the maxmin share}, which attracted a lot of attention in recent years \cite{Procaccia:first,amanatidis2015approximation,kurokawa2015can,Bouveret:first,caragiannis2016unreasonable,barman2017approximation,suksompong2017approximate,farhadi2017fair}. Imagine that we ask an agent $\agent_i$ to partition a set $\cal M$ of $m$ items into $n$ bundles and collect the bundle with the smallest value. To maximize his profits, agent $\agent_i$ tries to divide $\cal M$ in a way that maximizes the value of the bundle with the lowest value to him. Based on this, the maxmin share of an agent $\agent_i$, denoted by $\MMS_i$, is the value of the least valuable bundle in agent $\agent_i$'s allocation; that is, the maximum profit $\agent_i$ can obtain in this procedure. Clearly, $\MMS_i$ is the most that can be guaranteed to an agent, since if all valuations are the same, at least one agent obtains a valuation of at most $\MMS_i$ from his allocated set. The question is then, whether $\MMS_i$ is a feasible guarantee? Therefore, we call an allocation $\MMS$, if every agent $\agent_i$ receives a collection of items that are together worth at least $\MMS_i$ to him. Bouverret~\cite{Bouveret:first} showed that for the restricted cases, when the valuations of the items for each agent are either $0$ or $1$, or when $m \leq n+3$, an $\MMS$ allocation is guaranteed to exist. In other words, each $\agent_i$ can be guaranteed to receive a profit of at least $\MMS_i$ from his allocated items.


While the experiments support the existence of an $\MMS$ allocation in general~\cite{Bouveret:first}, this theory was refuted by the pioneering work of \procacciafirst~\cite{Procaccia:first}. {\procacciafirst}~\cite{Procaccia:first} provided a surprising counter-example that admits no $\MMS$ allocation. They also show that a $2/3$-$\MMS$ allocation always exists, i.e. there exists an algorithm that allocates the items to the agents in such a way that every agent $\agent_i$ receives a share that is worth at least $2/3 \MMS_i$ to him.
 In particular, they show for $n \leq 4$, their algorithm finds a $3/4$-$\MMS$ allocation. However, their algorithm does not run in polynomial time unless we assume the number of agents is bounded by a constant number. In a recent work, \amanatidisapproximationful~\cite{amanatidis2015approximation}, improve this result by presenting a PTAS algorithm for finding a $(2/3 - \epsilon)$-$\MMS$ allocation to any number of agents. However, the heart of their algorithm is the same as~\cite{Procaccia:first}. In addition to this, \amanatidisapproximation\enspace prove that for $n=3$, a $7/8$-$\MMS$ allocation is always possible. Note that, the counter example provided by \procacciafirst~\cite{Procaccia:first}  requires a number of goods that is exponential to the number of agents. Kurokawa, Procaccia, and Wang in ~\cite{kurokawa2015can} provided a better construction for the counter-example with a linear number of goods.

 In this work, we improve the result of \procacciafirst~\cite{Procaccia:first} by proving that a $3/4$-$\MMS$ allocation always exists. We also give a PTAS algorithm to find such an allocation in polynomial time. Of course, this only holds if the valuation of the agents for the items are additive. We further go beyond the additive setting and extend this result to the case of submodular, XOS, and subadditive settings. More precisely, we give constant approximation algorithms for submodular and XOS settings that run in polynomial time. For the subadditive case, we prove that a $1/10 \lceil\log m\rceil$-$\MMS$ allocation is guaranteed to exist. We emphasize that finding the exact value of $\MMS_i$ for an agent is NP-hard. Furthermore, to the best of our knowledge, no PATS is known for computing the  $\MMS$ values  in non-additive settings. Thus, any $\alpha$-$\MMS$ allocation algorithm in non-additive settings must overcome the difficulty that the value of $\MMS_i$ is not known in advance. Therefore, our algorithms don't immediately follow from our existential proofs.

In order to present the results and techniques, we briefly state the fair allocation problem. Note that you can find a formal definition of the problem with more details in Section \ref{prelim}. The input to a maxmin fair allocation problem is a set $\items$ of $m$ items and a set $\agents$ of $n$ agents. Fix an agent $\agent_i \in \agents$ and let $V_i: 2^\items \rightarrow \mathbb{R}^+$ be the valuation function of $\agent_i$. Consider the set $\Pi_r$ of all partitionings of the items in $\items$ into $r$ non-empty sets. We define $\MMS_{V_i}^r(\items)$ as follows:
 $$\MMS_{V_i}^r(\items) =  \max_{P^* = \langle P^*_1,P^*_2,\ldots,P^*_r\rangle \in \Pi_r} \min_{1 \leq j \leq r} V_i(P^*_j).$$
In the context of fair allocation, we denote the maxmin value of an agent $\agent_i$ by $\MMS_i = \MMS^n_{V_i}(\items)$. The fair allocation problem is defined as follows: \textit{for a given parameter $\alpha$, can we distribute the items among the agents in such a way that every agent $\agent_i$ receives a set of items with a value of at least $\alpha \MMS_i$ to him?} Such an allocation is called an $\alpha$-$\MMS$ allocation. We consider the fair allocation problem in both additive and non-additive settings (including submodular, XOS, and subadditive valuations).  For non-additive settings, we use oracle queries to  access the valuations. Note that, for non-additive settings, eliciting the exact valuation function of each agent needs an exponential number of queries. However, our methods for allocating the items in non-additive settings only uses a polynomial number of queries.

There are many applications for finding fair allocations in the additive and non-additive settings. For example,  
\emph{spliddit}, a popular fair division website\footnote{
\url{http://www.spliddit.org}} suggests indisputable and provably fair solutions for many real-world problems such as sharing rents, distributing tasks, dividing goods, etc. For dividing goods, \emph{spliddit} uses the maximum Nash welfare allocation (the allocation that maximizes the product of utilities). In \cite{caragiannis2016unreasonable} (EC'17), Caragiannis et.al., proved with a tight analysis that a maximum Nash welfare allocation is a $2/(1+\sqrt{4n-3})$-$\MMS$ allocation. However, the current best approximation gurantee and the state-of-the-art method for allocating indivisible goods is based on the result of ~\cite{Procaccia:first} that guarantees a $2/3$-$\MMS$ allocation. We believe our results can improve their performance.  It is worth mentioning that despite the complexity of analysis, the idea behind our algorithm is simple and it can be easily implemented. The reader can find a set of materials including the implementation of our method and an animated explanation of our algorithm in \href{https://www.cs.umd.edu/\~saeedrez/fair.html}{https://www.cs.umd.edu/$\sim$saeedrez/fair.html}.

\subsection{Relation to other Fundamental Problems}
In this work, we study the allocation of indivisible items to maximize fairness. However, maximizing fairness is not the only goal that has been considered in the literature. In the following, we briefly explain other variants of this problem that are very fundamental and have received a lot of attention in recent years.

\textbf{Welfare maximization:} Perhaps the simplest version of the resource allocation problem with indivisible goods is \textit{the welfare maximization problem}. In this problem, we are given a set $\items$ of goods, and the goal is to allocate the items to a number of agents to maximize the welfare. The problem is trivial if we assume the agents to be additive. Therefore, the main focus has been on submodular, XOS, and subadditive agents~\cite{dobzinski2005approximation,dobzinski2006improved,feige2009maximizing,lehmann2001combinatorial}(STOC'05, SODA'06, STOC'06). \feigemaximizing~\cite{feige2009maximizing} gives tight algorithms that solve the problem in the subadditive and XOS settings with approximation factors $1/2$ and $1-1/e$, respectively. The approximation factors match the existing lower bounds for these settings.

\textbf{Max-min allocation:} Another variant of this problem is to maximize the least value that any agent obtains from the allocation.  \asadpourapproximation~\cite{asadpour2010approximation}(STOC'07) give the first polynomial time algorithm for this problem that approximates the optimal solution within a factor of $O(\sqrt{n}\log^3n)$ in the additive setting. This was later improved by \chakrabartyallocating~\cite{chakrabarty2009allocating}(FOCS'09) to an $O(m^\epsilon)$ approximation factor. This problem has also been studied with non-additive agents by \goemansapproximating~\cite{goemans2009approximating}(SODA'09). They give an $O(\sqrt{m}n^{1/4}\log n \log^{3/2}m)$ approximation algorithm that runs in polynomial time and solves the problem when the agents valuations are submodular.

\textbf{Santa Claus:} A special case of the above problem in which the valuation of every agent for an item is either 0, or a fixed value is called \textit{the Santa Claus problem}. This problem was first introduced by \bansalsanta~\cite{bansal2006santa} in STOC'09. In this paper, they give an $O(\log \log n/\log \log \log n)$ approximation algorithm for this problem that runs in polynomial time. Later Feige~\cite{feige2008allocations}(SODA'08) showed that the objective value of the problem can be approximated within a constant factor in polynomial time. This was later turned into a constructive proof by Annamalai \textit{et al.}~\cite{annamalai2015combinatorial}(SODA'14).
\subsection{Our Results and Techniques}\label{results}
Throughout this paper, we study the fair allocation problem for additive and non-additive agents.
\procacciafirst~\cite{Procaccia:first} study the fair allocation problem and show a $2/3$-$\MMS$ allocation is guaranteed to exist for any number of additive agents. We improve this result in two different dimensions: (i) we improve the factor $2/3$ to a factor $3/4$ for additive agents. (ii) we provide similar guarantees for submodular, fractionally subadditive, and subadditive agents. Moreover, we provide algorithms that find such allocations in polynomial time. A brief summary of our results is illustrated in Tables \ref{resultstable} and \ref{resultstable2}.

\subsubsection{Additive Setting}

While the existence of a $1/2$-$\MMS$ allocation is trivial in additive setting (see the rest of this section for more details), obtaining a better bound is more complicated. As mentioned before, the pioneering work of \procacciafirst ~\cite{Procaccia:first} presented the first proof to the existence of a $2/3$-$\MMS$ allocation in the additive setting. 
On the negative side, they show that their analysis is tight, i.e. 
their method cannot be used to obtain a better approximation guarantee. However, whether or not a better bound could be achieved via a more efficient algorithm remains open as \procacciafirst~\cite{Procaccia:first} pose this question as an open problem.

We answer the above question in the affirmative. Our main contribution is a proof to the existence of a $3/4$-$\MMS$ allocation for additive agents. Furthermore, we show that such an allocation can be found in polynomial time. This result improves the work of \procacciafirst ~\cite{Procaccia:first} and \amanatidisapproximation ~\cite{amanatidis2015approximation} where the former gives a proof to the existence of a $2/3$-$\MMS$ allocation and the latter presents a PTAS algorithm for finding a $2/3$-$\MMS$ allocation.

\begin{table}[t]
	\caption{Summary of the results}
	\hspace{-1cm}\begin{tabular}{|l|c|c|c|c|}
		\hline
		& Additive & Submodular & XOS & Subadditive\\
		\hline
		Previous work (existential proof) & $2/3$-$\MMS$  ~\cite{Procaccia:first} & $1/10$-$\MMS$ \cite{barman2017approximation}\footnote{In a parallel work to ours, Barman and Murthy in \cite{barman2017approximation} (\textbf{EC'17}) considered the submodular case and proposed a $1/31$ approximation guarantee.}&- &-\\
		Previous work (polytime algorithm) &  $2/3-\epsilon$-$\MMS$  ~\cite{amanatidis2015approximation} & $1/31$-$\MMS$ ~\cite{barman2017approximation} & - & -\\
		Previous work (upper bound) & $1-\epsilon$-$\MMS$ ~\cite{Procaccia:first} & - & - & -\\
		\hline
		\color{magenta}Our results \color{black} (existential proof) & $3/4$-$\MMS$ & $1/3$-$\MMS$ & $1/5$-$\MMS$  & $1/10 \lceil\log m\rceil$-$\MMS$ \\
		& \color{magenta}Theorem \ref{34main}\color{black} & \color{magenta}Theorem \ref{submodulartheorem}\color{black}& \color{magenta} Theorem \ref{xosproof} \color{black}& \color{magenta} Theorem \ref{subadditiveproof} \color{black}\\
		\hline
		\color{magenta}Our results \color{black} (polytime algorithm) &  $3/4-\epsilon$-$\MMS$ & $1/3$-$\MMS$ & $1/8$-$\MMS$ & -\\
		
		& \color{magenta} Theorem \ref{addpoly} \color{black} & \color{magenta} Theorem \ref{subsubalg} \color{black} & \color{magenta} Theorem \ref{xa} \color{black} & \\
		
		\hline
		
		\color{magenta}Our results \color{black} (upper bound) & - & $3/4$-$\MMS$ & $1/2$-$\MMS$ & $1/2$-$\MMS$\\
		& & \color{magenta} Theorem \ref{xosupperbound} \color{black} & \color{magenta} Theorem \ref{subupperbound} \color{black} &  \color{magenta} Theorem \ref{subupperbound} \color{black} \\
 		\hline
	\end{tabular}

	\begin{tablenotes}
		\item
		\item
	\end{tablenotes}
	\label{resultstable}
\end{table}

\begin{table}[t]\centering
	\caption{Results for a limited number of agents in the additive setting}
	\begin{tabular}{|l|c|c|c|c|}
		\hline
		& $n=3$ & $n=4$\\
		\hline
		\procacciafirst ~\cite{Procaccia:first}  & $3/4$-$\MMS$  & $3/4$-$\MMS$\\
		\hline
		\amanatidisapproximation ~\cite{amanatidis2015approximation}  & $7/8$-$\MMS$  & -\\
		\hline
		\color{magenta}Our result \color{black} & - & $4/5$-$\MMS$\\
		& & \color{magenta} Theorem \ref{45main} \color{black}\\
		\hline
	\end{tabular}

	\begin{tablenotes}
		\item
		\item
		\item
	\end{tablenotes}
	
	\label{resultstable2}
\end{table}

\vspace{0.2cm}
{\noindent \textbf{Theorem} \ref{addpoly} [restated]. \textit{Any fair allocation problem with additive agents admits a $3/4$-$\MMS$ allocation. Moreover, a $(3/4-\epsilon)$-$\MMS$ allocation can be found in time $\poly(n,m)$ for any $\epsilon > 0$.\\}}

The result of Theorem \ref{addpoly} is surprising, since most of the previous methods provided for proving the existence of a $2/3$-$\MMS$ allocation were tight.  
This convinced many in the community that $2/3$ is the best that can be guaranteed. This shows that the current techniques and known structural properties of maxmin share are not powerful enough to prove the bounds better than $2/3$. In this paper, we provide a better understanding of this notion by demonstrating several new properties of maxmin share. For example, we introduce a generalized form of reducibility and develop double counting techniques that are closely related to the concept of maxmin-share.

For a better understanding of our algorithm, we start with the case where valuations of the agents for all items are small enough. More precisely, let $0 < \alpha < 1$ be a constant number and assume for every agent $\agent_i$ and every item $\ite_j$, the value of agent $\agent_i$ for item $\ite_j$ is bounded by $\alpha \MMS_i$. In this case, we propose the following simple procedure to allocate the items to the agents.

\begin{itemize}
	\item Arrange the items in an arbitrary order.
	\item Start with an empty bag and add the items to the bag one by one with respect to their order.
	\item Every time the valuation of an agent $\agent_i$ for the set of items in the bag reaches $(1-\alpha)\MMS_i$, give all items of the bag to that agent, and continue with an empty bag.  In case many agents are qualified to receive the items, we choose one of them arbitrarily. From this point on, we exclude the agent who received the items from the process.
\end{itemize}
We call this procedure the $\bagfilling$ algorithm. One can see this algorithm as an extension of the famous moving knife algorithm for indivisible items.
It is not hard to show that the $\bagfilling$ algorithm guarantees a $(1-\alpha)$-$\MMS$ allocation to all of the agents. The crux of the argument is to show that every agent receives at least one bag of items. To this end, one could argue that every time a set of items is allocated to an agent $\agent_i$, no other agent $\agent_{j}$ loses a value more than $\MMS_j$. This, together with the fact that $\valu_i(\items) \geq n \MMS_i$, shows that at the end of the algorithm, every agent receives a fair share ($(1-\alpha)$-$\MMS$) of the items. 

This observation sheds light on the fact that low-value items can be distributed in a more efficient way. Therefore, the main hardness is to allocate the items with higher values to the agents. To overcome this hardness, we devise a clustering method.
Roughly speaking, we divide the agents into three clusters according to their valuation functions. We prove desirable properties for the agents of each cluster. Finally, via a procedure that is similar in spirit to the $\bagfilling$ algorithm but more complicated, we allocate the items to the agents. 

Our clustering method is based on three important principles: \textit{reducibility}, \textit{matching allocation}, and \textit{cycle-envy-freeness}. 
We give a brief description of each principle in the following.

\textbf{Reducibility:} The reducibility principle is very simple and elegant but plays an important role in the allocation process. Roughly speaking, consider a situation where for an agent $\agent_i$ and a set $S$ of items we have the following properties:
$$\valu_i(S) \geq \alpha \MMS_i$$
and $$ \forall \agent_j \neq \agent_i \qquad \MMS_j^{n-1}(\items \setminus S) \geq \MMS_j,$$ where $\valu_i(S)$ is the valuation of agent $\agent_i$ for subset $S$ of items. Intuitively, since the maxmin shares of all agents except $\agent_i$ for the all items other than set $S$ are at least as much as their current maxmin shares, allocating set $S$ to $\agent_i$ cannot hurt the guarantee. In other words, given that an $\alpha$-$\MMS$ allocation is possible for all agents except $\agent_i$ with items not in $S$, we can allocate set $S$ to agent $\agent_i$ and recursively solve the problem for the rest of the agents. Although the definition of reducibility is more general than what mentioned above, the key idea is that reducible instances of the problem can be transformed into irreducible instances. More precisely, we show that in order to prove the existence of an $\alpha$-$\MMS$ allocation, it only suffices to show this for $\alpha$-irreducible instances of the problem (see Observation \ref{reducibility}). This makes the problem substantially simpler, since $\alpha$-irreducible instances of the problem have many desirable properties. For example, in such instances, the value of every agent $\agent_i$ for each item is less than $\alpha \MMS_i$ (see Lemma \ref{remove1}). By setting $\alpha = 1/2$, this observation along with the analysis of the $\bagfilling$ algorithm, proves the existence of a $1/2$-$\MMS$ allocation. It is worth to mention that a special form of reducibility, where $|S|=1$ is used in the previouse works~\cite{amanatidis2015approximation,Procaccia:first}. 


\textbf{Matching allocation:} At the core of the clustering part, we use a well-structured type of matching to allocate the items to the agents. Intuitively, we cluster the agents to deal with high-value or in other words \textit{heavy} items. 
In order to cluster a group of agents, we find a subset $T$ of agents and a subset $S$ of items, together with a matching $M$ from $S$ to $T$. We choose $T$, $S$, and $M$ in a way that (i) every item assigned to an agent has a value of at least $\beta$ to him, (ii) agents who do not receive any items have a value less than $\beta$ for each of the assigned items. Such an allocation requires careful application of several properties of maximal matchings in bipartite graphs described in Section \ref{additive:observations}.  A matching with similar structural properties is previousely used by Procaccia and Wang \cite{Procaccia:first} to allocate the bundles to the agents. In this paper, we reveal more details and precisely characterise the structure of such matchings. We use such matchings in two main steps: selecting the agents for the first and second clusters and merging the items.



\textbf{Cycle-envy-freeness:} Envy-freeness is itself a well-known notion for fairness in the resource allocation problems. However, this notion is perhaps more applicable to the allocation of divisible goods. In our algorithm, we use a much weaker notion of envy-freeness, namely \textit{cycle-envy-freeness}. A cycle-envy-free allocation contains no cyclic permutation of agents, such that each agent envies the next agent in the cycle. In the clustering phase, we choose a matching $M$ in a way that preserves cycle-envy-freeness for the clustered agents. More details about this can be found in Section \ref{additive:clusters}. 

Cycle-envy-freeness plays a key role in the second phase of the algorithm. As aforementioned, our method in the assignment phase is closely related to the $\bagfilling$ procedure described above. The difference is that the efficiency of our method depends on the order of the agents who receive the items. Based on the notion of cycle-envy-freeness, we prioritize the agents and, as such, we show the allocation is fair. An analogous concept is previousely used in \cite{Saberi:first}, albeit with a different application than ours. 

As mentioned before, our algorithm consists of two phases: (i) clustering the agents and (ii) satisfying the agents. In the first phase, we cluster the agents into three sets namely $\cone$,$\ctwo$, and $\cthree$. In addition to this, for $\cone$ and $\ctwo$ we also have refinement procedures to make sure the rest of the unallocated items have a low value to the agents of these clusters. In the second phase, based on a method similar to the $\bagfilling$ algorithm described above, we allocate the rest of the items to the agents. A flowchart of our algorithm is depicted in Figure \ref{go}. The main steps along with brief descriptions of each step are highlighted in the flowchart. In section \ref{overview}, we present the ideas behind each of these steps and show how the entire algorithm leads to a proper allocation.

In Appendix \ref{45}, we study the case where we only have four additive agents. \procacciafirst~\cite{Procaccia:first} showed that in this case a $3/4$-$\MMS$ allocation is possible. We improve this result by giving an algorithm that finds a $4/5$-$\MMS$ allocation in this restricted setting. Note that this also leads to an algorithm that finds a $4/5-\epsilon$-$\MMS$ allocation in polynomial time.  \amanatidisapproximation~\cite{amanatidis2015approximation} also show that a $7/8$-$\MMS$ allocation is possible when the number of agents is equal to 3. These results indicate that better bounds can be achieved for the additive setting. We believe our framework can be used as a building block to obtain better bounds (see Section \ref{overview} for more details).

\begin{figure}[!htbp]
\centerline{
\includegraphics[scale=0.8]{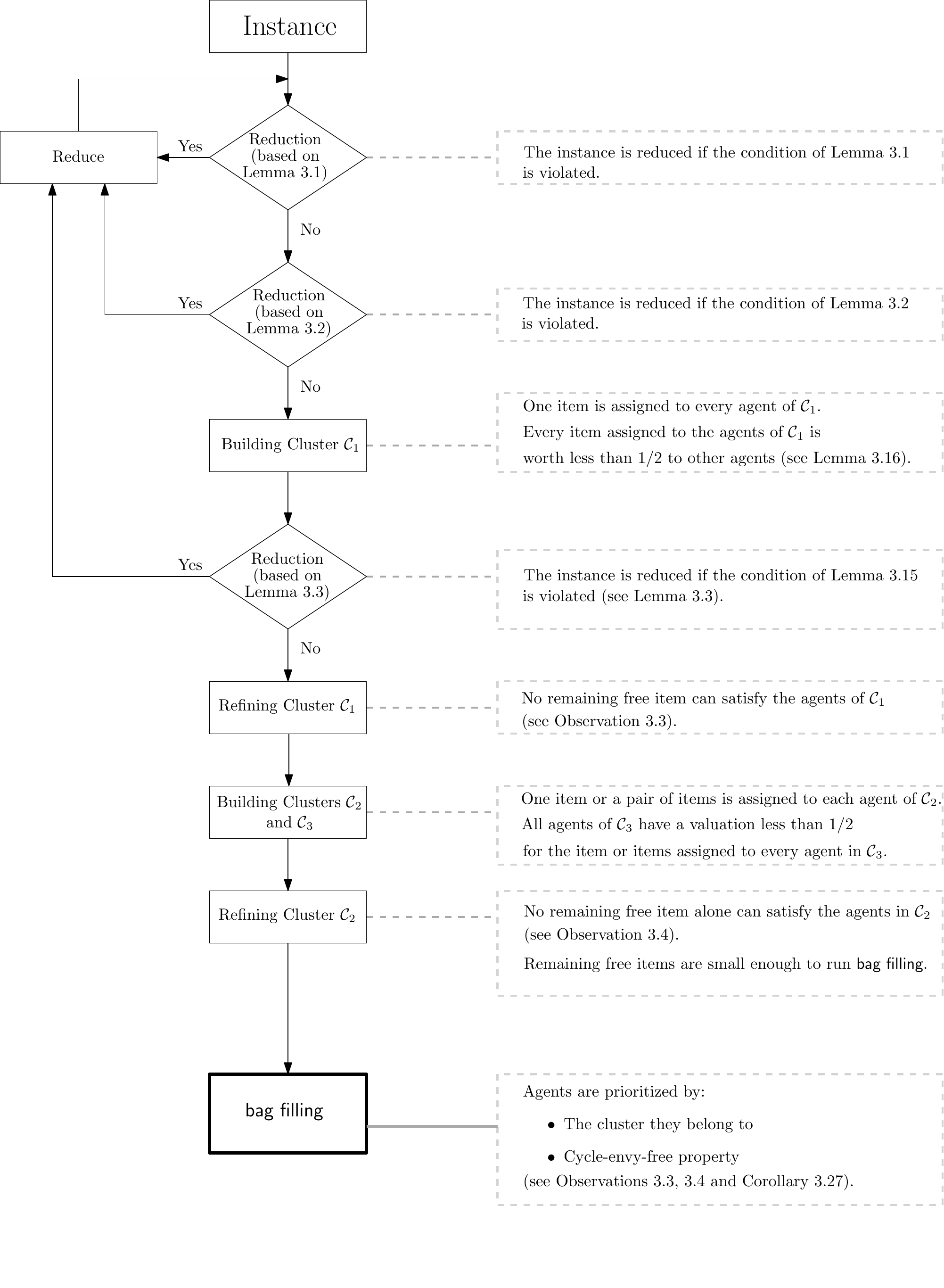}
}
\caption{A flowchart of the $3/4$-$\MMS$ allocation algorithm}
\label{go}
\end{figure} 

\subsubsection{Submodular, XOS, and Subadditive Agents}
Although the problem was initially proposed for additive agents, it is very well-motivated to extend the definition to other classes of set functions. For instance, it is quite natural to expect an agent prefers to receive two items of value 400, rather than receiving 1000 items of value 1. Such a constraint cannot be imposed in the additive setting. However, submodular functions which encompass $k$-demand valuations are strong tools for modeling these constraints. Such generalizations have been made to many similar problems, including the \textit{Santa Claus max-min fair allocation}, \textit{welfare maximization}, and \textit{secretary} problems ~\cite{bateni2013submodular,feige2009maximizing,feige2011maximizing,golovin2005max}. The most common classes of set functions that have been studied before are submodular, XOS, and subadditive functions. We consider the fair allocation problem when the agents' valuations are in each of these classes. In contrast to the additive setting in which finding a constant $\MMS$ allocation is trivial, the problem becomes much more subtle even when the agents' valuations are \textit{monotone submodular}. For instance, the $\bagfilling$ algorithm does not promise any constant approximation factor for submodular agents, while it is straight-forward to show it guarantees a $(1-\alpha)$-$\MMS$ allocation for additive agents.

We begin with submodular set functions. In Section \ref{submodular}, we show that the fair allocation problem with submodular agents admits a $1/3$-$\MMS$ allocation. In addition, we show, given access to \textit{query oracles}, one can find such an allocation in polynomial time. We further complement our result by showing that a $3/4$-$\MMS$ is the best guarantee that one can hope to achieve in this setting. This is in contrast to the additive setting for which the only upper bound is that $1$-$\MMS$ allocation is not always possible. We begin by stating an existential proof.

\vspace{0.2cm}
{\noindent \textbf{Theorem} \ref{submodulartheorem} [restated]. \textit{The fair allocation problem with submodular agents admits a $1/3$-$\MMS$ allocation. 
 \\}}

Our proof for submodular agents is fundamentally different from that of the additive setting. First, without loss of generality, we assume $\MMS_i = 1$ for every agent $\agent_i \in \agents$. Moreover, we assume the problem is $1/3$-irreducible since otherwise we can reduce the problem. Next, given a function $f(.)$, we define the \textit{ceiling function} $f^x(.)$ as follows:
$$f^x(S) = \min\{x, f(S)\} \hspace{1cm}\forall S \subseteq \domp(f).$$
An important property of the ceiling functions is that they preserve submodularity, fractionally subadditivity, and subadditivity (see Lemma \ref{ceilingfunctions}). We define the bounded welfare of an allocation $\mathcal{A}$ as $\sum \valu_i^{2/3}(A_i)$. Given that, we show an allocation that maximizes the bounded welfare is $1/3$-$\MMS$. To this end, let $\mathcal{A}$ be an allocation with the maximum bounded welfare and suppose for the sake of contradiction that in such an allocation, an agent $\agent_i$ receives a bundle of worth less than $1/3$ to him. Since $\MMS_i = 1$, agent $\agent_i$ can divide the items into $n$ sets, where each set is of worth at least $1$ to him. Now, we randomly select an element $\ite_j$ which is \textit{not} allocated to $\agent_i$. By the properties of submodular functions, we show the expected contribution of $\ite_j$ to the valuation function of $\agent_i$ is more than the expected contribution of $\ite_j$ to the bounded welfare of the allocation. Therefore, there exists an item $\ite_j$ such that if we allocate that item to agent $\agent_i$, the total bounded welfare of the allocation will be increased. This contradicts the maximality of the allocation.

Notice that Theorem \ref{submodulartheorem} is only an existential proof. A natural approach to find such a solution is to start with an arbitrary allocation and iteratively increase its bounded welfare until it becomes $1/3$-$\MMS$. The main challenge though is that we do not even know what the $\MMS$ values are. Furthermore, unlike the additive setting, we do not have any PTAS algorithm that provides us a close estimate to these values. To overcome this challenge, we propose a combinatorial trick to guess these values without incurring any additional factor to our guarantee. The high level idea is to start with large numbers as estimates to the $\MMS$ values. Every time we run the algorithm on the estimated values, it either finds a desired allocation, or reports that the maxmin value of an agent is misrepresented by at least a multiplicative factor. Given this, we divide the maxmin value of that agent by that factor and continue on with the new estimates. Therefore, at every step of the algorithm, we are guaranteed that our estimates are not less than the actual $\MMS$ values. Based on this, we show that the running time of the algorithm is polynomial, and that the allocation it finds in the end has the desired properties. The reader can find a detailed discussion in Section \ref{latter}.

\vspace{0.2cm}
{\noindent \textbf{Theorem} \ref{subsubalg} [restated]. \textit{Given access to query oracles, one can find a $1/3$-$\MMS$ allocation to submodular agents in polynomial time.\\}}

Finally, we show that in some instances with submodular agents, no allocation is better than $3/4$-$\MMS$.

\vspace{0.2cm}
{\noindent \textbf{Theorem} \ref{subupperbound} [restated]. \textit{For any integer number $c > 0$, there exists an instance of the fair allocation problem with $n \geq c$  submodular agents for which no allocation is better than $3/4$-$\MMS$.
\\}}

We show Theorem \ref{subupperbound} by a counter-example. In this counter-example we have $n$ agents and $2n$ items. Moreover, the valuation functions of the first $n-1$ agents are the same, but the last agent has a slightly different valuation function that makes it impossible to find an allocation which is better than $3/4$-$\MMS$. The number of agents in this example can be arbitrarily large.

In Section \ref{xos}, we study the problem with fractionally subadditive (XOS) agents. We first give a $1/2$ upper bound on the quality of any allocation. In other words, we show that for some instances of the problem, no allocation can guarantee anything better than $1/2$-$\MMS$ when the agents valuations are XOS. This is followed by a proof to the existence of a $1/5$-$\MMS$ allocation for any instance of the problem with XOS agents.

Similar to the submodular setting, we also provide an upper bound on the quality of any allocation in the XOS setting. We show Theorem \ref{xosupperbound} by a counter-example.

\vspace{0.2cm}
{\noindent \textbf{Theorem} \ref{xosupperbound} [restated].\textit{ For any integer number $c$, there is an instance of the fair allocation problem with XOS agents where $n \geq c$ and no allocation is better than $1/2$-$\MMS$.\\}}

Next, we state the main theorem of this section.

\vspace{0.2cm}
{\noindent \textbf{Theorem} \ref{xosproof} [restated].\textit{ The fair allocation problem with XOS agents admits a $1/5$-$\MMS$ allocation.\\}}

Our approach for proving Theorem \ref{xosproof} is similar to the proof of Theorem \ref{submodulartheorem}. Again, we scale the valuations to make sure $\MMS_i = 1$ all agents and define the notion of bounded welfare, but this time as $\sum \valu_i^{2/5}(A_i)$. However, as XOS functions do not adhere to the nice structure of submodular functions, we use a different analysis to prove this theorem. Let $\mathcal{A}$ be an allocation with the maximum bounded welfare. In case all agents receive a value of at least $1/5$, the proof is complete. Otherwise, let $\agent_i$ be an agent that receives a set of items whose value to him is less than $1/5$. In contrast to the submodular setting, giving no item alone to $\agent_i$ can guarantee an increase in the bounded welfare of the allocation. However, this time, we show there exists a set $S$ of items such that if we take them back from their recipients and instead allocate them to agent $\agent_i$, the bounded welfare of the allocation increases. The reason this holds is the following: since $\MMS_i = 1$, agent $\agent_i$ can split the items into $2n$ sets where every set is worth at least $2/5$ to $\agent_i$, otherwise the problem is $1/5$-reducible (see Lemma \ref{2nsets}). Moreover, since the valuation functions are XOS, we show that giving one of these $2n$ sets to $\agent_i$ will increase the bounded welfare of the allocation. Therefore, if $\mathcal{A}$ is maximal, then $\mathcal{A}$ is also $1/5$-$\MMS$.

Finally, we show that a $1/8$-$\MMS$ allocation in the XOS setting can be found in polynomial time. Our algorithm only requires access to demand and XOS oracles. Note that this bound is slightly worse than our existential proof due to some computational hardnesses. However, the blueprint of the algorithm is based on the proof of Theorem \ref{xosproof}.

\vspace{0.2cm}
{\noindent \textbf{Theorem} \ref{xa} [restated]. \textit{ Given access to demand and XOS oracles, we can find a $1/8$-$\MMS$ allocation for the problem with XOS agents in polynomial time.
\\}}

We start with an arbitrary allocation and increase the bounded welfare until the allocation becomes $1/8$-$\MMS$. The catch is that if the allocation is not $1/8$-$\MMS$, then there exists an agent $\agent_i$ and a set $S$ of items such that if we take back these items from their current recipients and allocate them to agent $\agent_i$, the bounded welfare of the allocation increases. In order to increase the bounded welfare, there are two computational barriers that need to be lifted. First, similar to the submodular setting, we do not have any estimates to the $\MMS$ values. Analogously, we resolve the first issue by iteratively guessing the $\MMS$ values. The second issue is that in every step of the algorithm, we have to find a set $S$ of items to allocate to an agent $\agent_i$ that results in an increase in the bounded welfare. Such a set $S$ cannot be trivially found in polynomial time. That is where the demand and XOS oracles take part. In Section \ref{former} we show how to find such a set in polynomial time. The high-level idea is the following: first, by accessing the XOS oracles, we determine the contribution of every item to the bounded welfare of the allocation. Next, we set the price of every element equal to three times the contribution of that element to the bounded welfare and run the demand oracle to find which subset has the highest profit for agent $\agent_i$. We show this subset has a value of at least $1/4$ to $\agent_i$. Next, we sort the elements of this set based on the ratio of contribution to the overall value of the set over the price of the item, and select a prefix for them that has a value of at least $1/4$ to $\agent_i$. Finally, we argue that allocating this set to $\agent_i$ increases the bounded welfare of the allocation by at least some known lower bound. This, married with the combinatorial trick to guess the $\MMS$ values, gives us a polynomial time algorithm to find a $1/8$-$\MMS$ allocation.

Note that an immediate corollary of Theorems \ref{xa} and \ref{subsubalg} is a polynomial time algorithm for approximating the maxmin value of a submodular and an XOS function within factors $1/3$ and $1/8$, respectively. 
\begin{corollary}
Let $f$ be a submodular/XOS function on a set of ground elements $S$, and let $n$ be an integer number. Given access to query oracle/demand and XOS oracles of $f$, we can partition the elements of $S$ into $n$ disjoint subsets $S_1, S_2, \ldots, S_n$ such that 
$$\min_{i=1}^n f(S_i) \geq c\cdot\MMS_f^n$$
where $\MMS_f^n$ denotes the maxmin value for function $f$ on $n$ subsets. Constant $c$ equals $1/3$ if $f$ is submodular and is equal to $1/8$ for the XOS case.
\end{corollary}

Finally, we investigate the problem when the agents are subadditive and present an existential proof based on a reduction to the XOS setting. In Section \ref{subadditive}, we present a lemma that enables us to 
reduce the problem with subadditive agents to the case where agents are XOS.
\begin{lemma}\label{rr1}
	Given a subadditive set function $f(.)$ which is defined on a set $\domp(f)$ and an integer number $n$, there exists an XOS function $g(.)$ such that  
	$$\MMS_g^n \geq \MMS_f^n / \bigg( 2 \lceil \log |\domp(f)|\rceil \bigg)$$
	and $g(S) \leq f(S)$ for every set $S \subseteq \domp(f)$.
\end{lemma}
Proof of Lemma \ref{rr1} follows from the known techniques for reducing subadditive valuations to XOS. For the sake of completeness, we bring a formal proof in Section \ref{subadditive}.

\vspace{0.2cm}
{\noindent \textbf{Theorem} \ref{subadditiveproof} [restated].\textit{ The fair allocation problem with subadditive agents admits a $1/10 \lceil\log m\rceil$-$\MMS$ allocation.\\}}



\section{Preliminaries}\label{prelim}

Throughout this paper we assume the set of agents is denoted by $\agents$ and the set of items is referred to by $\items$. Let $|\agents| = n$ and $|\items| = m$, we refer to the agents by $\agent_i$ and to the items by $\ite_i$, i.e., $\agents = \{\agent_1, \agent_2, \ldots, \agent_n\}$ and $\items = \{\ite_1, \ite_2, \ldots,\ite_m\}$. We denote the valuation of agent $\agent_i$ for a set $S$ of items by $\valu_i(S)$. Our interest is in valuation functions that are monotone and non-negative. More precisely, we assume $\valu_i(S) \geq 0$ for every agent $a_i$ and set $S \subseteq \items$, and for every two sets $S_1$ and $S_2$ we have
$$\forall \agent_i \in \agents\hspace{2cm} \valu_i(S_1 \cup S_2) \geq \max\{\valu_i(S_1), \valu_i(S_2)\}.$$

Due to obvious impossibility results for the general valuation functions\footnote{If the valuation functions are not restricted, no approximation guarantee can be achieved. For instance consider the case where we have two agents and 4 items. Agent $\agent_1$ has value 1 for sets $\{\ite_1,\ite_2\}$ and $\{\ite_3,\ite_4\}$ and 0 for the rest of the sets. Similarly, agent $\agent_2$ has value 1 for sets $\{\ite_1,\ite_3\}$ and $\{\ite_2,\ite_4\}$ and 0 for the rest of the sets. In this case, no allocation can provide both of the agents with sets which are of non-zero value to them.}, we restrict our attention to four classes of set functions:
\begin{itemize}
	\item \textbf{Additive}: A set function $V(.)$ is additive if $V(S_1) + V(S_2) = V(S_1 \cup S_2) - V(S_1 \cap S_2)$ for every two sets $S_1,S_2 \in \domp(V)$.
	\item \textbf{Submodular}: A set function $V(.)$ is submodular if $V(S_1) + V(S_2) \geq V(S_1 \cup S_2) - V(S_1 \cap S_2)$ for every two sets $S_1,S_2 \in \domp(V)$.
	\item \textbf{Fractionally Subadditive (XOS)}: An XOS set function $V(.)$ can be shown via a finite set of additive functions $\{V_1, V_2, \ldots, V_{\alpha}\}$ where $V(S) = \max_{i=1}^{\alpha} V_i(S)$ for any set $S \subseteq \domp(V)$. 
	\item \textbf{Subadditive}: A set function $V(.)$ is subadditive if $V(S_1) + V(S_2) \geq V(S_1 \cup S_2)$ for every two sets $S_1,S_2 \subseteq \domp(V)$. 
\end{itemize}

For additive functions, we assume the value of the function for every element is given in the input. However, representing other classes of set functions requires access to oracles. For submodular functions, we assume we have access to \textit{query oracle} defined below. Query oracles are great identifier for submodular functions, however, they are too weak when it comes to XOS and subadditive settings. For such functions, we use a stronger oracle which is called \textit{demand oracle}. It is shown that for some functions, such as gross substitutes, a demand oracle can be implemented via a query oracle in polynomial time~\cite{leme2014gross}. In addition to this, we consider a special oracle for XOS functions which is called \textit{XOS oracle}. Access to query oracles for submodular functions, XOS oracle for XOS functions, and demand oracles for XOS and subadditive functions are quite common and have been very fruitful in the literature~\cite{dobzinski2005approximation,feige2009maximizing,feige2011maximizing,feige2006approximation,feldmancombinatorial,leme2014gross,vondrak2008optimal}. In what follows, we formally define the oracles:
\begin{itemize}
	\item \textbf{Query oracle:} Given a function $f$, a query oracle $\mathcal{O}$ is an algorithm that receives a set $S$ as input and computes $f(S)$ in time $O(1)$.
	\item \textbf{Demand oracle:} Given a function $f$, a demand oracle $\mathcal{O}$ is an algorithm that receives a sequence of prices $p_1, p_2, \ldots, p_n$ as input and finds a set $S$ such that
	$$f(S) - \sum_{e \in S} p_e$$ is maximized. We assume the running time of the algorithm is $O(1)$.  
	\item \textbf{XOS oracle:} (defined only for an XOS functions $f$) Given a set $S$ of items, it returns the additive representation of the function that is maximized for $S$. In other words, it reveals the contribution of each item in $S$ to the value of $f(S)$. 
\end{itemize}

Let $\Pi_r$ be the set of all partitionings of $\items$ into $r$ disjoint subsets. For every $r$-partitioning $P^* \in \Pi_r$, we denote the partitions by $P^*_1,P^*_2,\ldots,P^*_r$. For a set function $f(.)$, we define $\MMS_f^r(\items)$ as follows:
$$ \MMS_f^r(\items) = \max_{P^* \in {\Pi_r}}  \min_{1 \leq j \leq r} f(P^*_j).$$
For brevity we refer to $\MMS_{f_i}^n(\items)$ by $\MMS_{i}$.


An allocation of items to the agents is a vector $\mathcal{A} = \langle A_1, A_2, \ldots, A_n\rangle$ where $\bigcup A_i = \items$ and $A_i \cap A_j = \emptyset$ for every two agents $\agent_i, \agent_j \in \agents$. An allocation $\mathcal{A}$ is  $\alpha$-$\MMS$, if every agent $a_i$ receives a subset of the items whose value to that agent is at least $\alpha$ times $\MMS_i$. More precisely, $\mathcal{A}$ is $\alpha$-$\MMS$ if and only if
$\valu_i(A_i) \geq \alpha\MMS_i$
for every agent $\agent_i \in \agents$.

We define the notion of \textit{reducibility} for an instance of the problem as follows.

\begin{definition}\label{d1}
	We say an instance of the problem is $\alpha$-reducible, if there exist a set $T \subset \agents$ of agents, a set $S$ of items, and an allocation $\mathcal{A} = \langle A_1, A_2, \ldots, A_n\rangle$ of $S$ to agents of $T$ such that 
	$$\forall \agent_i \in T \hspace{3cm} V_i(A_i) \geq \alpha \MMS_i$$
	and
	$$\forall \agent_i \notin T \hspace{1cm} \MMS_{V_i}^{n-|T|} (\items \setminus S)\geq \MMS_i.$$
\end{definition}
Similarly, we call an instance $\alpha$-\textit{irreducible} if it is not $\alpha$-reducible. The intuition behind Definition \ref{d1} is the following: In order to prove the existence of an $\alpha$-$\MMS$ allocation for every instance of the problem, it only suffices to prove this for the $\alpha$-irreducible instances.

\begin{observation}\label{reducibility}
	Every instance of the fair allocation problem admits an $\alpha$-$\MMS$ allocation if this holds for all $\alpha$-irreducible instances. 
\end{observation}
\begin{proof}
	Suppose for the sake of contradiction that all $\alpha$-irreducible instances of the problem admit an $\alpha$-$\MMS$ allocation, but there exists an $\alpha$-reducible instance of the problem which does not admit any $\alpha$-$\MMS$ allocation. Among all such instances, we consider the one with the lowest number of agents. Since this instance is $\alpha$-reducible, there exists a subset $T$ of agents and a subset $S$ of items such that an allocation of $S$ to agents of $T$ provides each of them with a valuation of at least $\alpha \MMS_i$. Moreover, the rest of the items and agents make an instance of the problem with a smaller $n$. Thus, an $\alpha$-$\MMS$ allocation can satisfy the rest of the agents and hence the instance admits an $\alpha$-$\MMS$ allocation. This contradicts the assumption. 
\end{proof}

The reducibility argument plays an important role in both the existential proofs and algorithms that we present in the paper. As we see in Sections \ref{additive:observations} and \ref{xos}, irreducible instances of the problem exhibit several desirable properties for additive and non-additive agents. We take advantage of these properties to improve the approximation guarantee of the problem for different classes of set functions.

\section{Additive Agents\protect\footnote{We have created a website at \href{https://www.cs.umd.edu/\~saeedrez/fair.html}{https://www.cs.umd.edu/$\sim$saeedrez/fair.html} for the implemented algorithm and all related materials.}}}\label{additive}
In this section we study the fair allocation problem in the additive setting. We present a proof to the existence of a $3/4$-$\MMS$ allocation when the agents are additive. This improves upon the work of \procacciafirst ~\cite{Procaccia:first} wherein the authors prove a $2/3$-$\MMS$ allocation exists for any combination of additive agents.  As we show, our proof is constructive; given an algorithm that determines the $\MMS$ of an additive set function within a factor $\alpha$, we can implement an algorithm that finds a $3/(4 \alpha)$-$\MMS$ allocation in polynomial time. This married with the PTAS algorithm of \epsteinefficient ~\cite{epstein2014efficient} for finding the $\MMS$ values, results is an algorithm that finds a $3/(4+\epsilon)$-$\MMS$ allocation in polynomial time.

The main idea behind the $3/4$-$\MMS$ allocation is \textit{clustering} the agents. Roughly speaking, we categorize the agents into three clusters, namely $\cone$, $\ctwo$, and $\cthree$. We show that the valuation functions of the agents within each cluster show similar behaviors. Along the clustering process, we allocate the heavy items (the items that have a valuation of at least $1/4$ to some agents) to the agents. By Observation \ref{reducibility}, proving a $3/4$-$\MMS$ guarantee can be narrowed down to only $3/4$-irreducible instances. The $3/4$-irreducibility of the problem guarantees that after the clustering process, the remaining items are light. This enables us to run a $\bagfilling$ process to satisfy the agents. In order to prove the correctness of the algorithm, we take advantage of the properties of each cluster separately.

The organization of this section is summarized in the following: we start by a brief and abstract explanation of the ideas in Section \ref{overview}. In Section \ref{additive:observations} we study the properties of the additive setting and state the main observations that later imply the correctness of our algorithm. Next, in Section \ref{additive:clusters} we discuss a method for clustering the agents and in Section \ref{additive:allocation} we show how we allocate the items to the agents of each cluster to ensure a $3/4$-$\MMS$ guarantee. Finally, in Section \ref{additive:algorithm} we explain the implementation details and prove a polynomial running time for the proposed algorithm.

Throughout this section, we assume $\MMS_i = 1$ for all agents $\agent_i \in \agents$. This is without loss of generality for the existential proof since one can scale the valuation functions to impose this constraint. However, the computational complexity of the allocation will be affected by this assumption since determining the $\MMS$ of an additive function is NP-hard~\cite{epstein2014efficient}. That said, we show in Section \ref{additive:algorithm} that this challenge can be overcome by incurring an additional $1+\epsilon$ factor to the approximation guarantee.  

For brevity, we defer the proofs of Sections \ref{additive:observations}, \ref{additive:clusters}, \ref{additive:allocation}, and \ref{additiveproofs} to Appendices \ref{additiveobservationsproof},\ref{clusteringappendix},\ref{clustering2appendix}, and \ref{additiveproofappendix}, respectively.
\subsection{A Brief Overview of the Algorithm}\label{overview}
The purpose of this section is to present an abstract overview over the ideas behind our algorithm for finding a $3/4$-$\MMS$ allocation in the additive setting. For simplicity, we start with a simple $1/2$-$\MMS$ algorithm mentioned in Section \ref{results}. Recall that the $\bagfilling$ procedure guarantees a $1-\alpha$ approximation solution when the valuations of the agents for each item is smaller than $\alpha$. Furthermore, we know that in every $\alpha$-irreducible instance, all the agents have a value less than $\alpha$ for every items. Thus, the following simple procedure yields a $1/2$-$\MMS$ allocation:
\begin{enumerate}
\item Reduce the problem until no agent has a value more than $1/2$ for any item.
\item Allocate the items to the agents via a $\bagfilling$ procedure.
\end{enumerate} 

Figure \ref{fig:mms12} shows a flowchart for this algorithm. 
\begin{figure}[h]
\centerline{
\includegraphics[scale=0.4]{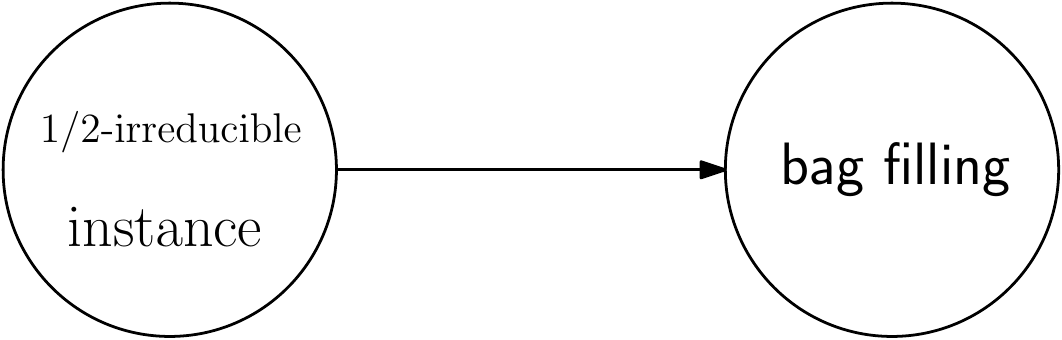}
}
\caption{$1/2$-$\MMS$ Algorithm}
\label{fig:mms12}
\end{figure}

We can extend the idea in $1/2$-$\MMS$ algorithm to obtain a more efficient algorithm. Here is the sketch of the $2/3$-$\MMS$ algorithm:
consider a $2/3$-irreducible instance of the problem. In this instance, we have no item with a value more than or equal to $2/3$ to any agent. Nevertheless, the items are not yet small enough to run a $\bagfilling$ procedure. The idea here is to divide the agents into two clusters $\cone$ and $\ctwo$. Along this clustering, the items with a value  in range $[1/3,2/3)$ are given to the agents. In particular, one item is allocated to every agent in $\cone$ that is worth at least $1/3$ to him. 
Next, we refine Cluster $\cone$. In the refining procedure, if any remaining item could singly satisfy an agent in $\cone$, we do so. After building $\cone$ and $\ctwo$ and refining $\cone$, the remaining items preserve the following two invariants:
\begin{enumerate}
\item Value of every remaining item is less than $1/3$ to every remaining agent.
\item  No remaining item can singly satisfy an agent in $\cone$ (regarding the item that is already allocated to them)
\end{enumerate} 
These two invariants enable us to run a $\bagfilling$ procedure over the remaining items. For this case, the $\bagfilling$ procedure must be more intelligent: in the case that multiple agents are qualified to receive the items of the bag, we prioritize the agents. Roughly speaking, the priorities are determined by two factors: the cluster they belong to, and the cycle-envy-freeness property of the agents in $\cone$. In Figure \ref{23mms} you can see a flowchart for this algorithm.

\begin{figure}[h]
\centerline{
\includegraphics[scale=0.4]{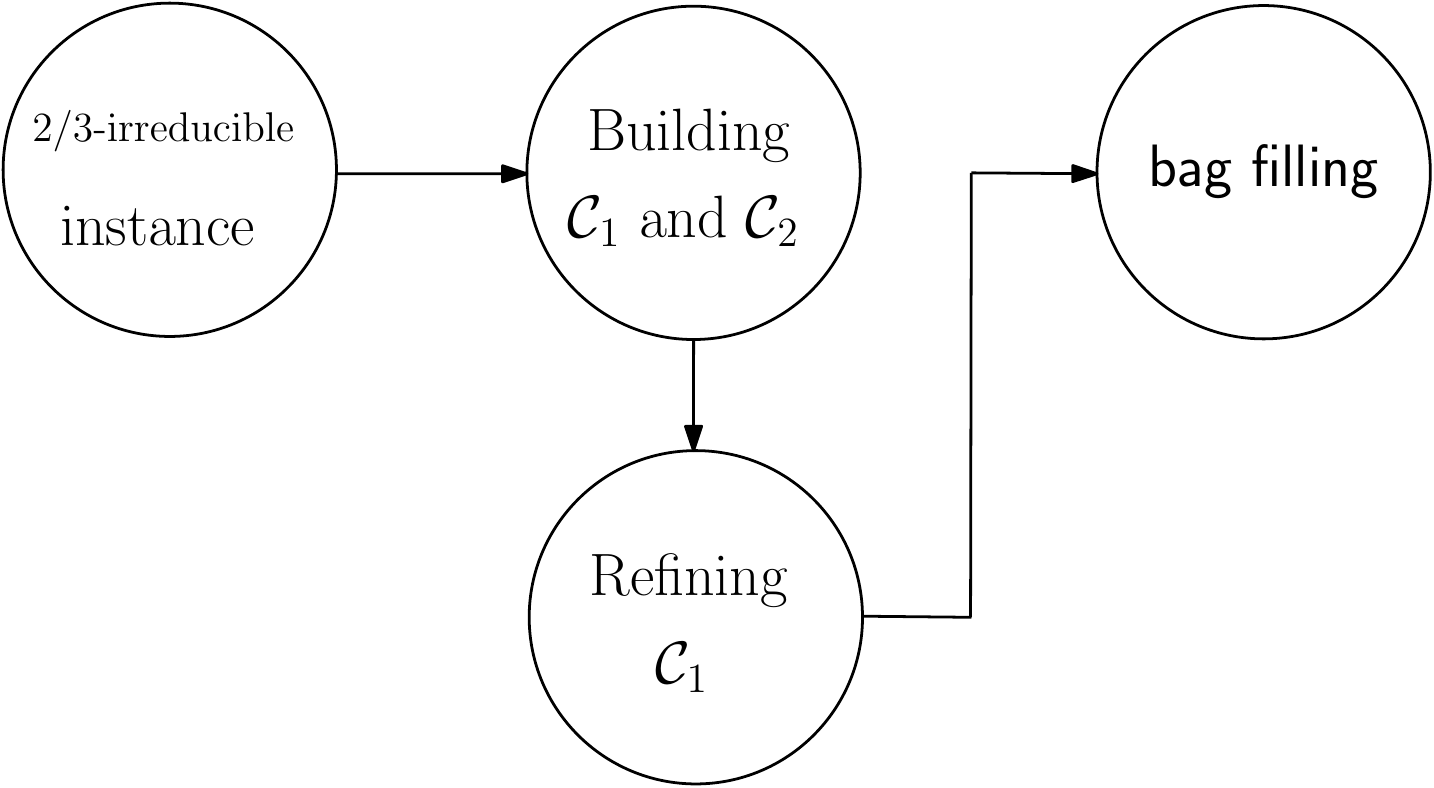}
}
\caption{$2/3$-$\MMS$ Algorithm}
\label{23mms}
\end{figure}     

Our method for a $3/4$-$\MMS$ allocation takes one step further from the previous $2/3$-$\MMS$ algorithm.  Again, we assume that the input is $3/4$-Irreducible since otherwise it can be further simplified.  Via similar ideas, we build Cluster $\cone$ and refine it. Next, we build Clusters $\ctwo$ and $\cthree$ and refine $\ctwo$. After refining Cluster $\ctwo$, the following invariants are preserved for the remaining items:
\begin{enumerate}
\item Almost every remaining item has a value less than $1/4$ to every remaining agent. More precisely, for every remaining agent $\agent_i$, there is at most one remaining item $\ite_j$ with $\valu_i(\{\ite_j\}) \geq 1/4$.
\item  No remaining item can singly satisfy an agent in $\cone$ and $\ctwo$ (regarding the item that is already allocated to them).
\end{enumerate}

Finally, we run a $\bagfilling$ procedure. Again, in the $\bagfilling$ procedure, the priorities of the agents are determined by the cluster they belong to, and the cycle-envy-freeness of the clusters. In Figure \ref{34mms} you can see the flowchart of the algorithm.

\begin{figure}[h]
\centerline{
\includegraphics[scale=0.4]{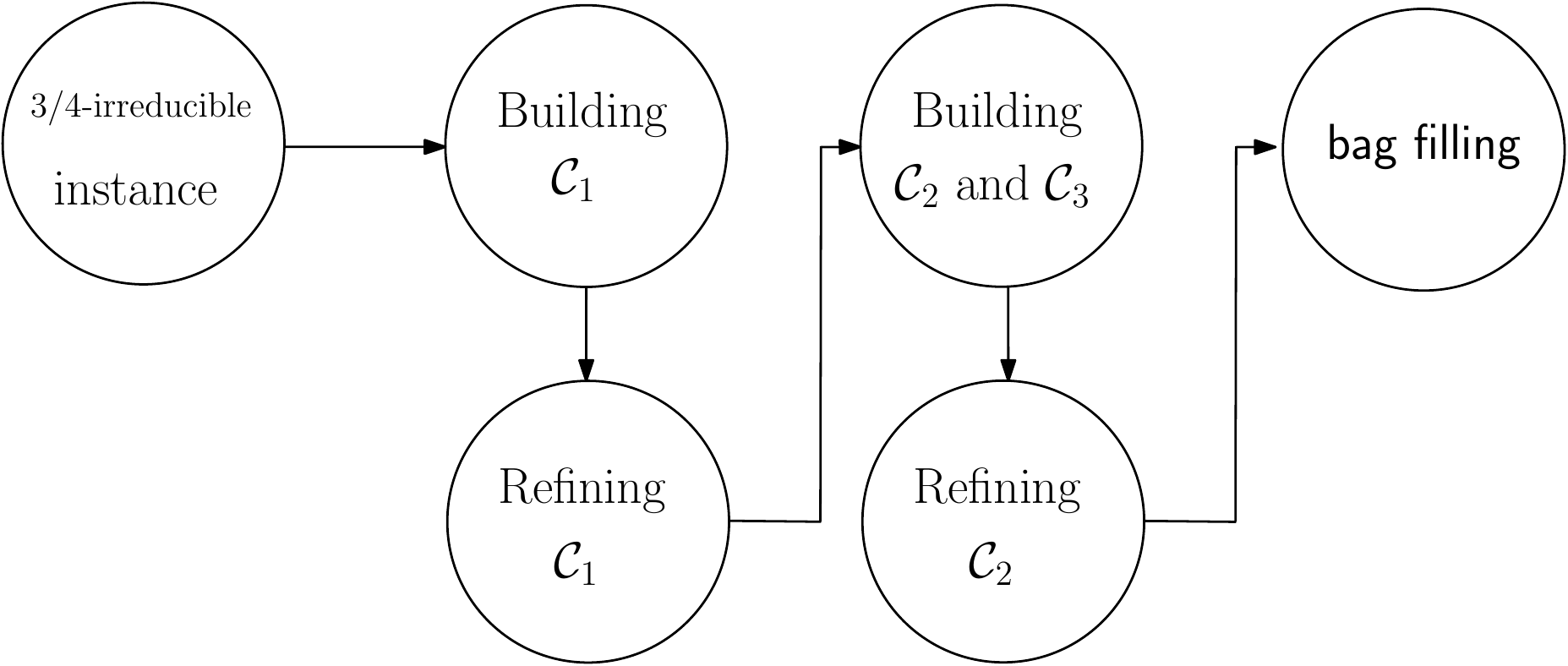}
}
\caption{$3/4$-$\MMS$ Algorithm}
\label{34mms}
\end{figure}     

Our assumption is that the input is $3/4$-irreducible. Hence, we describe our algorithm in two phases: a clustering phase and the $\bagfilling$ phase, as shown in Figure \ref{34mms2}. 
\begin{figure}[h]
\centerline{
\includegraphics[scale=0.4]{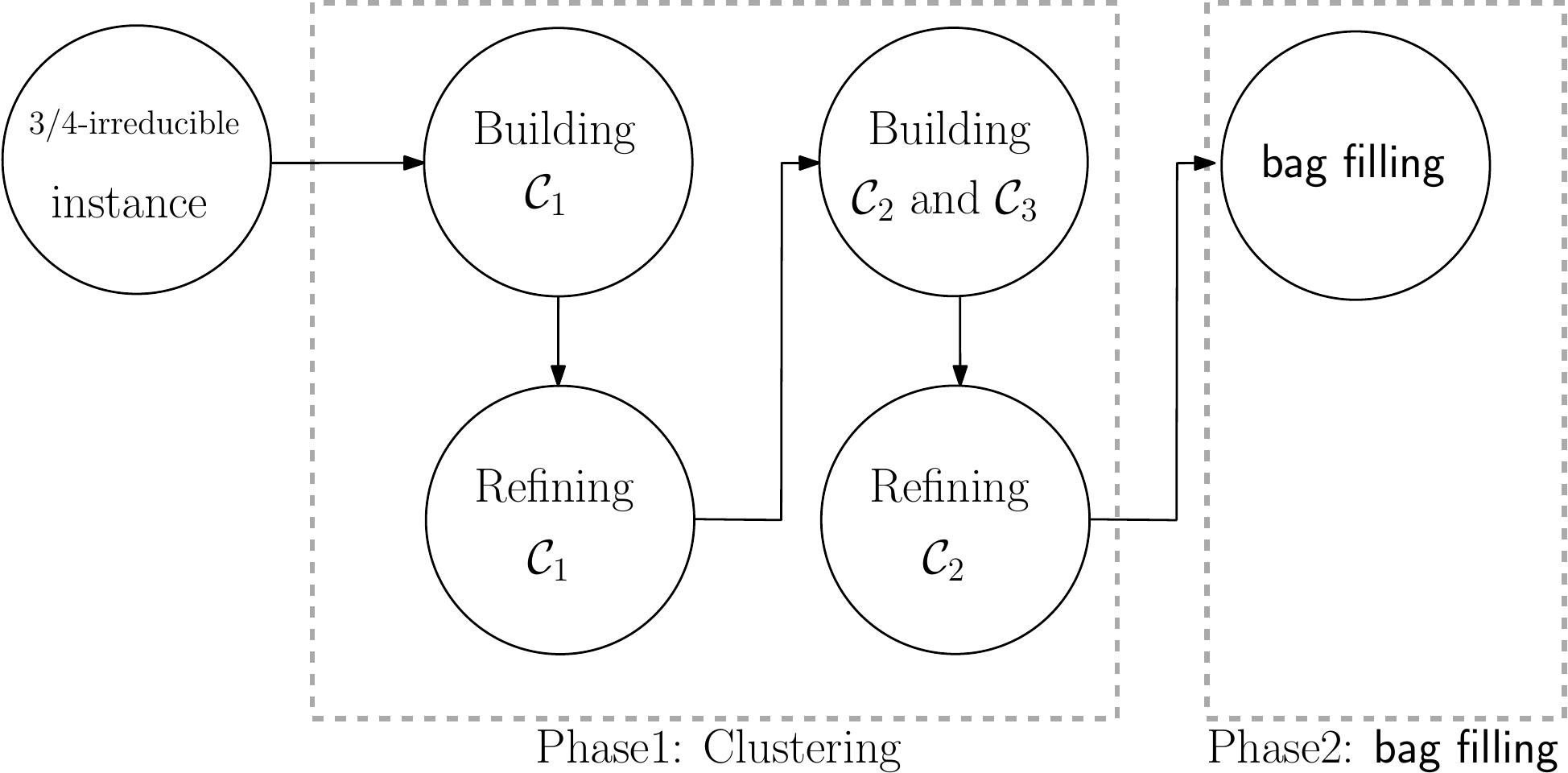}
}
\caption{Algorithm Phases}
\label{34mms2}
\end{figure}     
In Section \ref{additive:algorithm} we show that all the steps of the algorithm can be implemented in polynomial time. Furthermore, we show that the assumption that the input is $3/4$-irreducible is without loss of generality. In fact, in Section \ref{additive:algorithm} we show that it suffices to check some invariants of irreducibility to be held in certain points of the algorithm. In Figure \ref{go}, these steps are specified with caption \emph{Reduction}. 

As a future work, one can consider a more generalized form of this algorithm, where the agents are divided into more than $3$ clusters (see Figure \ref{epsmms}). We believe that this generalization might yield a $(1-\epsilon)$-$\MMS$ allocation, where $\epsilon$ is a small value that depends on the number of agents. However, such a generalization is faced with two main barriers. First, In order to extend the idea to more than 3 clusters, we need a generalized form of Lemmas \ref{remove2} and \ref{remove3} for more than two items. Furthermore, a challenging part of our approximation proof is to show that the second cluster is empty at the end of the algorithm. For this, we define a graph on the items in the second cluster and prove some bounds on the number of edges in this graph. To extend the idea for more clusters, we need to define hypergraphs on the items in the clusters and show similar bounds, which requires deeper and more complicated techniques..   
\begin{figure}[h]
\centerline{
\includegraphics[scale=0.55]{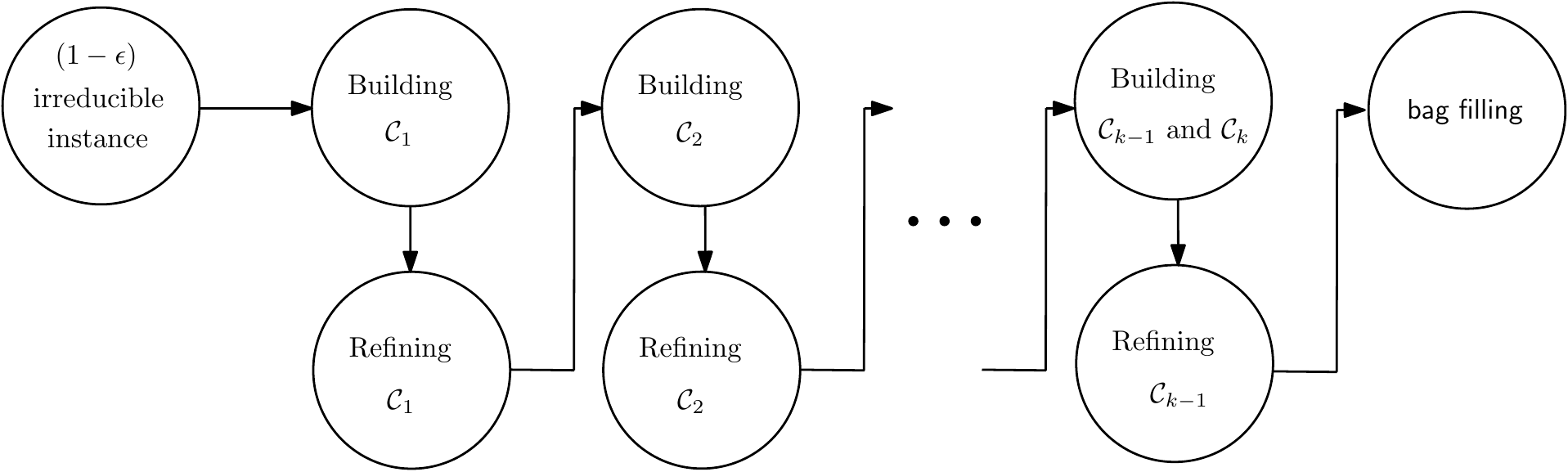}
}
\caption{Generalizing the algorithm into $k$ clusters}
\label{epsmms}
\end{figure}  

Before presenting the algorithm, in Section \ref{additive:observations} we discuss the consequences of irreducibility and  techniques to build the clusters and preserving cycle-envy-freeness in each cluster. Next, we describe the algorithm in more details. 

\subsection{General Definitions and Observations}\label{additive:observations}
Throughout this section we explore the properties of the fair allocation problem with additive agents.
\subsubsection{Consequences of Irreducibility}
 Since the objective is to prove the existence of a $3/4$-$\MMS$ allocation, by Observation \ref{reducibility}, it only suffices to show every $3/4$-irreducible instance of the problem admits a $3/4$-$\MMS$ allocation. Therefore, in this section we provide several properties of the $3/4$-irreducible instances. We say a set $S$ of items \textit{satisfies} an agent $\agent_i$ if and only if $\valu_i(S) \geq 3/4$. Perhaps the most important consequence of irreducibility is a bound on the valuation of the agents for every item. In the following we show if the problem is $3/4$-irreducible, then no agent has a value of $3/4$ or more for an item. 

\begin{lemma}\label{remove1} 
For every $\alpha$-irreducible instance of the problem we have 
$$\forall \agent_i \in \agents, \ite_j \in \items \hspace{1cm} \valu_i(\ite_j) < \alpha.$$
\end{lemma}

In other words, Lemma \ref{remove1} states that in a $3/4$-irreducible instance of the problem, no item alone can satisfy an agent. 

It is worth mentioning that the proof for Lemma \ref{remove1} does not rely on additivity of the valuation functions and holds as long as the valuations are monotone. Thus, regardless of the type of the valuation functions, one can assume that in any $\alpha$-irreducible instance, value of any item is less than $\alpha$ for any agent. Hence the statement carries over to the submodular, XOS, and subadditive settings.

As a natural generalization of Lemma \ref{remove1}, we show a similar observation for every pair of items. However, this involves an additional constraint on the valuation of the other agents for the pertinent items. In contrast to Lemma \ref{remove1}, Lemmas \ref{remove2} and \ref{remove3} are restricted to additive setting and their results  do not hold in more general settings.

\begin{lemma}
\label{remove2}
If the problem is $3/4$-irreducible and
$\valu_i(\{\ite_j,\ite_k\}) \geq 3/4$
holds for an agent $\agent_i \in \agents$ and items $\ite_j, \ite_k \in \items$, then there exists an agent $\agent_{i'} \neq \agent_i$ such that
$$\valu_{i'}(\{\ite_j,\ite_k\}) > 1$$
\end{lemma}

According to Lemma \ref{remove2}, in every $3/4$-irreducible instance of the problem, for every agent $\agent_i$ and items $\ite_j,\ite_k$, either $\valu_i(\{\ite_j,\ite_k\}) < 3/4$ or there exists another agent $\agent_{i'} \neq \agent_i$, such that $\valu_{i'}(\{\ite_j,\ite_k\}) > 1$. Otherwise, we can reduce the problem and find a $3/4$-$\MMS$ allocation recursively. More generally, let $S = \{\ite_{j_1},\ite_{j_2},\ldots,\ite_{j_{|S|}}\}$ be a set of items in $\cal{M}$ and $T =\{\agent_{i_1},\agent_{i_2},\ldots,\agent_{i_{|T|}}\}$ be a set of agents such that
\begin{description}
 \item (i) $|S| = 2|T|$
 \item (ii) For every $\agent_{i_a} \in T$ we have $\valu_{i_a}(\{\ite_{j_{2a-1}},\ite_{j_{2a}}\}) \geq 3/4$.
 \item (iii) For every $\agent_{i} \notin T$ we have $\valu_{i}(\{\ite_{j_{2a-1}},\ite_{j_{2a}}\}) \leq 1$ for every $1 \leq a \leq |T|$.
\end{description}
then the problem is $3/4$-reducible.
\begin{lemma}\label{remove3}
In every $3/4$-irreducible instance of the problem, for every set $T =\{\agent_{i_1},\agent_{i_2},\ldots,\agent_{i_{|T|}}\}$ of agents and set $S = \{\ite_{j_1},\ite_{j_2},\ldots,\ite_{j_{|S|}}\}$ of items at least one of the above conditions is violated.
\end{lemma}

\subsubsection{Modeling the Problem with Bipartite Graphs}\label{MtPwBG}
In our algorithm we subsequently make use of classic algorithms for bipartite graphs. Let $G = \langle V(G),E(G)\rangle$ be a graph representing the agents and the items. Moreover, let $V(G) = \itemsv \cup \agentsv$ where $\agentsv$ corresponds to the agents and $\itemsv$ corresponds to the items. More precisely, for every agent $\agent_i$ we have a vertex $\agentv_i \in \agentsv$ and every item $\ite_j$ corresponds to a vertex $\itemv_j \in \itemsv$. For every pair of vertices $\agentv_i \in \agentsv$ and $\itemv_j \in \itemsv$, there exists an edge $(\itemv_j,\agentv_i) \in E(G)$ with weight $w(\itemv_j,\agentv_i) = \valu_i(\{\ite_j\})$. We refer to this graph as \emph{the value graph}.

We define an operation on the weighted graphs which we call \textit{filtering}. Roughly speaking, a filtering is an operation that receives a weighted graph as input and removes all of the edges with weight less than a threshold from the graph. Next, we remove all of the isolated\footnote{A vertex is called isolated if no edge is incident to that vertex.} vertices and report the remaining as the filtered graph. In the following we formally define the notion of filtering for weighted graphs.
\begin{definition}
A $\beta$-filtering of a weighted graph $H\langle V(H),E(H)\rangle$, denoted by $H_{\beta}\langle V_\beta(H),E_\beta(H)\rangle$, is a subgraph of $H$ where $V_\beta(H)$ is the set of all vertices in $V(H)$ incident to at least one edge of weight $\beta$ or more and 
$$E_\beta(H) = \{(u,v) \in E(H)| w(u,v) \geq \beta\}.$$ 
\end{definition}
For the case of the value graph, we also denote by $\agentsv_\beta$ and $\itemsv_\beta$ the sets of agents and items corresponding to vertices of $V_\beta(G)$.
\begin{figure}[t!]
    \centering
    \includegraphics[scale=0.8]{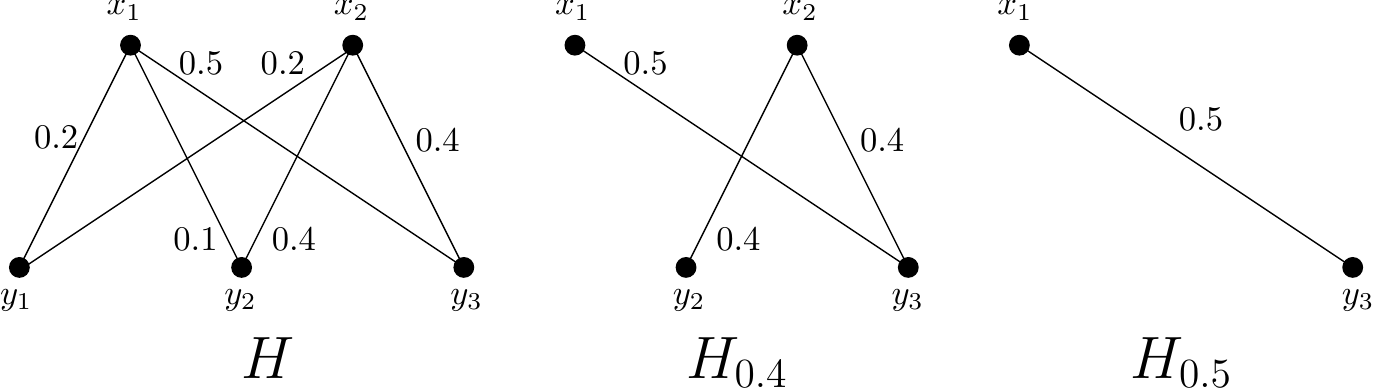}
    \caption{An example of $\beta$-filtering on a graph. After removing the edges with a value smaller than $\beta$, some vertices may become isolated. All such vertices are removed from the filtered graph.}
    \label{fig:filtering}
\end{figure}
Figure \ref{fig:filtering} illustrates an example of a graph $H$, together with $H_{0.4}$ and $H_{0.5}$. Note that none of the vertices in $H_{0.4}$ or $H_{0.5}$ are isolated. 
%

Denote by a maximum matching, a matching that has the highest number of edges in a graph. In definition \ref{FG}, we introduce our main tool for clustering the agents. 
\begin{definition}
\label{FG}
Let $H\langle V(H),E(H)\rangle$ be a bipartite graph with $V(H) = \partone \cup \parttwo$ and let $M$ be a maximum matching of $H$. Define $\parttwo_1$ as the set of the vertices in $\parttwo$ that are not saturated by $M$. Also, define $\parttwo_2$ as the set of vertices in $\parttwo$ that are connected to $\parttwo_1$ by an alternating path and let $\partone_2 = M(\parttwo_2)$, where  $M(\parttwo_2)$ is the set of vertices in $\partone$ that are matched with the vertices of $\parttwo_2$ in $M$. We define $F_{H}(M,\partone)$ as the set of the vertices in $\partone \setminus \partone_2$. 
\end{definition}

For a better understanding of Definition \ref{FG}, consider Figure \ref{fig:FG}. By the definition of alternating paths, there is no edge between the saturated vertices of $F_H(M,\partone)$ and $ \parttwo_1 \cup \parttwo_2$. On the other hand, since $M$ is maximum, the graph doesn't have any augmenting path. Thus, there is no edge between unsaturated vertices in $F_H(M,\partone)$ and $ \parttwo_1 \cup \parttwo_2$. As a result, there is no edge between $F_H(M,\partone)$ and $ \parttwo_1 \cup \parttwo_2$. Furthermore, $F_H(M,\partone)$ has another important property: there exists a matching from $N(F_H(M,\partone))$ to $F_H(M,\partone)$, that saturates all the vertices in $N(F_H(M,\partone))$, where $N(F_H(M,\partone))$ is the set of neighbors of  $F_H(M,\partone)$.

\begin{figure}
\centering
\includegraphics[scale=0.8]{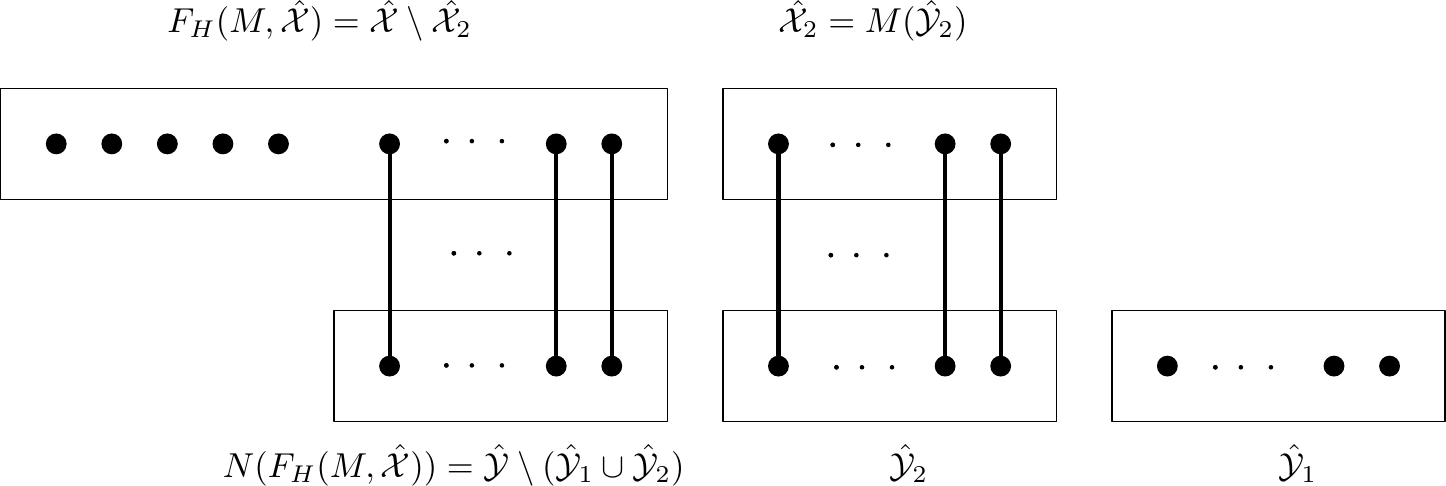}
\caption{Definition of $F_H$}
\label{fig:FG}
\end{figure}

In Lemmas \ref{iff} and \ref{rem}, we prove two remarkable properties for bipartite graphs. As a consequence of these two lemmas, Corollary \ref{remcol} holds for every bipartite graph. We leverage the result of Corollary \ref{remcol} in the clustering phase.
 
\begin{lemma}
\label{rem}
Let $H(V,E)$ be a bipartite graph with $V = \partone \cup \parttwo$ and let $M$ be a maximum matching of $H$. Then, for every set $T \subseteq \partone \setminus F_H(M,\partone)$ we have $|N(T)| > |T|$, where $N(T)$ is the set of neighbors of $T$. 
\end{lemma}

\begin{lemma}
\label{iff}
For a bipartite graph $H(V,E)$ with $V = \partone \cup \parttwo$, $F_H(M,\partone) = \emptyset$ holds, if and only if
for all $T \subseteq \partone$ we have  $|N(T)| > |T|$, where $N(T)$ is the set of neighbors of $T$.

\end{lemma}

\begin{corollary}[of Lemmas \ref{iff} and \ref{rem}]
\label{remcol}
Let $H(V,E)$ be a bipartite graph with $V = \partone \cup \parttwo$ and let $M$ be a maximum matching of $H$. 
Furthermore, let  $H'(V',E')$ be the induced sub-graph of $H$, with $V' = \partone' \cup \parttwo'$, where $\partone' = \partone \setminus F_H(M,\partone)$ and $\parttwo' = \parttwo \setminus N(F_H(M,\partone))$. Then, for any maximum matching $M'$ of $H'$, $F_{H'}(M',\partone') = \emptyset$ holds. 
\end{corollary}
\subsubsection{Cycle-envy-freeness and $\MCMWM$}\label{additive:cef}

In the algorithm, we satisfy each agent in two steps. More precisely, we allocate each agent two sets of items that are together of worth at least $3/4$ to him. We denote the first set of items allocated to agent $\agent_i$ by $\firstset_i$ and the second set by $\secondset_i$. Moreover, we attribute the agents with labels \textit{satisfied}, \textit{unsatisfied}, and \textit{semi-satisfied} in the following way:
\begin{enumerate}
	\item An agent $\agent_i$ is satisfied if $\valu_i(\firstset_i \cup \secondset_i) \geq 3/4$.
	\item An agent $\agent_i$ is semi-satisfied if $\firstset_i \neq \emptyset$ but $\secondset_i = \emptyset$. In this case we define $\epsilon_i = 3/4-\valu_i(\firstset_i)$.
	\item An agent $\agent_i$ is unsatisfied if $\firstset_i = \secondset_i = \emptyset$.
\end{enumerate}
As we see, the algorithm maintains the property that for every semi-satisfied agent $\agent_i$, $\valu_i(\firstset_i) \geq 1/2$ holds and hence, $\epsilon_i < 1/4$. 

To capture the competition between different agents, we define an attribution for an ordered pair of agents. We say a semi-satisfied agent envies another semi-satisfied agent, if he prefers to switch sets with the other agent. 

\begin{definition}
\label{winloose}
Let $T$ be a set of semi-satisfied agents. An agent $\agent_i \in T$ envies an agent $\agent_j \in T$, if $\valu_i(\firstset_j) \geq \valu_i(\firstset_i)$. Also, we call an agent $ \agent_i \in T$ a winner of $T$, if $\agent_i$ envies no other agent in $T$. Similarly, we call an agent $\agent_i$ a loser of $T$, if no other agent in $T$ envies $\agent_i$.
\end{definition}

Note that it could be the case that an agent $\agent_i$ is both a loser and a winner of a set $T$ of agents. Based on Definition \ref{winloose}, we next define the notion of \textit{cycle-envy-freeness}.

\begin{definition}
 We call a set $T$ of semi-satisfied agents cycle-envy-free, if every non-empty subset of $T$ contains at least one winner and one loser. 
\end{definition}

Let $C$ be a cycle-envy-free set of semi-satisfied agents. Define the representation graph of $C$ as a digraph $G_C(V(G_C),\overrightarrow{E}(G_C))$, such that for any agent $\agent_i \in C$, there is a vertex $v_i$ in $V(G_C)$ and there is a directed edge from $v_i$ to $v_j$ in $\overrightarrow{E}(G_C)$, if $\agent_i$ envies $\agent_j$. In Lemma \ref{dag}, we show that $G_C$ is acyclic.
\begin{lemma}
\label{dag}
For every cycle-envy-free set of semi-satisfied agents $C$, $G_C$ is a DAG. 
\end{lemma}
\begin{definition}
 A topological ordering of a cycle-envy-free set $C$ of semi-satisfied agents, is a total order $\prec_O$ corresponding to the topological ordering of the representation graph $G_C$. More formally, for the agents $\agent_i,\agent_j \in C$ we have $\agent_i \prec_O \agent_j$ if and only if $v_i$ appears before $v_j$, in the topological ordering of $G_C$.  
\end{definition}

Note that in the topological ordering of a cycle-envy-free set $C$ of semi-satisfied agents, if $\agent_i \in C$ envies $\agent_j \in C$, then $\agent_i \prec_O \agent_j$.

\begin{observation}
\label{epsofcluster}
Let $C$ be a cycle-envy-free set of semi-satisfied agents. Then, for every agent $\agent_i \in C$ such that $\agent_j \prec_O \agent_i$, we have:
$$\valu_i(\firstset_j) \leq 3/4 - \epsilon_i.$$ 

\end{observation}

We define a maximum cardinality maximum weighted matching of a weighted graph as a matching that has the highest number of edges and among them the one that has the highest total sum of edge weights. For brevity we call such a matching an $\MCMWM$. In Lemma \ref{wm}, we show that an $\MCMWM$ of a weighted bipartite graph has certain properties that makes it useful for building cycle-envy-free clusters. 

\begin{lemma}
\label{wm}
Let $H\langle V(H),E(H)\rangle$ be a weighted bipartite graph with $V(H) = \partone \cup \parttwo$ and let $M = \{(\vone_1,\vtwo_1),...,(\vone_k,\vtwo_k)\}$ be an $\MCMWM$ of $H$. Then, for every subset $T \subseteq \{\vtwo_1,\vtwo_2, \ldots,\vtwo_k\}$, the following conditions hold:

\begin{minipage}[t]{\linegoal}
\begin{enumerate}[leftmargin=*]
 \item There exists a vertex $\vtwo_j \in T$ which is a winner in $T$, i.e.,  $w(\vone_j,\vtwo_{j}) \geq w(\vone_i,\vtwo_j)$, for all $\vone_i \in M(T)$ and  $(\vone_i,\vtwo_j) \in E(H)$. 
\item There exists a vertex $ \vtwo_j \in T$ which is a loser in $T$, i.e.,  $w(\vone_i,\vtwo_i)  \geq w(\vone_j,\vtwo_i) $, for all $\vtwo_i \in T$ and $(\vone_j,\vtwo_i) \in E(H)$.
\item For any vertex $\vtwo_i \in T$ and any unsaturated vertex $\vone_j \in \partone$ such that $(\vone_j,\vtwo_i) \in E(H)$, $w(\vone_i,\vtwo_i) \geq w(\vone_j,\vtwo_i)$. 
\end{enumerate}
\end{minipage}
\\[6pt]
where $M(T)$ is the set of vertices which are matched by the vertices of $T$ in $M$.
\end{lemma}

Notice the similarities of the first and the second conditions of Lemma \ref{wm} with the conditions of the winner and loser in 
 Definition \ref{winloose}. In Section \ref{additive:clusters}, we assign items to the agents based on an $\MCMWM$ of the value-graph. Lemma \ref{wm} ensures that such an assignment results in a cycle-envy-free set of semi-satisfied agents.   





\subsection{Phase 1: Building the Clusters}\label{additive:clusters}
In this section, we explain our method for clustering the agents. Intuitively, we divide the agents into three clusters $\cone,\ctwo$ and $\cthree$. As mentioned before, during the algorithm, two sets of items $\firstset_i,\secondset_i$ are allocated to each agent $\agent_i$. Throughout this section, we prove a set of lemmas that are labeled as \emph{value-lemma}. In these lemmas we bound the value of $f_i$ and $g_i$ allocateed to any agent for other agents. A summary of these lemmas is shown in Tables \ref{table0}, \ref{table4} and \ref{table1}.

After constructing each cluster, we refine that cluster. In the refinement phase of each cluster, we target a certain subset of the remaining items. If any item in this subset could satisfy an agent in the recently created cluster, we allocate that item to the corresponding agent. The goal of the refinement phase is to ensure that the remaining items in the targeted subset are light enough for the agents in that cluster, i.e., none of the remaining items can satisfy an agent in this cluster.

We denote by $\satagents$, the set of satisfied agents. In addition, denote by $\satagents_1, \satagents_2$, and $\satagents_3$ the subsets of $\satagents$, where $\satagents_i$ refers to the agents of $\satagents$ that previously belonged to ${\mathcal C}_i$. Furthermore, we use $\satagents_1^r$ and $\satagents_2^r$ to refer to the agents of $\satagents_1$ and $\satagents_2$ that are satisfied in the refinement phases of $\cone$ and $\ctwo$, respectively.
\subsubsection{Cluster $\cone$} \label{cluster1:building}
Consider the filtering $G_{1/2}\langle V_{1/2}(G),E_{1/2}(G) \rangle$ of the value-graph $G$ and let $M$ be an $\MCMWM$ of $G_{1/2}$. We define Cluster $\cone$ as the set of agents whose corresponding vertex is in $N(F_{G_{1/2}}(M,\itemsv_{1/2}))$. 

For brevity, denote by $V_{\cone}$ the set of vertices in $V(G)$ that correspond to the agents of $\cone$. In other words:
$$V_{\cone} = N(F_{G_{1/2}}(M,\itemsv_{1/2})).$$

Also, let $F_{G_{1/2}}(M,\itemsv_{1/2}) $ be $U_1 \cup S_1$, where $U_1$ is the set of unsaturated vertices in $F_{G_{1/2}}(M,\itemsv_{1/2})$ and $S_1$ is the set of the saturated vertices. For each edge $(\itemv_j,\agentv_i) \in M$ such that $\itemv_j \in S_1$, we allocate  item $\ite_j$ to agent $\agent_i$. More precisely, we set $\firstset_i = \{\ite_j \}$. Since $w(\itemv_j,\agentv_i)\geq {1/2}$, we have:
$$\forall \agent_k \in \cone \qquad V_k(f_k) \geq {1/2}.$$
According to the definition of $\epsilon_i$, we have
\begin{equation}
\forall \agent_k \in \cone \qquad \epsilon_k \leq {1/4}.
\end{equation} 

By the definition of $F_{G_{1/2}}$, for every agent which is not in $\cone$, the condition of Lemma \ref{forc2c3} holds. Note that all the agents that are not in $\cone$, belong to either $\ctwo$ or $\cthree$.

\begin{lemma}[value-lemma]
\label{forc2c3}
For all $\agent_i \in \ctwo \cup \cthree$ we have: \[ \forall \agent_j \in \cone \qquad \valu_i(\firstset_j) < 1/2. \]
\end{lemma}

For each vertex $\agentv_i \in V_{\cone}$, denote by $N_{\agentv_i}$ the set of vertices $\itemv_j  \in \itemsv \setminus \itemsv_{1/2}$, where $w(\itemv_j,\agentv_i) \geq \epsilon_i$ and let \[W_1 = U_1 \cup \bigcup_{\agentv_i \in V_{\cone}} N_{\agentv_i}.\]

Note that by definition, for any vertex $\itemv_j \in U_1$ and $\agentv_i \notin V_{\cone}$, there is no edge between $\itemv_j$ and $\agentv_i$ in $G_{1/2}$ and hence $w(\itemv_j,\agentv_i)<1/2$. Also, since the rest of the vertices in $W_1$ are from $\itemsv \setminus \itemsv_{1/2}$, for any vertex $\agentv_i$ and $\itemv_j \in (W_1 \setminus U_1)$, $w(\itemv_j,\agentv_i)<1/2$ holds. Thus, we have the following observation:

\begin{observation}
\label{w1small}
For every item $\ite_j$ with $\itemv_j \in W_1$ and every agent $\agent_i$ with $\agentv_i \notin V_{\cone}$, $\valu_i(\{\ite_j\})<1/2$.
\end{observation}

Now, define $\itemsv'$ and $ \agentsv' $ as follows:

$$\itemsv' = \itemsv \setminus (W_1 \cup S_1),$$ $$\agentsv' = \agentsv \setminus V_{\cone}.$$ 

Let $G'\langle V(G'),E(G')\rangle $ be the induced subgraph of $G$ on $V(G') = \agentsv' \cup \itemsv'$. We use graph $G'$ to build Cluster $\ctwo$.

\subsubsection{Cluster $\cone$ Refinement} Before building Cluster $\ctwo$, we satisfy some of the agents in $\cone$ with the items corresponding to the vertices of $W_1$. Consider the subgraph $G_1 \langle V(G_1),E(G_1) \rangle$ of $G$ with $V(G_1) = W_1 \cup V_{\cone}$. In $G_1$, There is an edge between $\agentv_i \in V_{\cone}$ and $\itemv_j \in W_1$, if $V_i(\{\ite_j\}) \geq \epsilon_i$. Note that $G_1 \langle V(G_1),E(G_1) \rangle$ is not necessarily an induced subgraph of $G$. We use $G_1$ to satisfy a set of agents in $\cone$. To this end, we first show that $G_1$ admits a special type of matching, described in Lemma \ref{nicematch}.

\begin{lemma}
\label{nicematch}
There exists a matching $M_1$ in $G_1$, that saturates all the vertices of $W_1$ and for any edge $(\itemv_i,\agentv_j) \in M_1$ and any unsaturated vertex $\agentv_k \in N(\itemv_i)$, $\agent_k$ does not envy $\agent_j$. 
\end{lemma}

Let $M_1$ be a matching of $G_1$ with the property described in Lemma \ref{nicematch}. For every edge $(\agentv_i,\itemv_j) \in M_1$, we allocate  item $\ite_j$ to agent $\agent_i$ i.e., we set $\secondset_i = \{\ite_j\}$. By the definition, $\agent_i$ is now satisfied. Thus, we remove $\agent_i$ from $\cone$ and add it to $\cal S$. Note that, after refining $\cone$, none of the items whose corresponding vertex is in $\itemsv' \setminus \itemsv'_{1/2}$ can satisfy any remaining agent in $\cone$. Thus, Observation \ref{fsmallc1} holds.

\begin{observation}
\label{fsmallc1}
For every item $\ite_j$ such that $\itemv_j \in \itemsv'$, either $\itemv_j \in \itemsv'_{1/2}$ or for all $\agent_i \in \cone$, $V_i(\{\ite_j\}) < \epsilon_i$.
\end{observation}

At this point, all the agents of $\satagents$ belong to $\satagents_1^r$. Each one of these agents is satisfied with two items, i.e., for any agent $\agent_j \in \satagents_1^r$, $|\firstset_j| = |\secondset_j| = 1$. In Lemma \ref{gsmallc1r} we give an upper bound on $\valu_i(\secondset_j)$ for every agent $\agent_j \in \satagents_1^r$ and every agent $\agent_i$ in $\ctwo \cup \cthree$.  

\begin{lemma}[value-lemma]
\label{gsmallc1r}
For every agent $\agent_i \in \ctwo \cup \cthree$, we have
$$ \forall \agent_j \in \satagents_1^r \qquad \valu_i(\secondset_j)< 1/2.$$
\end{lemma}

Lemmas \ref{gsmallc1r} and  \ref{forc2c3}  state that for every agent $\agent_i \in \ctwo \cup \cthree$  and every agent $\agent_j \in \satagents_1^r$, $\valu_i(\firstset_j)$ and $\valu_i(\secondset_j)$ are upper bounded by $1/2$. This, together with the fact that $|\firstset_j| = |\secondset_j|=1$, results in Lemma \ref{forc2}.
\begin{lemma}
\label{forc2}
For all $\agent_i \notin \cone$, we have
\[ \MMS_{\valu_i}^{|\agents \setminus \satagents_1^r|} ( {\items} \setminus \bigcup_{\agentv_j \in \satagents_1^r} \firstset_j \cup \secondset_j) \geq 1.\]
\end{lemma}

\subsubsection{Cluster $\ctwo$}
Recall graph $G' \langle V(G') , E(G') \rangle$ as described in the last part of Section \ref{cluster1:building} and let $G'_{1/2}\langle V_{1/2}(G'), E_{1/2}(G') \rangle$ be a $1/2$-filtering of $G'$. Lemma \ref{rem} states that the size of the maximum matching between $\itemsv'_{1/2}$ and $\agentsv'_{1/2}$ is $|\itemsv'_{1/2}|$. Also, according to Corollary \ref{remcol}, for any maximum matching $M'$ of $G'_{1/2}$, $F_{G'_{1/2}}(M',\itemsv'_{1/2})$ is empty. In what follows, we increase the size of the maximum matching in  $G'_{1/2}$ by merging the vertices of $\itemsv' \setminus \itemsv'_{1/2}$ as described in Definition \ref{merge}.

\begin{figure}[t!]
    \centering
    \includegraphics[scale=1]{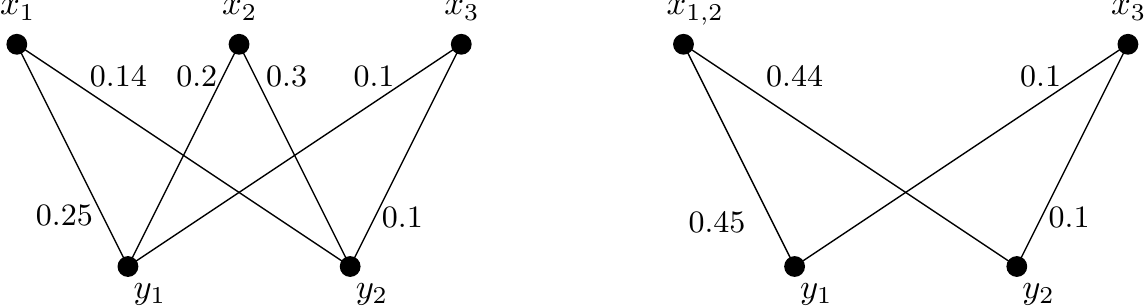}
    \caption{Merging $\itemv_1$ and $\itemv_2$}
    \label{fig:merge}
\end{figure}

\begin{definition}
\label{merge}
For merging vertices $\itemv_i,\itemv_j$ of $G'(\itemsv',\agentsv')$, we create a new vertex labeled with $\itemv_{i,j}$. Next, we add $\itemv_{i,j}$ to $\itemsv'$ and for every vertex $\agentv_k \in \agentsv'$, we add an edge from $\agentv_k$ to $\itemv_{i,j}$ with weight $w(\agentv_k,\itemv_i) + w(\agentv_k,\itemv_j)$. Finally we remove vertices $\itemv_i$ and $\itemv_j$ from $\itemsv$. See Figure ~\ref{fig:merge}.
\end{definition}

In Lemmas \ref{c1small2} and \ref{pairsmall}, we give upper bounds on the value of the pair of items corresponding to a merged vertex. In Lemma \ref{c1small2}, we show that the value of a merged vertex is less than $2\epsilon_i$ to every agent $\agent_i \in \cone$. This fact is a consequence of Observation \ref{fsmallc1}. Also, in Lemma \ref{pairsmall}, we prove that the value of the items corresponding to a merged vertex is less than $3/4$ to any agent. Lemma \ref{pairsmall} is a direct consequence of $3/4$-irreducibility. In fact, we show that if the condition of Lemma \ref{pairsmall} does not hold, then the problem can be reduced. 

\begin{lemma}
\label{c1small2}
For any agent $\agent_k \in \cone$ and any pair of vertices $\itemv_i, \itemv_j \in \itemsv' \setminus \itemsv'_{1/2}$, $\valu_k(\{\ite_i,\ite_j\}) < 2\epsilon_k$ holds. In particular, total value of the items that belong to a merged vertex is less than $2\epsilon_k$ for $\agent_k$.
\end{lemma}

\begin{lemma}
\label{pairsmall}
 For any pair of vertices $\itemv_i , \itemv_j \in \itemsv' \setminus \itemsv'_{1/2}$ and any vertex $\agentv_k \in \agentsv$, we have $V_k(\{\ite_i,\ite_j\}) < {3/4}$.
\end{lemma}

\begin{corollary} [of Lemma \ref{pairsmall}]
\label{forc2small}
For any agent $\agent_i$ with $\agentv_i \in \agentsv$, there is at most one item $\ite_j$, with $\itemv_j \in \itemsv' \setminus \itemsv'_{1/2}$ and $\valu_i(\{\ite_j\}) \geq {3/8}$.
\end{corollary}

Consider the vertices in $\itemsv' \setminus \itemsv'_{1/2}$. We call a pair $(\itemv_i,\itemv_j)$ of distinct vertices in $\itemsv' \setminus \itemsv'_{1/2}$ \textit{desirable} for $\agentv_k \in \agentsv'$, if $w(\agentv_k,\itemv_i) + w(\agentv_k,\itemv_j) \geq {1/2}$. With this in mind, consider the process described in Algorithm \ref{addvertex}. 

In each step of this process, we find an $\MCMWM$ $M'$ of $G'_{1/2}$. Note that $M'$ changes after each step of the algorithm. Next, we find a pair $(\itemv_i,\itemv_j)$ of the vertices in $\itemsv' \setminus \itemsv'_{1/2}$ that is desirable for at least one agent in $T = \agentsv' \setminus N(F_{G'_{1/2}}(M',\itemsv'_{1/2}))$. If no such pair exists, we terminate the algorithm. Otherwise, we select an arbitrary desirable pair $(\itemv_i,\itemv_j)$ and merge them to obtain a vertex $\itemv_{i,j}$. According to the definition of $T$ in Algorithm \ref{addvertex}, merging a pair  $(\itemv_i,\itemv_j)$ results in an augmenting path in $G'_{1/2}$. Hence, the size of the maximum matching in $G'_{1/2}$ is increased by one. Note that after the termination of Algorithm \ref{addvertex}, either $T = \emptyset$ or  no pair of vertices in $\itemsv' \setminus \itemsv'_{1/2}$ is desirable for any vertex in $T$. 

\begin{lemma} 
\label{sizeeq}
After running Algorithm \ref{addvertex}, we have
$$|F_{G'_{1/2}}(M',\itemsv'_{1/2})| = |N(F_{G'_{1/2}}(M',\itemsv'_{1/2}))|.$$  
\end{lemma}

\begin{algorithm}[t!]
 \KwData{$G'(V(G'),E(G'))$}
 \While{True}{
  $M' = \MCMWM \mbox{  of } G'_{1/2}$\; 
  Find $F_{G'_{1/2}}(M',\itemsv'_{1/2})$\;
  $T = \agentsv' \setminus N(F_{G'_{1/2}}(M',\itemsv'_{1/2}))$\;
   $Q = $ Set of all desirable pairs in $\itemsv' \setminus \itemsv'_{1/2}$ for the agents in $T$\;
  \eIf{$ Q = \emptyset$ }{
   STOP\;
   }{
   Select an arbitrary pair $\itemv_i,\itemv_j$ from $Q$\;
   Merge($\itemv_i,\itemv_j$)\;
  }
 }
 \caption{Merging vertices in $G'$}
 \label{addvertex}
\end{algorithm}

We define Cluster $\ctwo$ as the set of agents that correspond to the vertices of $N(F_{G'_{1/2}}(M',\itemsv'_{1/2}))$. Also, denote by $V_{\ctwo}$ the vertices in $N(F_{G'_{1/2}}(M',\itemsv'_{1/2}))$. For each agent $\agent_i \in \ctwo$, we allocate the item corresponding to $M'(\agentv_i)$ (or pair of items in case  $M'(\agentv_i)$ is a merged vertex) to $\agent_i$.

Note that we put the rest of the agents in Cluster $\cthree$. Therefore, Lemma \ref{forc3} holds for all the agents of $\cthree$.

\begin{lemma}[value-lemma]
\label{forc3}
For all $\agent_i \in \cthree$ we have \[ \forall \agent_j \in \ctwo, \valu_i(\firstset_j) < 1/2. \]
\end{lemma}

\subsubsection{Cluster $\ctwo$ Refinement}
The refinement phase of $\ctwo$, is semantically similar to the refinement phase of $\cone$. In the refinement phase of $\ctwo$, we satisfy some of the agents of $\ctwo$ by the items with vertices in $\itemsv' \setminus \itemsv'_{1/2}$. Note that none of the vertices in $\itemsv' \setminus \itemsv'_{1/2}$ is a merged vertex.

The refinement phase of $\ctwo$ is presented in Algorithm \ref{c2ref}. Let $\agent_{i_1}, \agent_{i_2}, \ldots, \agent_{i_k}$ 
be the topological ordering of the agents in $\ctwo$ as described in Section \ref{additive:cef}
. In Algorithm \ref{c2ref}, We start with $\agentv_{i_1}$ and $W_2 = \emptyset$ and check whether there exists a vertex $\itemv_j \in \itemsv' \setminus (\itemsv'_{1/2} \cup W_2)$ such that $V_{i_1}(\{\ite_j\}) \geq \epsilon_{i_1}$. If so, we add $\itemv_j$ to $W_2$ and satisfy $\agent_{i_1}$ by allocating $\ite_j$ to $\agent_{i_1}$. Next, we repeat the same process for $\agentv_{i_2}$ and continue on to $\agentv_{i_k}$. Note that at the end of the process, $W_2$ refers to the vertices whose corresponding items are allocated to the agents that are satisfied during the refinement step of $\ctwo$. For convenience, let $S_2 = F_{G'_{1/2}}(M',\itemsv'_{1/2})$ and define $\itemsv''$ and $\agentsv''$ as follows:
$$\itemsv'' = \itemsv' \setminus (W_2 \cup S_2),$$
$$\agentsv'' = \agentsv' \setminus V_{\ctwo}.$$

Let $G'' \langle V(G''),E(G'') \rangle$ be the induced subgraph of $G'$ on $V(G'') = \itemsv'' \cup \agentsv''$. We use $G''$ to build Cluster $\cthree$.

\begin{algorithm}[t!]
 \KwData{$G'(V(G'),E(G'))$}
 \KwData{$\agent_{i_1},\agent_{i_2},\ldots,\agent_{i_k}$ = Topological ordering of agents in $\ctwo$}
  \For{$l:1\rightarrow k$}{
	\If{$ \exists \itemv_j \in \itemsv' \setminus (\itemsv'_{1/2} \cup W_2)$ s.t. $V_{i_1}(\{\ite_j\}) \geq \epsilon_{i_l}$)}
	{
		$\secondset_{i_l} = \ite_j$ \;
		$W_2 = W_2 \cup \itemv_j$\;
		$\ctwo = \ctwo \setminus \agent_{i_l}$\;
		${\satagents} = {\satagents} \cup \agent_{i_l}$\;
	}
  }
 \caption{Refinement of $\ctwo$}
 \label{c2ref}
\end{algorithm}

\begin{observation} 
\label{fsmallc2}
After running Algorithm \ref{c2ref}, For every item $\ite_j$ with $\itemv_j \in \itemsv'' \setminus \itemsv''_{1/2} $ and every agent $ \agent_i \in \ctwo$, we have $V_i(\{\ite_j\}) < \epsilon_i$. 
\end{observation}

In the following two lemmas, we give upper bounds on the value of $\secondset_i$ for every agent $\agent_i \in \satagents_2^r$. First, in Lemma \ref{cr2smallc1}, we show that for every agent $\agent_j \in \cone$, $\valu_j(\secondset_i)$ is upper bounded by $\epsilon_j$. Furthermore, by the fact that the agents that are not selected for Clusters $\cone$ and $\ctwo$ belong to Cluster $\cthree$, we show that $\valu_j(\secondset_i)$ is upper bounded by $1/2$ for every agent $\agent_j \in \cthree$. 
\begin{lemma}[value-lemma]
\label{cr2smallc1}
Let $\agent_i \in \satagents_2^r$ be an agent that is satisfied in the refinement phase of Cluster $\ctwo$ and $\agent_j$ be an  agent in $\cone$. Then, $\valu_j(\secondset_i)<\epsilon_j$.
\end{lemma}

\begin{lemma}[value-lemma]
\label{cr2smallc3}
Let $\agent_i \in \satagents_2^r$ be an agent that is satisfied in the refinement phase of Cluster $\ctwo$ and $\agent_j$ be an agent in $\cthree$. Then, $\valu_j(\secondset_i)< 1/2$.
\end{lemma}

\subsubsection{Cluster $\cthree$.} Finally, Cluster $\cthree$ is defined as the set of agents corresponding to the vertices of $\agentsv''$. Let $M''$ be an $\MCMWM$ of $G''_{1/2}$. Note that by Lemma \ref{rem}, all the vertices in $\itemsv''_{1/2}$ are saturated by $M''$. 
For each vertex $\agentv_i$ that is saturated by $M''$, we allocate the item (or pair of items in a case that $M''(\agentv_i)$ is a merged vertex) corresponding to $M''(\agentv_i)$ to $\agent_i$. Unlike the previous clusters, this allocation is temporary. A semi-satisfied agent $\agent_i$ in $\cthree$ may \emph{lend} his $f_i$ to the other agents of $\cthree$. Therefore, we have three type of agents in $\cthree$: 
\begin{enumerate}
    \item \textbf{The semi-satisfied agents}: we denote the set of semi-satisfied agents in $\cthree$ by $\cthree^s$
    \item \textbf{The borrower agents}: the agents that may borrow from a semi-satisfied agent. An agent $\agent_j$ in $\cthree$ is a borrower, if $\agent_j \notin \cthree^s$ and $\max_{\agent_i \in \cthree^S} V_j(f_i) \geq {1/2}$. We denote the set of borrower agents in $\cthree$ by $\cthree^b$.
    \item \textbf{The free agents}: the remaining agents in $\cthree$. We denote the set of free agents by $\cthree^f$.
\end{enumerate}
So far, the agents corresponding to unsaturated vertices in $\agentsv''_{1/2}$ belong to $\cthree^b$ and the agents corresponding to the vertices in $\agentsv'' \setminus \agentsv''_{1/2}$ are in $\cthree^f$. As we see, during the second phase, agents in $\cthree$ may change their type. For example, an agent in $\cthree^s$ may move to $\cthree^f$ or vice versa.  For convenience, for every agent $\agent_i \in \cthree^b$, we define $\epsilon_i$ as follows: 
\begin{equation}
\label{borrowers}
 3/4 - \max_{\agent_j \in \cthree^s}\valu_i(\firstset_j)
\end{equation} 
Note that by the definition, $\epsilon_i \leq 1/4$ holds for every agent of $\cthree^b$.

\begin{figure}[t!]
\centering
\includegraphics[scale=0.5]{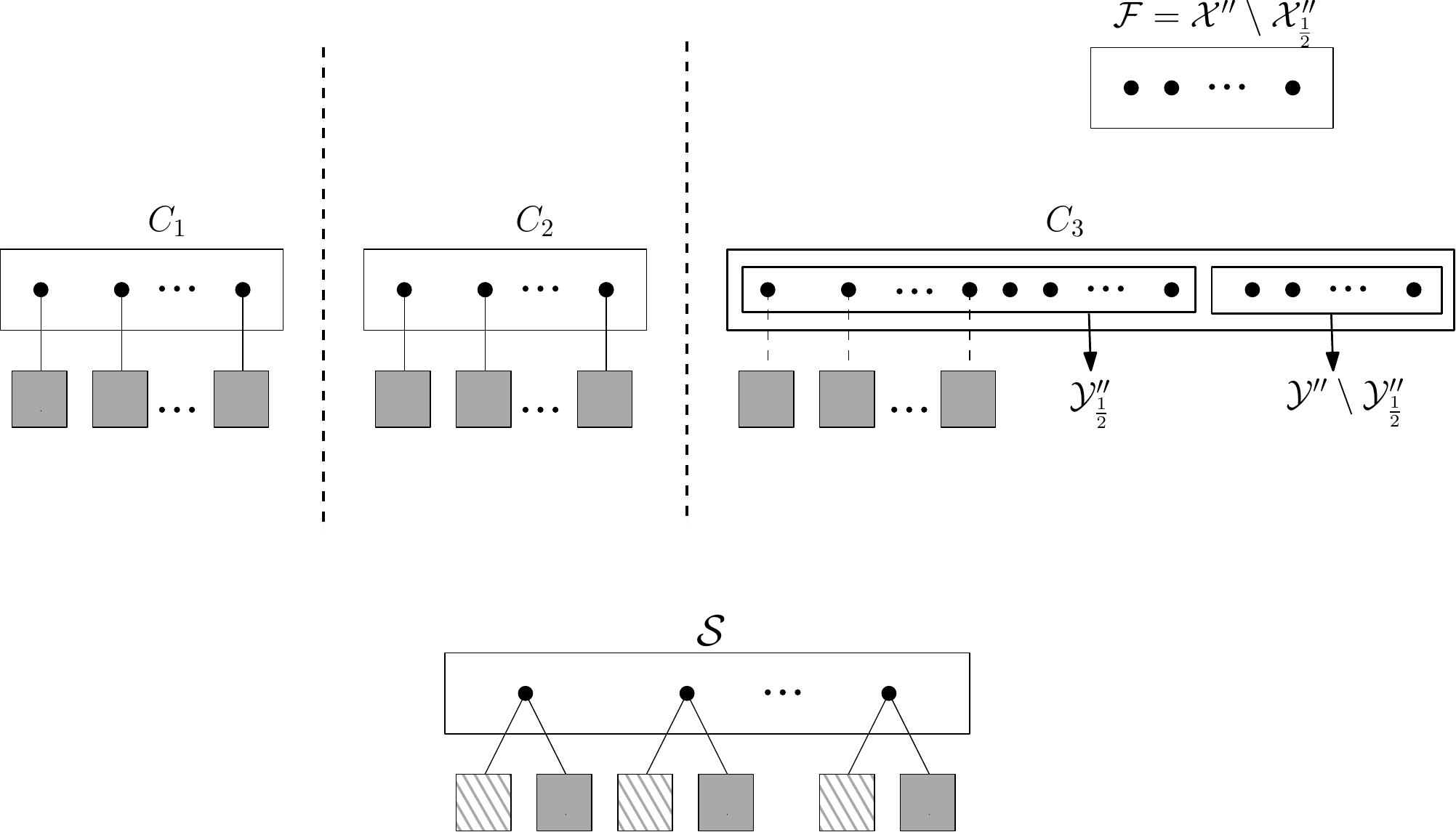}
\caption{Overview on the state of the algorithm}
\label{fig:overview}
\end{figure}

In Lemma \ref{lsmall_c3}, we show that the remaining items are not \emph{heavy} for the agents in $\cthree$. The main reason that Lemma \ref{lsmall_c3} holds, is the fact that no pair of vertices is desirable for any agents in $\cthree$ at the end of Algorithm \ref{addvertex}. 
\begin{lemma}
\label{lsmall_c3}
For all $\agent_i \in \cthree$ and $\itemv_j,\itemv_k \in \itemsv'' \setminus \itemsv''_{1/2}$, we have  $V_i(\{\ite_j,\ite_k\}) < {1/2}$.
\end{lemma}

\begin{corollary}[of Lemma \ref{lsmall_c3}]
\label{small_c3}
For any agent $\agent_i \in \cthree$, there is at most one vertex $\itemv_j \in \itemsv'' \setminus \itemsv''_{1/2}$, such that $V_i(\{\ite_j\}) \geq {1/4}$.
\end{corollary}

\subsection{Phase 2: Satisfying the Agents}\label{additive:allocation}
\subsubsection{An Overview on the State of the Algorithm}
Before going through the second phase, we present an overview of the current state of the agents and items. In Figure \ref{fig:overview}, for every agent $\agent_i \in \cone \cup \ctwo \cup \satagents$, $\firstset_i$ is shown by a gray rectangle and for every agent $\agent_i \in \satagents$, $\secondset_i$  is shown by a hatched rectangle. 

Currently, we know that every agent in $\satagents$ belongs to $\satagents_1^r$ or $\satagents_2^r$. These agents are satisfied in the refinement phases of $\cone$ and $\ctwo$. The rest of the agents will be satisfied in the second phase. For brevity, for $i \leq 2$ we use $\satagents_i^s$ to refer to the agents in $\satagents_i$ that are satisfied in the second phase. More formally, $$\mbox{ for }i=1,2 \qquad \satagents_i^s = \satagents_i \setminus \satagents_i^r .$$

Since we didn't refine Cluster $\cthree$, all the agents in the Cluster $\cthree$ are satisfied in the second phase. As mentioned in the previous section, the item allocation to the semi-satisfied agents in $\cthree$ is temporary; That is, we may alter such allocations later. Therefore, in Figure \ref{fig:overview} we illustrate such allocations by dashed lines.

In this section, we denote the set of free items (the items corresponding to the vertices in $\itemsv''\setminus \itemsv''_{1/2}$ at the end of the first phase) by $\fitems$. By Observations \ref{fsmallc1}, \ref{fsmallc2} and Corollary \ref{small_c3}, we know that the items in $\fitems$ have the following properties:
\begin{enumerate}
\item For every agent $\agent_i$ in $\cone$, $\valu_i(\{\ite_j\}) < \epsilon_i$ holds for all $\ite_j \in \fitems$ (Observation \ref{fsmallc1}).
\item For every agent $\agent_i$ in $\ctwo$, $\valu_i(\{\ite_j\}) < \epsilon_i$ holds for all $\ite_j \in \fitems$ (Observation \ref{fsmallc2}).
\item For every agent $\agent_i$ in $\cthree$, there is at most one item $\ite_j \in \fitems$, such that $\valu_i(\{\ite_j\}) \geq 1/4$ (Corollary \ref{small_c3}).
\end{enumerate}

\begin{table}[t]
	\caption{Summary of value lemmas for $f_i$}
	\label{table0} 
	\begin{center}
\begin{tabular}{|c|c|c|c|}
\hline
	& $\forall \agent_i \in \cone$&  $\forall \agent_i \in \ctwo$ & $\forall \agent_i \in \cthree$\\
\hline
$\forall \agent_j \in \cone$&	- & $\valu_i(\firstset_j) < 1/2$ ($\star$)& $\valu_i(\firstset_j) < 1/2$ ($\star$) \\
\hline
$\forall \agent_j \in \ctwo$ & $ \valu_i(\firstset_j) < 3/4 $ ($\ddagger$)  & - & $\valu_i(\firstset_j) < 1/2$ ($\dagger$)\\ 
\hline
$\forall \agent_j \in \cthree^s$ & $ \valu_i(\firstset_j) < 3/4 $($\ddagger$)  & $\valu_i(\firstset_j) <3/4$($\ddagger$)  &  - \\ 
\hline

\end{tabular}
\end{center}
$\hspace{110pt} \star$: Lemma \ref{forc2c3} $\hspace{10pt} \dagger$: Lemma \ref{forc3} $\hspace{10pt} \ddagger$: Lemma \ref{general}\\
\end{table}

\begin{table}[t]
	\caption{Summary of value lemmas for the agents in $\satagents_i^r$}
	\label{table4} 
\begin{center}
	\begin{tabular}{|c|c|c|c|}
\hline
	&  $\forall \agent_i \in \cone $ & $\forall \agent_i \in \ctwo$ & $\forall \agent_i \in \cthree$\\
\hline
$\forall \agent_j \in \satagents_1^r $	 & - & $\valu_i(\secondset_j)<1/2$ ($\star$)& $\valu_i(\secondset_j)<1/2$ ($\star$)\\
\hline
$\forall \agent_j \in \satagents_2^r $	 & $\valu_i(\secondset_j)<\epsilon_i (\dagger)$ & - & $\valu_i(\secondset_j)<1/2$ ($\ddagger$) \\
\hline
\end{tabular}
\end{center}
$\hspace{105pt}$ $\star$: Lemma \ref{gsmallc1r} $\hspace{10pt}$ $\dagger$: Lemma \ref{cr2smallc1} $\hspace{10pt}$ $\ddagger$: Lemma \ref{cr2smallc3}\\

\end{table}

In summary, items of $\fitems$ are small enough, therefore we can run a process similar to the $\bagfilling$ algorithm described earlier to allocate them to the agents. Recall that our clustering and refinement methods preserve the conditions stated in Lemmas \ref{forc2c3}, \ref{gsmallc1r}, \ref{forc3}, \ref{cr2smallc1} and \ref{cr2smallc3}. In addition to this, we state Lemma \ref{general} as follows.

\begin{lemma}[value-lemma]
\label{general}
For every agent $\agent_i \in \cone \cup \ctwo \cup \cthree^s$, we have
$$\forall \agent_j \in \cone \cup \ctwo \cup \cthree \qquad \valu_j(\firstset_i)<3/4.$$
\end{lemma}  
A brief summary of Lemmas \ref{forc2c3}, \ref{gsmallc1r}, \ref{forc3}, \ref{cr2smallc1}, \ref{cr2smallc3} and \ref{general} is illustrated in Tables \ref{table0} and \ref{table4}. Moreover, since sets $\cone,\ctwo$ and $\cthree^s$ are cycle-envy-free, Observation \ref{epsofcluster} holds for these sets.

\subsubsection{Second Phase: $\bagfilling$}
We begin this section with some definitions. In the following, we define the notion of feasible subsets and, based on that, we define $\phi(S)$ for a feasible subset $S$ of items.
\begin{definition}
A subset $S$ of items in $\fitems$ is feasible, if at least one of the following conditions are met:
\begin{minipage}[t]{\linegoal}
\begin{enumerate}[leftmargin=30pt]
    \item There exists an agent $\agent_i \in \cthree^f $ such that  $\valu_i(\{S\}) \geq {1/2}$. 
    \item There exists an agent $\agent_i \in \cone \cup \ctwo \cup \cthree^s \cup \cthree^b$ such that  $\valu_i(\{S\}) \geq \epsilon_i$.
\end{enumerate}
\end{minipage}
\end{definition}

\begin{definition}
For a feasible set $S$, we define $\Phi(S)$ as the set of agents, that set $S$ is feasible for them. 
\end{definition}

Recall the notion of cycle-envy-freeness and the topological ordering of the agents in a cycle-envy-free set of semi-satisfied agents. Based on this, we define a total order $\prec_{pr}$ to prioritize the agents in the $\bagfilling$ algorithm. 

\begin{definition}
\label{priority}
Define a total order $\prec_{pr}$ on the agents of $\cone \cup \ctwo \cup \cthree$ with the following rules: 
\begin{minipage}[t]{\linegoal}
	
\begin{enumerate}[leftmargin=50pt]
    \item $\agent_{i_5} \prec_{pr} \agent_{i_1} \prec_{pr} \agent_{i_2} \prec_{pr} \agent_{i_3}  \prec_{pr} \agent_{i_4} \qquad$  $\forall \agent_{i_1} \in \cone, \agent_{i_2} \in \ctwo, \agent_{i_3} \in \cthree^s, \agent_{i_4} \in \cthree^b, \agent_{i_5} \in \cthree^f$
    \item $\agent_i \prec_{pr} \agent_j \Leftrightarrow \agent_i \prec_o \agent_j \hspace{85pt}$ $\forall \agent_i, \agent_j \in \cone \cup \ctwo \cup \cthree^s,  \agent_i ,\agent_j \mbox{ in the same cluster }$
	\item $\agent_i \prec_{pr} \agent_j \Leftrightarrow i < j \hspace{100pt}$ $\forall \agent_i,\agent_j \in \cthree^b \vee \agent_i,\agent_j \in \cthree^f$
\end{enumerate}
\end{minipage}
\end{definition}    

Recall that $\prec_o$ refers to the topological ordering of a semi-satisfied set of agents. Roughly speaking, for the semi-satisfied agents in the same cluster, $\prec_{pr}$ behaves in the same way as $\prec_{o}$. Furthermore, for the agents in different clusters, agents in $\cthree^f , \cone , \ctwo, \cthree^s , \cthree^b$ have a lower priority, respectively. Finally, the order of the agents in $\cthree^b$ and $\cthree^f$ is determined by their index, i.e., the agent with a lower index has a lower priority.

The second phase consists of several rounds and every round has two steps. Each of these two steps is described below. We continue running this algorithm until $\fitems$ is no longer feasible for any agent.
\begin{itemize}
\item \textbf{Step1}: In the first step, we run a process very similar to the $\bagfilling$ algorithm described in Section \ref{introduction}. That is, we find a feasible subset $S \subseteq \fitems$, such that $|S|$ is minimal. Such a subset can easily be found, using a slight modification of the $\bagfilling$ process (see Section \ref{sphase}).  

\item \textbf{Step2}: In the second step, we choose an agent to allocate set $S$ to him. In contrast to the $\bagfilling$ algorithm, we do not select an arbitrary agent. Instead, we select the agent in $\Phi(S)$ with the lowest priority regarding $\prec_{pr}$, i.e., smallest element in $\Phi(S)$ regarding $\prec_{pr}$. Let $\agent_i$ be the selected agent. We consider three cases separately:

\begin{minipage}[t]{\linegoal}
\begin{enumerate}[leftmargin=50pt]
    \item $\agent_i \in \cthree^f$: temporarily allocate $S$ to $\agent_i$, i.e., set $\firstset_i = S$. 
    \item $\agent_i \in \cthree^b$: let $\agent_j$ be the agent that $\valu_i(\firstset_j) = {3/4} - \epsilon_j$. 
Take back $\firstset_j$ from $\agent_j$ and allocate $\firstset_j \cup S$ to $\agent_i$ i.e. set $\firstset_i = \firstset_j$, $\firstset_j=\emptyset$ and $\secondset_i = S$. Remove $\agent_i$ from $\cthree$ and add it to $\satagents$.
    \item $\agent_i \in \cone \cup \ctwo \cup \cthree^s$: satisfy agent $\agent_i$ by $S$, i.e., set $\secondset_i = S$ and remove $\agent_i$ from its corresponding cluster and add it to $\satagents$. 
\end{enumerate}
\end{minipage}

By the construction of $\cthree^s,\cthree^b$, and $\cthree^f$, the above process may cause agents in $\cthree$ to move in between $\cthree^s,\cthree^b$ and $\cthree^f$. For example, if the first case happens, then $\agent_i$ is moved from $\cthree^f$ to $\cthree^s$. In addition, all other agents in $\cthree^f$ for which $S$ is feasible are moved to $\cthree^b$. For the second case, $\agent_j$ is moved to one of $\cthree^f$ or $\cthree^b$, based on $\valu_j(\firstset_k)$ for every $\agent_k \in \cthree^s$; that is, if there exists an agent $\agent_k \in \cthree^s$ such that $\valu_j(\firstset_k) \geq 1/2$, $\agent_j$ is moved to $\cthree^b$. Otherwise, $\agent_j$ is moved to $\cthree^f$. For both the second and the third cases, some of the agents in $\cthree^b$ may move to $\cthree^f$. 
\end{itemize}
The second phase terminates, when $\fitems$ is no longer feasible for any agent. More details about the second phase can be found in Algorithm \ref{second-phase}.  In Algorithm \ref{second-phase}, we use $Update(\cthree)$ to refer the process of moving agents among $\cthree^s, \cthree^b$ and $\cthree^f$.

\begin{algorithm}[t!]
 \KwData{$\fitems, \cone,\ctwo,\cthree$}
  \While{$\fitems$ is feasible}{
	$S$ = a minimal feasible subset of $\fitems$ \;
	$\agent_i = $ agent in $\Phi(S)$ with lowest order regarding  $\prec_{pr}$\;
	\If{$\agent_i \in C_3^f$}
	{
		$\firstset_i = S$ \;
		$Update(\cthree)$ \;
	}
	\If{$\agent_i \in \cthree^b$}
	{
		Let $\agent_j$ be the agent that $\valu_i(\firstset_j) = 3/4 - \epsilon_i$ \;
		$\firstset_i = \firstset_j$ \;
		$\secondset_i = S$ \;
		$\satagents = \satagents \cup \agent_i$ \;
		$\firstset_j = \emptyset$\;
		$\cthree = \cthree \setminus \agent_i$ \;
		$Update(\cthree)$ \;
	}
	\If {$\agent_i \in \cthree^s$}
	{
		$\secondset_i = S$\;
		$\satagents = \satagents \cup \agent_i$\;
		$\cthree = \cthree \setminus \agent_i$ \;
		$Update(\cthree)$ \;
	}
	\If {$\agent_i \in \cone \cup \ctwo$}
	{
		$\secondset_i = S$\;
		remove $\agent_i$ from its corresponding cluster \;
		$\satagents = \satagents \cup \agent_i$\;

	}
}
 \caption{The Second Phase}
 \label{second-phase}
\end{algorithm}

In each round of the second phase, either an agent is satisfied or an agent in $\cthree^f$ becomes semi-satisfied. In Lemma \ref{c3fsmall}, we show that if an agent $\agent_i \in \cthree^f$ is selected in some round of the second phase, then $\valu_j(\firstset_i)$ is upper bounded by $2\epsilon_j$ for every agent $\agent_j \in \cthree \cup \ctwo \cup \cone^s \cup \cone^b$. As a consequence of Lemma \ref{c3fsmall}, in Lemma \ref{cef} we show that sets $\cone,\ctwo$ and $\cthree$ remain cycle-envy-free during the second phase. For convenience, we use $\mathbb{R}_z$ to refer to the $z$'th round of the second phase. 

\begin{lemma}
\label{c3fsmall}
Let $\mathbb{R}_z$ be a round of the second phase that an agent $\agent_i \in \cthree^f$ is selected. Then, for every agent $\agent_j \in \cthree \cup \ctwo \cup \cone^s \cup \cone^b$, we have $\valu_j(\firstset_i)<2\epsilon_j<3/4$.
\end{lemma}

\begin{lemma}
\label{cef}
During the second phase, the $\cone,\ctwo$ and $\cthree^s$ maintain the property of cycle-envy-freeness. 
\end{lemma}

Finally, for the rounds that an agents $\agent_i$ is satisfied, Lemmas \ref{prvalue} and \ref{m_1} give upper bounds on the value of $\secondset_i$ for remaining agents in different clusters. 

\begin{lemma}[value-lemma]
\label{prvalue}
Let $\agent_i \in \satagents$ be an agent that is satisfied in the second phase. Then, for every other agent $\agent_j \in \cone \cup \ctwo$ we have:

\begin{minipage}[t]{\linegoal}
\begin{enumerate}[leftmargin=30pt]
\item If $\agent_j \prec_{pr} \agent_i$, then $\valu_j(\secondset_i) < \epsilon_j$.
\item If $\agent_i \prec_{pr} \agent_j$, then $\valu_j(\secondset_i) < 2\epsilon_j$.
\end{enumerate}
\end{minipage}
\end{lemma}

\begin{lemma}[value-lemma]
\label{m_1}
Let $\agent_i$ be an agent in $\satagents_1^s \cup \satagents_2^s$. Then, for every agent $\agent_j \in \cthree$, we have $\valu_j(\secondset_i) < {1/2}$.
\end{lemma}

The results of Lemmas \ref{prvalue} and \ref{m_1} are summarized in Table \ref{table1}.

\begin{table}[t]
	\caption{Summary of value lemmas for $\secondset_i$}
	\label{table1} 
	\begin{center}
\begin{tabular}{ |c|c|c|c|c|}
	\hline
	& $\forall \agent_i \in \cone$ & $\forall \agent_i \in \ctwo$ & $\forall \agent_i \in \cthree$\\ 
	\hline
	$\forall \agent_j \in \satagents_1^s$ & - & $\valu_i(\secondset_j) < 2\epsilon_i (\star)$  & $\valu_i(\secondset_j) < 1/2(\dagger)$     \\  
	\hline
	$\forall \agent_j \in \satagents_2^s$ & $\valu_i(\secondset_j) < \epsilon_i(\star)$ & - & $\valu_i(\secondset_j) < 1/2(\dagger)$  \\
	\hline
	$\forall \agent_j \in \satagents_3$ & $\valu_i(\secondset_j) < \epsilon_i(\star)$  & $\valu_i(\secondset_j) < \epsilon_i(\star)$  & -  \\
	\hline
	
\end{tabular}
\end{center}
$\hspace{150pt}\star$ : Lemma \ref{prvalue} $\hspace{10pt}\dagger$: Lemma \ref{m_1}\\
\end{table}

\subsection{The Algorithm Finds a $3/4$-$\MMS$ Allocation}\label{additiveproofs}
In the rest of this section, we prove that the algorithm finds a $3/4$-$\MMS$ allocation. For the sake of contradiction, suppose that the second phase is terminated, which means $\fitems$ is not feasible anymore, but not all agents are satisfied. Such an unsatisfied agent belongs to one of the Clusters $\cone$ or $\ctwo$, or $\cthree$. In Lemmas \ref{c3null}, \ref{c1null}, and \ref{c2null}, we separately rule out each of these possibilities. This implies that all the agents are satisfied and contradicts the assumption. For brevity the proofs are omitted and included in Appendix \ref{additiveproofappendix}. We begin with Cluster $\cthree$.
\begin{lemma}
	\label{c3null}
	At the end of the algorithm we have $\cthree = \emptyset$.
\end{lemma}
To prove Lemma \ref{c3null} we consider two cases separately. If $\cthree \neq \emptyset$, either there exists an agent $\agent_i \in \cthree^s \cup \cthree^b$ or all the agents of $\cthree$ are in $\cthree^f$. If the former holds, we show $\cthree^s$ is non-empty and assume $\agent_i$ is a winner of $\cthree^s$. We bound the total value of $\agent_i$ for all the items dedicated to other agents and show the value of the remaining items in $\fitems$ is at least $\epsilon_i$ for $\agent_i$. This shows set $\fitems$ is feasible for $\agent_i$ and contradicts the termination of the algorithm. In case all agents of $\cthree$ are in $\cthree^f$, let $\agent_i$ be an arbitrary agent of $\cthree^f$. With a similar argument we show that the value of $\agent_i$ for the remaining unassigned items is at least $3/4$ and conclude that $\fitems$ is feasible for $\agent_i$ which again contradicts the termination of the algorithm.

Next, we prove a similar statement for $\cone$. 
\begin{lemma}
	\label{c1null}
	At the end of the algorithm we have $\cone = \emptyset$.
\end{lemma}
Proof of Lemma \ref{c1null} follows from a coloring argument. Let $\agent_i$ be a winner of $\cone$. We color all items in either blue or white. Roughly speaking, blue items are in a sense \textit{heavy}, i.e., they may have a high valuation to $\agent_i$ whereas white items are somewhat \textit{lighter} and have a low valuation to $\agent_i$. Next, via a double counting argument, we show that $\agent_i$'s value for the items of $\fitems$ is at least $\epsilon_i$ and thus $\fitems$ is feasible for $\agent_i$. This contradicts $\cone = \emptyset$ and shows at the end of the algorithm all agents of $\cone$ are satisfied.

Finally, we show that all the agents in Cluster $\ctwo$ are satisfied by the algorithm.
\begin{lemma}
	\label{c2null}
	At the end of the algorithm we have $\ctwo = \emptyset$.
\end{lemma}
The proof of Lemma \ref{c2null} is a similar to both proofs of Lemmas \ref{c3null} and \ref{c1null}. Let $\agent_i$ be winner of Cluster $\ctwo$. We consider two cases separately. (i) $\epsilon_i \geq 1/8$ and (ii) $\epsilon_i < 1/8$.
In case $\epsilon_i \geq 1/8$, we use a similar argument to the proof of Lemma \ref{c3null} and show $\fitems$ is feasible for $\agent_i$. If $\epsilon_i < 1/8$ we again use a coloring argument, but this time we color the items with 4 different colors. Again, via a double counting argument we show $\fitems$ is feasible for $\agent_i$ and hence every agent of $\ctwo$ is satisfied when the algorithm terminates. 
\begin{theorem}
	\label{34main}
	All the agents are satisfied before the termination of the algorithm.
\end{theorem}
\begin{proof}
	By Lemmas \ref{c3null}, \ref{c1null}, and \ref{c2null}, at the end of the algorithm all agents are satisfied which means each has received a subset of items which is worth at least $3/4$ to him.
\end{proof}

\subsection{Algorithm}\label{additive:algorithm}
In this section, we present a polynomial time algorithm to find a $(3/4-\epsilon)$-$\MMS$ allocation in the additive setting. More precisely, we show that our method for proving the existence of a $3/4$-$\MMS$ allocation can be used to find such an allocation in polynomial time. 
Recall that our algorithm consists of two main phases: The clustering phase and the $\bagfilling$ phase. In Sections \ref{algcluster} and \ref{sphase} we separately explain how to implement each phase of the algorithm in polynomial time. Given this, there are still a few computational issues that need to be resolved. First, in the existential proof, we assume $\MMS_i = 1$ for every agent $\agent_i \in \agents$.  Second, we assume that the problem is $3/4$-irreducible. Both of these assumptions are without loss of generality for the existential proof due to Observation \ref{reducibility} and the fact that one can scale the valuation functions to ensure $\MMS_i =1$ for every agent $\agent_i$. However, the computational aspect of the problem will be affected by these assumptions.  
The first issue can be alleviated by incurring an additional $1+\epsilon$ factor to the approximation guarantee. \epsteinefficient ~\cite{epstein2014efficient} show that for a given additive function $f$, $\MMS_f^n$ can be approximated within a factor $1+\epsilon$ for constant $\epsilon$ in time $\poly(n)$. Thus, we can scale the valuation functions to ensure $\MMS_i = 1$ while losing a factor of at most $1+\epsilon$. Therefore, finding a $(3/4-\epsilon)$-$\MMS$ allocation can be done in polynomial time if the problem is $3/4$-irreducible. Finally, in Section \ref{irre} we show how to reduce the $3/4$-reducible instances and extend the algorithm to all instances of the problem. The algorithm along with the reduction yields Theorem \ref{addpoly}

\begin{theorem}
	\label{addpoly}
	For any $\epsilon > 0$, there exists an algorithm that finds a $(3/4 - \epsilon)$-$\MMS$ allocation in polynomial time. 
\end{theorem}

\subsubsection{The Clustering Phase}\label{algcluster}
Recall that in the clustering phase we cluster the agents into three sets $\cone$,$\ctwo$, and $\cthree$. In order to build Cluster $\cone$, we find an $\MCMWM$ of the $1/2$-filtering of the value graph. This can be trivially done in polynomial time since finding an $\MCMWM$ is polynomially tractable~\cite{west2001introduction}. However, the refinement phase of Cluster $\cone$ requires finding $F_G(\itemsv,M)$ for a giving graph $G$ and a matching $M$. In what follows, we show this problem can also be solved in polynomial time.

Notice that finding an $\MCMWM$ of $G$ can be done in polynomial time~\cite{west2001introduction}. Therefore, in order to determine $F_H(M,\partone)$, it only suffices to find the vertices of $\partone$ that are reachable from the unmatched vertices of $\parttwo$ by an alternating path. Let $\hat{X}$ be the set of these vertices. We can find $\hat{X}$ using a depth-first-search from the unmatched vertices of $\parttwo$. By definition, $F_H(M,\partone) = \parttwo \setminus \hat{X}$. Therefore, $F_H(M,\partone)$ can be found in polynomial time.

In addition to $F_G(\itemsv,M)$, we also need to find a matching of the graph which satisfies the conditions of Lemma \ref{nicematch}. We show in the following that this problem also can be solved in polynomial time. First, note that in Lemma \ref{v1size} we prove that $G_1$ has a matching that saturates all the vertices of $W_1$. Now, let $p_{\agent_k}$ be the position of $\agent_k$ in the topological ordering of $\cone$, as described in the proof of Lemma \ref{nicematch}. Furthermore, Let $M_1$ be a matching that minimizes the following expression.
$$ \sum_{(x_j,y_i) \in M_1} p_i.$$ Recall that in the proof Lemma \ref{nicematch}, we show that $M_1$ satisfies the condition described in Lemma \ref{nicematch}. Here, we show that $M_1$ can be found in polynomial time. To this end, we model this with a network design problem. 

Orient every edge $(x_j,y_i) \in G_1$ from $y_i$ to $x_j$ and set the cost of this edge to $p_{a_i}$. Also, add a source node $s$ and add a directed edge from $s$  to every vertex of $V_{\cone}$ with cost $0$. Furthermore, add a sink node $t$ and add directed edges from the vertices of $W_1$ to $t$ with cost $0$. Finally, set the capacity of all edges to $1$. 

One can observe that in a minimum cost maximum flow from $s$ to $t$ in this network, the edges with non-zero flow between $V_{\cone}$ and $W_1$ form a maximum matching $M_1$. In addition to this, since the cost of the flow is minimal, $\sum_{(x_j,y_i) \in M_1} cost(x_j,y_i)$ is minimized. Therefore, in this matching, 
$\sum_{(x_j,y_i) \in M_1} p_i$
is minimized. Thus, the matching with desired properties of Lemma \ref{nicematch} can be found in polynomial time.

The same algorithms can be used to compute Cluster $\ctwo$. Finally, we put the rest of the agents in Cluster $\cthree$.

\subsubsection{The $\bagfilling$ Phase}\label{sphase}
In each round of the second phase, we iteratively find a minimal feasible subset of $\fitems$ and allocate its items to the agent with the lowest priority in $\Phi(S)$.  Note that for a feasible set $S$, one can trivially find the agent with lowest priority in $\Phi(S)$ in polynomial time. Thus, it only remains to show that we can find a minimal feasible subset of $\fitems$ in polynomial time. 

Consider the following algorithm, namely \emph{reverse $\bagfilling$ algorithm}: Start with a bag containing all the items of $\fitems$ and so long as there exists an item $\ite_j$ in the bag such that after removing $\ite_j$, the set of items in the bag is still feasible, remove $\ite_j$ from the bag. After this process, the remaining items in the bag  form  a minimally feasible subset of $\fitems$. Therefore, this phase can be run in polynomial time.

\subsubsection{Reducibility}\label{irre}
The most challenging part of our algorithm is dealing with the $3/4$-irreducibility assumption. The catch is that, in order to run the algorithm, we don't necessarily need the $3/4$-irreducibility assumption. Recall that we leverage the following three consequences of irreducibility to prove the existential theorem.
\begin{itemize}
	\item The value of every item in $\items$ is less that $3/4$ to every agent.
	\item Every pair of items in $\itemsv'' \setminus \itemsv''_{1/2}$ is in total worth less than $3/4$ to any agent.
	\item The condition of Lemma \ref{v1size} holds.
\end{itemize}
 Therefore, the algorithm works so long as the mentioned conditions hold. Note that, although it is not clear whether determining if an instance of the problem is $3/4$-reducible is polynomially tractable, all of the above conditions can be validated in polynomial time. This is trivial for the first two conditions; we iterate over all items or pairs of items and check if the condition holds for these items. The last condition, however, is harder to validate.


The condition of Lemma \ref{v1size} holds if for all $S \subseteq W_1$, $|N(S)| > |S|$. Recall that in the proof of Lemma \ref{v1size} we showed that if this condition does not hold, then $F_{G_1}(M,\itemsv)$ is non-empty. Next, we showed that if $F_{G_1}(M,\itemsv)$ is non-empty, then we can reduce the problem via satisfying every agents of $F_{G_1}(M,\itemsv)$ by his matched item in $M$. Therefore, on the computational side, we only need to find whether $F_{G_1}(M,\itemsv)$ is empty which indeed can be determined in polynomial time. 


Note that every time we reduce the problem, $|\agents|$ is decreased by at least $1$, which implies the number of times we reduce the problem is no more than $n$. Moreover, our reduction takes a polynomial time. Thus, the running time of the algorithm is polynomial. 

\section{Submodular Agents}\label{submodular}
Previous work on the fair allocation problem was limited to the additive agents \cite{amanatidis2015approximation,Procaccia:first}. In real-world, however, valuation functions are usually more complex than additive ones. As an example, imagine an agent is interested in at most $k$ items. More precisely, he is indifferent between receiving $k$ items or more than $k$ items. Such a valuation function is called $k$-demand and cannot be modeled by additive functions. $k$-demand functions are a subclass of submodular set functions which have been extensively studied the literature of different contexts, e.g., optimization, mechanism design, and game theory \cite{buchbinder2012tight,buchbinder2015tight, fujishige2005submodular,gupta2010constrained,kim2011distributed,krause2010sfo,lee2009non,minoux1978accelerated,vondrak2008optimal}. 

 In this section, we study the fair  allocation problem where the valuations of agents are submodular. We begin by presenting an impossibility result; We show in Section \ref{submodular-upperbound} that the best guarantee that we can achieve for submodular agents is upper bounded by $3/4$. Next, we give a proof to the existence of a $1/3$-$\MMS$ allocation in this setting. This is followed by an algorithm that finds such an allocation in polynomial time. This is surprising since even finding the $\MMS$ of a submodular function is NP-hard and cannot be implemented in polynomial time unless P=NP ~\cite{epstein2014efficient}. In our algorithm, we assume we have access to query oracle for the valuation of agents; That is, for any set $S$ and any agent $\agent_i$, $\valu_i(S)$ can be computed via a given query oracle in time $O(1)$. 
 
\subsection{Upper Bound}\label{submodular-upperbound}
We begin by providing an upper bound. In this section, we show for some instances of the problem with submodular agents, no allocation can be better than $3/4$-$\MMS$. Our counter-example is generic; We show this result for any number of agents.
\begin{theorem}\label{subupperbound}
	For any $n \geq 2$, there exists an instance of the fair allocation problem with $n$ submodular agents where no allocation is better than $3/4$-$\MMS$.
\end{theorem}
\begin{proof}
	We construct an instance of the problem that does not admit any $3/4+\epsilon$-$\MMS$ allocation. To this end, let $n$ be the number of agents and $\items = \{\ite_1, \ite_2, \ldots,\ite_m\}$ where $m = 2n$. Furthermore, let $f:2^{\items} \rightarrow \mathbb{R}$ be as follows:
	
	$$f(S) =
	\begin{cases}
	0, & \text{if }|S| = \emptyset \\
	1, & \text{if }|S| = 1  \\
	2, & \text{if }|S| > 2 \\
	2, & \text{if }S = \{\ite_{2i}, \ite_{2i+1} \} \text{ for some }i \\
	3/2, & \text{if }|S| = 2 \text{ and }S \neq \{\ite_{2i}, \ite_{2i+1} \} \text{ for any }i. \\
	\end{cases}$$
	
	Notice that $\MMS_f^n = 2$. Moreover, in what follows we show that $f$ is submodular. To this end, suppose for the sake of contradiction that there exist sets $S$ and $S'$ such that $S \subseteq S'$ and for some element $\ite_i$ we have:
	\begin{equation}
	f(S' \cup \{\ite_i\}) - f(S') > f(S \cup \{\ite_i\}) - f(S). \label{gharch}
	\end{equation}
	Since $f$ is monotone and $S' \neq S$, $f(S' \cup \{\ite_i\}) - f(S') > 0$ holds and thus $S'$ cannot have more than two items. Therefore, $S'$ contains at most two items and thus $S$ is either empty or contains a single element. If $S$ is empty, then adding every element to $S$ has the highest increase in the value of $S$ and thus Inequality \eqref{gharch} doesn't hold. Therefore, $S$ contains a single element and $S'$ contains exactly two elements. Thus, $f(S) = 1$ and $f(S') \geq 3/2$. Therefore, f(S $\cup \{\ite_i\}) - f(S) \geq 1/2$ and $f(S' \cup \{\ite_i\}) - f(S') \leq 1/2$ which contradicts Inequality \eqref{gharch}.
	
	 Now, for agents $\agent_1, \agent_2, \ldots, \agent_{n-1}$ we set $\valu_i = f$ and for agent $\agent_n$ we set $\valu_n = f(\inc(S))$ where $\ite_i$ is in $\inc(S)$ if and only if either $i > 1$ and $\ite_{i-1} \in S$ or $i=1$ and $\ite_m \in S$. 
	
	The crux of the argument is that for any allocation of the items to the agents, someone receives a value of at most $3/2$. In case an agent receives fewer than two items, his valuation for his items would be at most $1$. Similarly, if an agent receives more than two items, someone has to receive fewer than $2$ items and the proof is complete. Therefore, the only case to investigate is where everybody receives exactly two items. We show in such cases, $\min \valu_i(A_i) = 3/2$ for all possible allocations. If all agents $\agent_1, \agent_2, \ldots, \agent_{n-1}$ receive two items whose value for them is exactly equal to $2$, then by the construction of $f$, the value of the remaining items is also equal to $2$ to them. Thus, $\agent_n$'s valuation for the items he receives is equal to $3/2$.
\end{proof}
 
Remark that one could replace function $f$ with an XOS function 
	$$g(S) =
	\begin{cases}
	0, & \text{if }|S| = \emptyset \\
	1, & \text{if }|S| = 1  \\
	2, & \text{if }|S| > 2 \\
	2, & \text{if }S = \{\ite_{2i}, \ite_{2i+1} \} \text{ for some }i \\
	1, & \text{if }|S| = 2 \text{ and }S \neq \{\ite_{2i}, \ite_{2i+1} \} \text{ for any }i. \\
	\end{cases}$$
and make the same argument to achieve a $1/2$-$\MMS$ upper bound for XOS and subadditive agents.

\begin{theorem}\label{xosupperbound}
	For any $n > 1$, there exists an instance of the fair allocation problem with $n$ XOS agents where no allocation is better than $1/2$-$\MMS$.
\end{theorem}

\subsection{Existential Proof}\label{submodularep}
In this section we provide an existential proof to a $1/3$-$\MMS$ allocation. Due to the algorithmic nature of the proof, we show in Section \ref{submodularalgorithm} that such an allocation can be computed in time $\mathsf{poly}(n,m)$. For simplicity, we scale the valuation functions to ensure $\MMS_i = 1$ for every agent $\agent_i$.

We begin by introducing the ceiling functions.
\begin{definition}\label{fxfunction}
	Given a set function $f(.)$, we define $\ceil{f}{x}(.)$ as follows:
	
	$$\ceil{f}{x}(S) =
	\begin{cases}
	f(S), & \text{if }f(S) \leq x \\
	x, & \text{if }f(S) > x.
	\end{cases}$$
\end{definition}
A nice property of the ceiling functions is that they preserve submodularity, fractionally subadditivity, and sub-additivity as we show in Appendix \ref{submodular-appendix}.

\begin{lemma}\label{ceilingfunctions}
	For any real number $x \geq 0$, we have:
	\begin{enumerate}
		\item Given a submodular set function $f(.)$, $f^x(.)$ is submodular.
		\item Given an XOS set function $f(.)$, $f^x(.)$ is XOS.
		\item Given an subadditive set function $f(.)$, $f^x(.)$ is also subadditive.
	\end{enumerate}
\end{lemma}

The idea behind the existence of a $1/3$-$\MMS$ allocation is simple: Suppose the problem is $1/3$-irreducible and let  $\mathcal{A} = \langle A_1, A_2, \ldots, A_n\rangle$ be an allocation of items to the agents that maximizes the following expression:
\begin{equation}\label{exex}
\sum_{\agent_i \in \agents} \ceil{\valu_i}{2/3}(A_i)
\end{equation}
We refer to Expression \eqref{exex} by $\mathsf{ex}^{(2/3)}(\mathcal{A})$. We prove $\valu_i(A_i) \geq 1/3$ for every agent $\agent_i \in \agents$. By the reducibility principal, it only suffices to show every $1/3$-irreducible instance of the problem admits a $1/3$-$\MMS$ allocation. The main ingredients of the proof are Lemmas \ref{remove1}, \ref{submodularaval} and \ref{submodulardovom}. For brevity we skip the proofs and include them in Appendix \ref{submodular-appendix}.

\begin{lemma}\label{submodularaval}
	Let $S_1, S_2, \ldots, S_k$ be $k$ disjoint sets and $f_1, f_2, \ldots, f_k$ be $k$ submodular functions. We remove an element $e$ from $\bigcup S_i$ uniformly at random to obtain sets $S^*_1 = S_1 \setminus \{e\}, S^*_2 = S_2 \setminus \{e\}, \ldots, S^*_k = S_k \setminus \{e\}$. In this case we have
	$$\mathbb{E}[\sum f_i(S^*_i)] \geq \sum f_i(S_i)\frac{|\bigcup S_i| -1}{|\bigcup S_i|}.$$
\end{lemma}

The high-level intuition behind the proof of Lemma \ref{submodularaval} is as follows: For submodular functions, the smaller the size of a set is, the higher the marginal values for adding items to that set will be. Based on that, we show the summation of marginal decreases for removing each element is bounded by the total value of the set and that completes the proof. A complete proof is included in Appendix \ref{submodular-appendix}. 

\begin{lemma}\label{submodulardovom}
	Let $f$ be a submodular function and $S_1, S_2, \ldots, S_k$ be $k$ disjoint sets such that $f(S_i) \geq 1$ for every set $S_i$. Moreover, let $S \subseteq \bigcup S_i$ be a set such that $f(S) < 1/3$. If we pick an element $\{e\}$ of $\bigcup S_i \setminus S$ uniformly at random, we have:
	$$\mathbb{E}[f(S \cup \{e\}) - f(S)] \geq \frac{2k/3}{|\bigcup S_i \setminus S|}.$$ 
\end{lemma}
The proof of Lemma \ref{submodulardovom} is very similar to that of Lemma \ref{submodularaval}. The main point is that in submodular functions, the marginal increase decreases as the sizes of sets grow.

Next, we show the fair allocation problem with submodular agents admits a $1/3$-$\MMS$ allocation\footnote{Almost one year after the first draft of our work, the existense of a $1/10$-$\MMS$ allocation in the submodular case along with an algorithm to find a $1/31$ approximation algorithm for the submodular case is also proved in ~\cite{barman2017approximation}. They also study the problem in the additive setting and present another $2/3$-$\MMS$ algorithm. This work is completely parallel to and independent of our paper. Moreover, their analysis is fundamentally different from our analysis and also their bounds are looser.}.
\begin{theorem}\label{submodulartheorem}
	The fair allocation problem with submodular agents admits a $1/3$-$\MMS$ allocation.
\end{theorem}
\begin{proof}
	By Lemma \ref{reducibility}, the problem boils down to the case of $1/3$-irreducible instances. Let the problem be $1/3$-irreducible and $\mathcal{A}$ be an allocation that maximizes $\mathsf{ex}^{(2/3)}$. Suppose for the sake of contradiction that $\valu_i(A_i) < 1/3$ for some agent $\agent_i$.
	In this case we select an item $\ite_r$ from $\items \setminus A_i$ uniformly at random to create a new allocation $\mathcal{A}^r$ as follows:
	 
	 $$A^r_j =
	 \begin{cases}
	 A_j \setminus \{\ite_r\}, & \text{if }i \neq j \\
	 A_j \cup \{\ite_r\} & \text{if }i = j.
	 \end{cases}$$
	
	In the rest we show $\mathbb{E}[\mathsf{ex}^{(2/3)}(\mathcal{A}^r)] > \mathsf{ex}^{(2/3)}(\mathcal{A})$ which contradicts the maximality of $\mathcal{A}$. 
	Note that by Lemma \ref{submodularaval} the following inequality holds: 
	\begin{equation}\label{rex1}
	\mathbb{E}[\sum_{j\neq i} \valu_j^{2/3}(A^r_j)] \geq \sum_{j\neq i} \valu_j^{2/3}(A_j) \frac{|\items \setminus A_i| - 1}{|\items \setminus A_i|}.
	\end{equation}
	Moreover, by Lemma \ref{submodulardovom} we have 
	\begin{equation}\label{rex2}
	\mathbb{E}[\valu_i(A^r_i) - \valu_i(A_i)] \geq \frac{2n/3}{|\items \setminus A_i|}.
	\end{equation}
	Inequality \eqref{rex1} along with Inequality \eqref{rex2} shows
	\begin{equation}
	\begin{split}
	\mathbb{E}[\mathsf{ex}^{(2/3)}(\mathcal{A}^r)] &= \mathbb{E}[\sum_{j\neq i} \valu_j^{2/3}(A^r_j)] + \mathbb{E}[\valu_i(A^r_i)]\\
	&\geq \sum_{j\neq i} \valu_j^{2/3}(A_j) \frac{|\items \setminus A_i| - 1}{|\items \setminus A_i|} + \mathbb{E}[\valu_i(A^r_i)]\\
	&\geq \sum_{j\neq i} \valu_j^{2/3}(A_j) \frac{|\items \setminus A_i| - 1}{|\items \setminus A_i|} + \frac{2n/3}{|\items \setminus A_i|} + \valu_i(A_i)\\
	&\geq \sum_{j\neq i} \valu_j^{2/3}(A_j) \frac{|\items \setminus A_i| - 1}{|\items \setminus A_i|} + \frac{2n/3}{|\items \setminus A_i|} + \valu_i^{(2/3)}(A_i)\\
	&\geq \sum_{j\neq i} \valu_j^{2/3}(A_j) \frac{|\items \setminus A_i| - 1}{|\items \setminus A_i|} + \frac{2n/3}{|\items \setminus A_i|} + \valu_i^{(2/3)}(A_i)\frac{|\items \setminus A_i| - 1}{|\items \setminus A_i|}\\
	&=  \mathsf{ex}^{(2/3)}(\mathcal{A}) \frac{|\items \setminus A_i| - 1}{|\items \setminus A_i|} + \frac{2n/3}{|\items \setminus A_i|}.\\
	\end{split}
	\end{equation}
Recall that by Lemma \ref{remove1}, the value of agent $\agent_i$ for any item alone is bounded by $1/3$ and thus $\mathbb{E}[\valu_i(A^r_i) - \valu_i(A_i)] = \mathbb{E}[\valu^{2/3}_i(A^r_i) - \valu^{2/3}_i(A_i)]$. 
Notice that by the definition, $\valu_j^{(2/3)}$ is always bounded by $2/3$ and also $\valu_i(A_i) < 1/3$, therefore, $\mathsf{ex}^{(2/3)}(\mathcal{A}) \leq 2n/3-1/3$ and thus
\begin{equation}
\begin{split}
\mathbb{E}[\mathsf{ex}^{(2/3)}(\mathcal{A}^r)] &\geq  \mathsf{ex}^{(2/3)}(\mathcal{A}) \frac{|\items \setminus A_i| - 1}{|\items \setminus A_i|} + \frac{2n/3}{|\items \setminus A_i|}\\
&\geq \mathsf{ex}^{(2/3)}(\mathcal{A}) + \frac{1/3}{|\items \setminus A_i|}\\
&\geq \mathsf{ex}^{(2/3)}(\mathcal{A}) + 1/3m.
\end{split}
\end{equation}
\end{proof}

\subsection{Algorithm}\label{submodularalgorithm}
In this section we give an algorithm to find a $1/3$-$\MMS$ allocation for submodular agents. We show our algorithm runs in time $\poly(n,m)$.

For simplicity, we assume for every agent $\agent_i$, $\MMS_i$ is given as input to the algorithm. However, computing $\MMS_i$ alone is an NP-hard problem. That said, we show in Section \ref{latter} that such a computational barrier can be lifted by a combinatorial trick. We refer the reader to Section \ref{latter} for a more detailed discussion. The procedure is illustrated in Algorithm \ref{submodularalg}:
\begin{algorithm}
	\KwData{$\agents, \items, \langle \valu_1, \valu_2, \ldots, \valu_n\rangle, \langle \MMS_1, \MMS_2, \ldots, \MMS_n\rangle$}
	For every $\agent_j$, scale $\valu_j$ to  ensure $\MMS_j = 1$\;
	\While{there exist an agent $\agent_i$ and an item $\ite_j$ such that $\valu_i(\{\ite_j\}) \geq 1/3$}{
		Allocate $\{\ite_j\}$ to $\agent_i$\;
		$\items = \items \setminus \ite_j$\;
		$\agents = \agents \setminus \agent_i$\;
	}
	$\mathcal{A} = $ an arbitrary allocation of the items to the agents\;
	\While{$\min \valu^{2/3}_j(A_j) < 1/3$}{
		$i = $ the agent who receives the lowest value in allocation $\mathcal{A}$\; 
		Find an item $\ite_e$ such that:\label{find}
		$\mathsf{ex}(\langle A_1 \setminus \{\ite_e\}, A_2 \setminus \{\ite_e\}, \ldots, A_{i-1} \setminus \{\ite_e\}, A_i \cup \{\ite_e\}, A_{i+1} \setminus \{\ite_e\}, \ldots, A_n \setminus \{\ite_e\}\rangle) \geq \mathsf{ex}(\mathcal{A})+1/3m$\;
		$\mathcal{A} = \langle A_1 \setminus \{\ite_e\}, A_2 \setminus \{\ite_e\}, \ldots, A_{i-1} \setminus \{\ite_e\}, A_i \cup \{\ite_e\}, A_{i+1} \setminus \{\ite_e\}, \ldots, A_n \setminus \{\ite_e\}\rangle$\;
	}
	For every $\agent_i \in \agents$ allocate $A_i$ to $\agent_i$\;
	\caption{Finding a $1/3$-$\MMS$ allocation for submodular agents}
	\label{submodularalg}
\end{algorithm}
Based on Theorem \ref{submodulartheorem}, one can show that in every iteration of the algorithm value of $\mathsf{ex}^{2/3}(\mathcal{A})$ is increased by at least $1/3m$. Moreover, such an element $\ite_e$ can be easily found by iterating over all items in time $O(m)$. Furthermore, the number of iterations of the algorithm is bounded by $2nm$, since $\mathsf{ex}^{2/3}(\mathcal{A})$ is bounded by $2n/3$. Therefore, Algorithm \ref{submodularalg} finds a $1/3$-$\MMS$ allocation in time $\poly(n,m)$.

\begin{theorem}\label{subsubalg}
	Given access to query oracles, one can find a $1/3$-$\MMS$ allocation for submodular agents in polynomial time.
\end{theorem}

As a corollary of Theorem \ref{subsubalg}, one can show that the problem of finding the maxmin value of a submodular function admits a 3 approximation algorithm.

\begin{corollary}
	For a given submodular function $f$, we can in polynomial time split the elements of ground set into $n$ dijsoint sets $S_1, S_2, \ldots, S_n$ such that 
	$$f(S_i) \geq \MMS_f^n/3$$
	for every $1 \leq i \leq n$.
\end{corollary}
\section{XOS Agents}\label{xos}
Class of fractionally subadditive (XOS) set functions is a super class of submodular functions. These functions too, have been subject of many studies in recent years \cite{christodoulou2008bayesian, bhawalkar2011welfare, feige2009maximizing,blumrosen2007welfare, syrgkanis2012bayesian,feldman2013simultaneous,fu2012conditional,feldmancombinatorial,milchtaich1996congestion}. Similar to sub-modular functions, in this section we show a $1/5$-$\MMS$ allocation is possible when all agents have XOS valuations. Furthermore, we complement our proof by providing a polynomial algorithm to find a $1/8$-$\MMS$ allocation in Section \ref{xosalgorithm}.


\subsection{Existential Proof}\label{ep}
In this section we show every instance of the fair allocation problem with XOS agents admits a $1/5$-$\MMS$ allocation. 
Without loss of generality, we assume $\MMS_i = 1$ for every agent $\agent_i$. Recall the definition of ceiling functions.
\begin{definition}\label{fxfunction}
Given a set function $f(.)$, we define $\ceil{f}{x}(.)$ as follows:
$$\ceil{f}{x}(S) =
\begin{cases}
	f(S), & \text{if }f(S) \leq x \\
	x, & \text{if }f(S) > x.
\end{cases}$$

\end{definition}

As stated in Lemma \ref{ceilingfunctions}, for every XOS function and every real number $x \geq 0$, $f^x$ is also XOS. The proof of this section is similar to the result of Section \ref{submodularalg}. However, the details are different since XOS functions do not adhere to the nice structure of submodular functions. For every allocation $\mathcal{B}$, we define $\mathsf{ex}^{2/5}(\mathcal{B})$ as follows:
 
$$\mathsf{ex}^{2/5}(\mathcal{B}) = \sum_{\agent_i \in \agents} \ceil{\valu_i}{2/5}(B_i).$$

Now Let $\mathcal{A} = \langle A_1, A_2, \ldots, A_n\rangle$ be an allocation of items to the agents that maximizes $\mathsf{ex}^{2/5}$. Provided that the problem is $1/5$-irreducible, we show $\mathcal{A}$ is a $1/5$-$\MMS$ allocation. Before we proceed to the main proof, we state Lemmas  \ref{xos2lemma}, and \ref{2nsets} as auxiliary observations. 

\begin{lemma}\label{xos2lemma}
Let $f(.)$ be an XOS set function and  $f(S) = \beta$ for a set $S \subseteq \domp(f)$. If we divide $S$ into $k$ (possibly empty)  sets $S_1, S_2, \ldots, S_k$ then 
$$\sum_{i=1}^k \Big(f(S) - f(S\setminus S_i)\Big) \leq f(S).$$
\end{lemma}
The complete proof of Lemma \ref{xos2lemma} is included in Appendix \ref{xosappendix}. Roughly speaking, the proof follows from the fact that for at least one of the additive set functions in the representation of $f$, we have $g_j(S) = \beta$. The rest of the proof is trivial by the additive properties of $g_j$.

By Lemma \ref{remove1}, we know that in every $1/5$-irreducible instance of the problem, the value of every item for a person is bounded by $1/5$. 
For XOS functions, we again, leverage the reducibility principal to show another important property of the $1/5$-irreducible instances of the problem. 
\begin{lemma}
\label{2nsets}
In a $1/5$-irreducible instance of the problem, for a given agent $\agent_i$ we can divide the items into $2n$ sets $S_1, S_2, \ldots, S_{2n}$ such that
$$\valu_i(S_i) \geq 2/5$$
for every $1 \leq i \leq 2n$.
\end{lemma}
We first apply Lemma \ref{remove1} and show in such instances of the problem the valuation of every agent for every item is bounded by $1/5$. We remark that for every agent $\agent_i$, one can split the items into $n$ partitions such that each partition is worth at least $1$ to $\agent_i$. Combining the two observations, we conclude that such a decomposition is possible for every agent $\agent_i$. The full proof of this lemma is included in Appendix \ref{xosappendix}. Next we prove the main theorem of this section.
\begin{theorem}\label{xosproof}
The fair allocation problem with XOS agents admits a $1/5$-$\MMS$ allocation.
\end{theorem}
\begin{proof}
Similar to what we did in Section \ref{additive}, we only prove this for $1/5$-irreducible instances of the problem. By Observation \ref{reducibility}, we can extend this result to all instances of the problem.

Consider an allocation $\mathcal{A} = \langle A_1, A_2, \ldots, A_n\rangle$ of items to the agents that maximizes $\mathsf{ex}^{2/5}$.
We show that such an allocation is $1/5$-$\MMS$. Suppose for the sake of contradiction that there exists an agent $\agent_i$ who receives a set of items which are together of worth less than $1/5$ to him. More precisely,
$$\valu^{2/5}_i(A_i) = \valu_i(A_i) < 1/5.$$ 
Since the problem is $1/5$-irreducible, by Lemma \ref{2nsets}, we can divide the items into $2n$ sets $S_1, S_2, \ldots, S_{2n}$ such that $\valu_i(S_j) \geq 2/5$ for every $1 \leq j \leq 2n$. Note that in this case, $\valu^{2/5}_i(S_j) = 2/5$ follows from the definition. Moreover by monotonicity, $\valu^{2/5}_i(S_j \cup A_i) = 2/5$ holds for every $j$.

Now consider $2n$ allocations $\mathcal{A}^1, \mathcal{A}^2, \ldots, \mathcal{A}^{2n}$ such that 
$$\mathcal{A}^{j} = \langle A^j_1, A^j_2 \ldots, A^j_n\rangle$$ for every $1 \leq j \leq 2n$ where 

$$A^j_k =
\begin{cases}
A_k \cup S_j, & \text{if }k = i \\
A_k \setminus S_j, & \text{if }k \neq i.
\end{cases}$$
We show at least one of these allocations has a higher for $\mathsf{ex}^{2/5}$ than $\mathcal{A}$.
Since $\valu^{2/5}_i$ is XOS, by Lemma \ref{xos2lemma} we have
\begin{equation*}
\sum_{j=1}^{2n}\Big(\valu^{2/5}_k(A_k) - \valu^{2/5}_k(A_k \setminus S_j)\Big) \leq \valu^{2/5}_k(A_j)
\end{equation*}
for every $\agent_k \neq \agent_i$ and thus
\begin{equation}\label{gulu1}
\begin{split}
\sum_{j=1}^{2n} \valu^{2/5}_k(A^j_k) &= \sum_{j=1}^{2n} \valu^{2/5}_k(A_j \setminus S_j)\\
 &\geq 2n \valu^{2/5}_k(A_k) - \valu^{2/5}_k(A_k)\\
&= (2n-1)\valu^{2/5}_k(A_k)
\end{split}
\end{equation}
Moreover, since $\valu^{2/5}_i(A_i) < 1/5$, we have
\begin{equation}\label{gulu2}
\begin{split}
\sum_{\agent_j \neq \agent_i} \ceil{\valu_j}{2/5}(A_j) & > \sum_{\agent_j \in \agents} \ceil{\valu_j}{2/5}(A_j) - 1/5\\
&=\mathsf{ex}^{2/5}(\mathcal{A})-1/5.
\end{split}
\end{equation}
Furthermore, since $\valu_i^{2/5}(S_j \cup A_i) = 2/5$ for every $1 \leq j \leq 2n$, we have
\begin{equation}\label{gulu3}
\begin{split}
\sum_{\agent_k \neq \agent_i} \valu^{2/5}_k(A^j_k) &= \sum_{\agent_k \in \agents} \valu^{2/5}_k(A^j_k) - 2/5\\
&= \mathsf{ex}^{2/5}(\mathcal{A}^j) - 2/5
\end{split}
\end{equation}
Finally, by combining Inequalities \eqref{gulu1}, \eqref{gulu2}, and \eqref{gulu3} we have
\begin{equation*}
	\begin{split}
	\sum_{j=1}^{2n} \mathsf{ex}^{2/5}(\mathcal{A}^j) & = \sum_{j=1}^{2n}(2/5 +\sum_{\agent_k \neq \agent_i} \valu^{2/5}_k(A^j_k)) \\
	& = 4n/5 + \sum_{j=1}^{2n}\sum_{\agent_k \neq \agent_i} \valu^{2/5}_k(A^j_k)\\  
	&\geq 4n/5 + \sum_{\agent_k \neq \agent_i} (2n-1) \valu^{2/5}_k(A_k)\\
	&\geq 4n/5 + (2n-1)(\mathsf{ex}^{2/5}(\mathcal{A}) - 1/5)\\
	&\geq 2n \cdot \mathsf{ex}^{2/5}(\mathcal{A}) + (4n-2n+1)/5 - \mathsf{ex}^{2/5}(\mathcal{A})\\
	&\geq 2n \cdot \mathsf{ex}^{2/5}(\mathcal{A}) + (2n+1)/5 - \mathsf{ex}^{2/5}(\mathcal{A})
	\end{split} 
\end{equation*}
Now notice that since $\valu^{2/5}_k(A_k) \leq 2/5$, we have
\begin{equation*}
	\begin{split}
	\mathsf{ex}^{2/5}(\mathcal{A}) & = \sum_{k=1}^n \valu^{2/5}_k(A_k)\\
	& \leq \sum_{k=1}^n 2/5\\
	& \leq 2n/5. 
	\end{split}
\end{equation*}
and thus 
\begin{equation*}
	\begin{split}
	\sum_{j=1}^{2n} \mathsf{ex}^{2/5}(\mathcal{A}^j) & \geq 2n \cdot \mathsf{ex}^{2/5}(\mathcal{A}) + (2n+1)/5 - \mathsf{ex}^{2/5}(\mathcal{A})\\
	&\geq 2n \cdot \mathsf{ex}^{2/5}(\mathcal{A}) + (2n+1)/5 - 2n/5\\
	&\geq 2n \cdot \mathsf{ex}^{2/5}(\mathcal{A}) + 1/5.
	\end{split}
\end{equation*}
Therefore, $\mathsf{ex}^{2/5}(\mathcal{A}^j) > \mathsf{ex}^{2/5}(\mathcal{A})+1/10n$ holds for at least one $\mathcal{A}^j$ which contradicts the maximality of $\mathcal{A}$.
\end{proof}

\subsection{Algorithm}\label{xosalgorithm}
In this section we provide a polynomial time algorithm for finding a $1/8$-$\MMS$ allocation for the fair allocation problem with XOS agents. The algorithm is based on a similar idea that we argued for the proof of Theorem \ref{xosproof}. Remark that our algorithm only requires access to demand and XOS oracles. It does \textit{not} have any additional information about the maxmin values. This makes the problem computationally harder since computing the maxmin values is NP-hard~\cite{epstein2014efficient}. We begin by giving a high-level intuition of the algorithm and show the computational obstacles can be overcome by combinatorial tricks. Consider the pseudo-code described in Algorithm \ref{xosalg}.
\begin{algorithm}
	\KwData{$\agents, \items, \langle \valu_1, \valu_2, \ldots, \valu_n\rangle$}
	For every $\agent_j$, scale $\valu_j$ to  ensure $\MMS_j = 1$\;
	\While{there exist an agent $\agent_i$ and an item $\ite_j$ such that $\valu_i(\{\ite_j\}) \geq 1/8$}{
		Allocate $\{\ite_j\}$ to $\agent_i$\;
		$\items = \items \setminus \ite_j$\;
		$\agents = \agents \setminus \agent_i$\;
	}
	$\mathcal{A} = $ an arbitrary allocation of the items to the agents\;
	\While{$\min \valu^{1/4}_j(A_j) < 1/8$}{
		$i = $ the agent who receives the lowest value in allocation $\mathcal{A}$\; 
		Find a set $S$ such that:\label{find}
		$\mathsf{ex}^{1/4}(\langle A_1 \setminus S, A_2 \setminus S, \ldots, A_{i-1} \setminus S, A_i \cup S, A_{i+1} \setminus S, \ldots, A_n \setminus S\rangle) \geq \mathsf{ex}^{1/4}(\mathcal{A})+1/12n$\;
		$\mathcal{A} = \langle A_1 \setminus S, A_2 \setminus S, \ldots, A_{i-1} \setminus S, A_i \cup S, A_{i+1} \setminus S, \ldots, A_n \setminus S\rangle$\;
	}
	For every $\agent_i \in \agents$ allocate $A_i$ to $\agent_i$\;
\caption{Algorithm for finding a $1/8$-$\MMS$ allocation}
\label{xosalg}
\end{algorithm}

As we show in Section \ref{former}, Command \ref{find} of the algorithm is always doable. More precisely, there always exists a set $S$ that holds in the condition of Command \ref{find}. Notice that in every step of the algorithm, $\mathsf{ex}^{1/4}(\mathcal{A})$ is increased by at least $1/12n$ and this value is bounded by $1/4\cdot n = n/4$. Therefore the algorithm terminates after at most $3n^2$ steps and the allocation is guaranteed to be $1/8$-$\MMS$.

That said, there are two major computational obstacles in the way of running Algorithm \ref{xosalg}. Firstly, finding a set $S$ that holds in the condition of Command \ref{find} can not be trivially done in polynomial time. Second, scaling the valuation functions to ensure $\MMS_i = 1$ for all agents is NP-hard and cannot be done in polynomial time unless P=NP. To overcome the former, in Section \ref{former} we provide an algorithm for finding such a set $S$ in polynomial time. Next, in Section \ref{latter}, we present a combinatorial trick to run the algorithm in polynomial time without having to deal with NP-hardness of scaling the valuation functions.

\subsubsection{Executing Command \ref{find} in Polynomial Time}\label{former}
In this section we present an algorithm to execute Command \ref{find} of Algorithm \ref{xosalg}. We show that such a procedure can be implemented via demand oracles.

Let for every $\ite_j \notin A_i$, $c_j$ be the amount of contribution that $\ite_j$ makes to $\mathsf{ex}^{1/4}(\mathcal{A})$. We set $p_e = 3(n/(n-1))c_e$ and ask the demand oracle of $\valu_i$ to find a set $S$ that maximizes $\valu_i(S) - \sum_{\ite_j \in S} p_j$. Via a trivial calculation, one can show that $\valu_i(S) - \sum_{\ite_j \in S} p_j \geq 1/4$ holds for at least one set of items. The reason this is correct is that one can divide the items into $n$ partitions where each is worth at least $1$ to $\agent_i$. Moreover, the summation of prices for the items is bounded by $3n/(n-1)\cdot (\sum_{j\neq i} \valu_j^{1/4} (A_j) )\leq 3n/4$. Therefore, for at least one of those partitions  $\valu_i(S) - \sum_{\ite_j \in S} p_j$ is at least $1/4$. Thus, the set that the oracle reports is worth at least $1/4$ to $\agent_i$. 

Now, let $S^*$ be the set that the oracle reports and for every $\ite_j \in S^*$, $c^*_j$ be the contribution of $\ite_j$ to $\valu_i(S^*)$. We sort the items of $S^*$ based on $c^*_j - p_j$ in non-increasing order. Next, we start with an empty bag and add the items in their order to the bag until the total value of the items in the bag to $\agent_i$ reaches $1/4$. Since the value of every item alone is bounded by $1/8$, the total value of the items in the bag to $\agent_i$ is bounded by $3/8$. Thus the contribution of those items to $\mathsf{ex}^{1/4}(\mathcal{A})$ is at most $(3/8) / (3n/(n-1)) \leq 1/8 - 1/(10n)$. Therefore, removing items of the bag from other allocations and adding them to $A_i$, increases $\mathsf{ex}^{1/4}(\mathcal{A})$ by at least $1/10n$.

Remark that one can use the same argument to prove this even if $\MMS_i \geq 1/(1+1/10n)$.

\subsubsection{Running Algorithm \ref{xosalg} in Polynomial Time}\label{latter}
As aforementioned, scaling valuation functions to ensure $\MMS_i = 1$ for every agent $\agent_i$ is an NP-hard problem since determining the maxmin values is hard even for additive agents~\cite{Procaccia:first}. Therefore, unlike Section \ref{former}, in this section we massage the algorithm to make it executable in polynomial time.

Suppose an oracle gives us the maxnmin values of the agents. Provided that we can run Command \ref{find} of Algorithm \ref{xosalg} in polynomial time, we can find a $1/8$-$\MMS$ allocation in polynomial time. Therefore, in case the oracle reports the actual maxmin values, the solution is trivial. However, what if the oracle has an error in its calculations? There are two possibilities: (i) Algorithm \ref{xosalg} terminates and finds an allocation which is $1/8$-$\MMS$ with respect to the reported maxmin values. (ii) The algorithm fails to execute Command \ref{find}, since no such set $S$ holds in the condition of Command \ref{find}. The intellectual merit of this section boils down to investigation of the case when algorithm fails to execute Command \ref{find}. We show, this only happens due to an overly high misrepresentation of the maxmin value for agent $\agent_i$. Note that $\agent_i$ is the agent who receives the lowest value in the last cycle of the execution.

\begin{observation}\label{beautiful}
	Given $\langle d_1, d_2, \ldots, d_n\rangle$ as an estimate for the maxmin values, if Algorithm \ref{xosalg} fails to execute Command \ref{find} for an agent $\agent_i$, then we have
	$$d_i \geq (1+1/10n) \MMS_i.$$
\end{observation} 
Proof of Observation \ref{beautiful} follows from the argument of Section \ref{former}. More precisely, as mentioned in Section \ref{former}, such a set $S$ exists, if $\MMS_i \geq 1/(1+1/10n)$. Thus, given that the procedure explained in Section \ref{former} fails to find such a set, one can conclude the the reported value for $\MMS_i$ is at least $(1/(1+1/10n))$ times its actual value. Based on Observation \ref{beautiful}, we propose Algorithm \ref{oracle} for implementing a maxmin oracle.

\begin{algorithm}
	\KwData{$\agents, \items, \langle \valu_1, \valu_2, \ldots, \valu_n\rangle$}
	\For {every $\agent_i \in \agents$}{
		$d_i \leftarrow \valu_i(\items)$\;
	}
	\While{true}{
		Run Algorithm \ref{xosalg} assuming maxmin values are $d_1, d_2, \ldots, d_n$\;
		\If {the Algorithm fails to run Command \ref{find} for an agent $\agent_i$}{
			$d_i \leftarrow d_i / (1+1/10n)$\; 
		}\Else{
			Report the allocation and terminate the algorithm\;
		}
	}
	\caption{Implementing a maxmin oracle}
	\label{oracle}
\end{algorithm}

Note that in the beginning of the algorithm, we set $d_i = \valu_i(\items)$ which is indeed greater than or equal to $\MMS_i$. By Lemma \ref{beautiful}, every time we decrease the value of $d_i$ for an agent $\agent_i$, we preserve the condition $d_i \geq \MMS_i$ for that agent. Therefore, in every step of the algorithm, we have $d_i \geq \MMS_i$ and thus the reported allocation which is $1/8$-$\MMS$ with respect to $d_i$'s is also $1/8$-$\MMS$ with respect to true maxmin values. Thus, the algorithm provides a correct $1/8$-$\MMS$ allocation in the end. All that remains is to show the running time of the algorithm is polynomial.

Notice that every time we decrease $d_i$ for an agent $\agent_i$, we multiply this value by $1/(1+1/10n)$, hence the number of such iterations is polynomial in $n$, unless the valuations are super-exponential in $n$. Since we always assume the input numbers are represented by $\mathsf{poly}(n)$ bits, the number of iterations is bounded by $\mathsf{poly}(n)$ and hence the algorithm terminates after a polynomial number of steps.

\begin{theorem}\label{xa}
	Given access to demand and XOS oracles, there exists a polynomial time algorithm that finds a $1/8$-$\MMS$ allocation for XOS agents.
\end{theorem}

An elegant consequence of Theorem \ref{xa} is a $8$-approximation algorithm for determining the maxmin value of an XOS function with $r$ partitions.
\begin{corollary}
	Given an XOS function $f$, an integer number $r$, and access to demand and XOS oracles of $f$, there exists a $8$-approximation polytime algorithm for determining $\MMS^r_f$.
\end{corollary}
\begin{proof}
	We construct an instance of the fair allocation problem with $r$ agents, all of whom have a valuation function equal to $f$. We find a $1/8$-$\MMS$ allocation of the items to the agents in polynomial time and report the minimum value that an agent receives as output.
	
	The $1/8$ guarantee follows from the fact that every agent receives a subset of values that are worth $1/8$-$\MMS_i$ to him, and since $\MMS_i$ is exactly equal to $\MMS^r_f$, every partition has a value of at least $\MMS^r_f/8$.
\end{proof}
\begin{remark}
	A similar procedure can also be used to overcome the challenge of computing the maxmin values for the algorithm described in Section \ref{submodularalgorithm}.
\end{remark}
\section{Subadditive Agents}\label{subadditive}
In this section we present a reduction from subadditive agents to XOS agents. More precisely, we show for every subadditive set function $f(.)$, there exists an XOS function $g(.)$, where $g$ is dominated by $f$ but the maxmin value of $g$ is within a logarithmic factor of the maxmin value of $f$. We begin by an observation. Suppose we are given a subadditive function $f$ on set $\domp(f)$, and we wish to approximate $f$ with an additive function $g$ which is dominated by $f$. In other words, we wish to find an additive function $g$ such that 
$$\forall S \subseteq \domp(f) \hspace{1cm} g(S) \leq f(S)$$
and $g(\domp(f))$ is maximized. One way to formulate $g$ is via a linear program. Suppose $\domp(f)=\{\ite_1,\ite_2,\ldots,\ite_m\}$ and let $g_1, g_2, \ldots, g_m$ be $m$ variables that describe $g$ in the following way:
$$\forall S \subseteq \domp(f) \hspace{1cm} g(S) = \sum_{\ite_i \in S} g_i.$$
Based on this formulation, we can find the optimal additive function $g$ by LP \ref{lp1}.
\begin{alignat}{3}\label{lp1}
\text{maximize: }& \hspace{0.5cm} &  \sum_{\ite_i \in \domp(f)}g_i & &\\
\text{subject to: }& & \sum_{\ite_i \in S}g_i \leq f(S) & \hspace{1cm}&\forall S \subseteq \domp(f)\nonumber\\
& & g_i \geq 0 & &\forall \ite_i \in \domp(f)\nonumber
\end{alignat}
We show the objective function of LP \ref{lp1} is lower bounded by $f(\domp(f)) / \log m$. The basic idea is to first write the dual program and then based on a probabilistic method, lower bound the optimal value of the dual program by $f(\domp(f))/ \log m$. 
\begin{lemma}\label{jj1}
	The optimal solution of LP \ref{lp1} is at least $f(\domp(f))/ \log m$.
\end{lemma}
\begin{proof}
To prove the lemma, we write the dual of LP \ref{lp1} as follows:
\begin{alignat}{3}\label{lp2}
\text{minimize: }& \hspace{0.5cm}& \sum_{S \subseteq \domp(f)} \alpha_S f(S)   & &\\
\text{subject to: }& & \sum_{S \ni \ite_i} \alpha_S \geq 1 & \hspace{1cm}&\forall \ite_i \in \domp(f)\nonumber\\
& & \alpha_S \geq 0 & &\forall S \subseteq \domp(f)\nonumber
\end{alignat}
By the strong duality theorem, the optimal solutions of LP \ref{lp1} and LP \ref{lp2} are equal~\cite{bachem1992linear}. Next, based on the optimal solution of LP \ref{lp2}, we define a randomized procedure to draw a set of elements: We start with an empty set $S^*$ and for every set $S \subseteq \domp(f)$ we add \textit{all} elements of $S$ to $S^*$ with probability $\alpha_S$. Since $f$ is subadditive, the marginal increase of $f(S^*)$ by adding elements of a set $S$ to $S^*$ is bounded by $f(S)$ and thus the expected value of $f(S^*)$ is bounded by the objective of LP \ref{lp2}. In other words:
\begin{equation}\label{jef0}
\mathbb{E}[f(S^*)] \leq \sum_{S \subseteq \domp(f)} \alpha_S f(S)
\end{equation}
Remark that we repeat this procedure for all subsets of $\domp(S)$ independently and thus for every $\ite_i \in \domp(f)$, $\sum_{S \ni \ite_i} \alpha_S \geq 1$ holds we have
\begin{equation}\label{jef1}
\mathsf{PR}[\ite_i \in S^*] \geq 1-1/e \simeq 0.632121 > 1/2
\end{equation}
for every element $\ite_i \in \domp(s)$. Now, with the same procedure, we draw $\lceil \log m \rceil + 2$ sets $S^*_1, S^*_2, \ldots, S^*_{\lceil \log m \rceil + 2}$ \textit{independently}. We define $\hat{S} = \bigcup S^*_i$. By Inequality \eqref{jef1} and the union bound we show
\begin{equation*}
\begin{split}
\mathsf{PR}[\hat{S} = \domp(f)] & \geq 1- \sum_{\ite_i \in \domp(i)} \mathsf{PR}[\ite_i \notin \hat{S}]\\
& = 1- \sum_{\ite_i \in \domp(i)} \mathsf{PR}[\ite_i \notin S^*_1 \text{ and } \ite_i \notin S^*_1 \text{ and } \ldots \text{ and } \ite_i \notin S^*_{\lceil \log m \rceil + 2}]\\
& = 1- \sum_{\ite_i \in \domp(i)} \prod_{j=1}^{\lceil \log m \rceil + 2} \mathsf{PR}[\ite_i \notin S^*_j]\\
& \geq 1- \sum_{\ite_i \in \domp(i)} \prod_{j=1}^{\lceil \log m \rceil + 2} 1/2\\
& = 1- \sum_{\ite_i \in \domp(i)} \prod_{j=1}^{\lceil \log m \rceil + 2} \mathsf{PR}[\ite_i \notin S^*_j]\\
& \geq 1- \sum_{\ite_i \in \domp(i)} 1/4m\\
& = 1- 1/4\\
& = 3/4\\
\end{split}
\end{equation*}
and thus $\mathbb{E}[f(\hat{S})] \geq 3/4 f(\domp(f))$. On the other hand, by the linearity of expectation and the fact that $f$ is subadditive we have:
\begin{equation*}
\begin{split}
\mathbb{E}[f(\hat{S})] &= \mathbb{E}[f(\bigcup S^*_i)]\\
& \leq \mathbb{E}[\sum f(S^*_i)]\\
& \leq (\lceil \log m \rceil + 2) (\sum_{S \subseteq \domp(f)} \alpha_S f(S))
\end{split}
\end{equation*}
Therefore $\sum_{S \subseteq \domp(f)} \alpha_S f(S) \geq 3/4 f(\domp(f)) / (\lceil \log m \rceil + 2)$, which means $$\sum_{S \subseteq \domp(f)} \alpha_S f(S) \geq f(\domp(f)) / (2\lceil \log m \rceil)$$ for big enough $m$. This shows the optimal solution of LP \ref{lp1} is lower bounded by $f(\domp(f)) / (2\lceil \log m \rceil)$ and the proof is complete.
\end{proof}

In what follows, based on Lemma \ref{jj1}, we provide a reduction from subadditive agents to XOS agents. An immediate corollary of Lemma \ref{jj1} is the following:
\begin{corollary}[of Lemma \ref{jj1}]\label{kk}
For any subadditive function $f$ and integer number $n$, there exists an XOS function $g$ such that
$$g(S) \leq f(S) \qquad \forall S \subseteq \domp(f)$$
and 
$$\MMS_g^n \geq \MMS_f^n/2\lceil \log n \rceil.$$
\end{corollary} 
\begin{proof}
	By definition, we can divide the items into $n$ disjoint sets such that the value of $f$ for every set is at least $\MMS_f^n$. Now, based on Lemma \ref{jj1}, we approximate $f$ for each set with an additive function $g_i$ wile losing a factor of at most $\lceil 2 |\log \domp(f)|\rceil$ and finally we set $g = \max g_i$. Based on Lemma \ref{jj1}, both conditions of this lemma are satisfied by $g$.
\end{proof}

Based on Theorem \ref{xosproof} and Lemma \ref{kk} one can show that a $1/10\lceil \log m \rceil$-$\MMS$ allocation is always possible for subadditive agents.
\begin{theorem}\label{subadditiveproof}
	The fair allocation problem with subadditive agents admits a $1/10\lceil \log m \rceil$-$\MMS$ allocation.
\end{theorem}
\section{Acknoledgment}
We would like to thank the anonymouse reviewers for their thoughtful comments and direction.

\bibliographystyle{abbrv}
	
\bibliography{frugal}
\newpage

\appendix
\section{A $4/5$-$\MMS$ Allocation for Four Agents}\label{45}
In this section we propose an algorithm to find a $4/5$-$\MMS$ for $n=4$ in the additive setting. Since the number of agents is exactly 4, we assume $\agents = \{\agent_1, \agent_2, \agent_3, \agent_4\}$. Again, for simplicity we assume $\MMS_i = 1$ for every $\agent_i \in \agents$. In general, our algorithm is consisted of  three main steps: first, $\agent_1$ optimally partitions the items into $4$ bundles with values at least 1 to him. Next, $\agent_2$ selects three of the bundles and repartitions them. Finally, we satisfy one of $\agent_1$ or $\agent_2$ with a bundle and solve the problem for remaining agents and items via Lemma ~\ref{recurse}. Without loss of generality, we assume the valuation of every agent $\agent_i$ for every bundle in his optimal $n$-partitioning is exactly equal to 1. Therefore, from here on, we assume that the summation of the values of the items within the same bundle for every agent is at most $1$. In addition, we suppose that the problem is $4/5$-irreducible, since Lemma \ref{reducibility} narrows down the problem into such instances. Thus, by Lemma \ref{remove1}, the value of every item is less than $4/5$ to any agent. 
We begin this section by stating a number of definitions and observations. In this section, we use the term \emph{bundle} to refer to a set of items.
\begin{definition}
\label{perfect}
A set $S$ of bundles is perfect for a set $T$ of agents, if $(i)$ $|S| = |T|$ and $(ii)$ there exists an allocation of the bundles in $S$ to the agents of $T$ such that all the agents in $T$ are satisfied by their allocated bundle.
\end{definition}

\begin{observation}\label{firstobs}
\label{sum}
Let $\agent_i$ be an agent and $S$ be a set of items where $\valu_i(\{\ite_j\}) \leq v$ for every item $\ite_j \in S$. If $V_i(S) > v$, then there exists a subset $S' \subseteq S$ of items such that $ v \leq \valu_i(S') < 2v$.
\end{observation}
\begin{proof}
We begin with an empty set $S'$ and add the items of $S$ to $S'$ one by one, until $\valu_i(S')$ exceeds $v$. Before adding the last element to $S'$, the valuation of $\agent_i$ for $S'$ was no more than $v$ and every item alone is of value less than $v$ to $\agent_i$. Therefore, after adding the last item to $S’$, its value is less than $2v$ to $\agent_i$.
\end{proof}

\begin{definition}
\label{core}
For a bundle $B$ of items that satisfies $\agent_i$, the core of $B$ with respect to agent $\agent_i$, denoted by $C_i(B)$, is defined as follows: let $m_1,m_2,..,m_k$ be the items of $B$ in the increasing order of their values to $\agent_i$. Then $C_i(B) = \{m_j,m_{j+1},...,m_{k}\}$ , where $j$ is the highest index, such that set of items $\{m_j,m_{j+1},...,m_k\}$ satisfies $\agent_i$.
\end{definition}

Note that for every subset $B$ with $\valu_i(B) \geq 4/5$, $C_i(B)$ is a subset of $B$ with the minimum size that satisfies $\agent_i$. Since the items in $C_i(B)$ satisfy $\agent_i$, we have $\valu_i(C_i(B))\geq 4/5$. On the other hand, by the fact that $|C_i(B)|$ is minimal, removing any item from $C_i(B)$ results in a subset that no longer satisfies $\agent_i$. Thus, Observation \ref{eps2} holds.

\begin{observation}\label{eps2} 
If $\valu_i(C_i(B))  = 4/5 + \beta$, then the value of every item in $C_i(B)$ is more than $\beta$ for $\agent_i$.
\end{observation}

By the fact the value of every item in $B$ is less than $4/5$ we have Observation \ref{cre}.

\begin{observation}\label{cre}
For every agent $\agent_i$ and any subset $B$ of items, $\valu_i(C_i(B))<8/5$.
\end{observation}

\begin{lemma}
\label{3p}
Suppose that $S = \{X,Y,Z\}$ is a 3-partitioning of a set of items with the following properties for an agent $\agent_i$: 

\begin{minipage}[t]{\linegoal}
\begin{enumerate}[leftmargin=*]
\item $\valu_i(X) < 4/5$ and $\valu_i(Y)<4/5$. 

\item $V_i(X\cup Y \cup Z) >16/5$. 
\end{enumerate}
\end{minipage}
\\

Then we can move some items from $Z$ to an arbitrary bundle of $\{X, Y\}$, such that, both $Z$ and the corresponding bundle will be worth at least $4/5$ to $\agent_i$. 

\end{lemma}

\begin{proof}
Since $V_i(X \cup Y \cup Z ) > 16/5$, $V_i(Z) > 8/5$ holds. Moreover, by Observation \ref{cre}, we have $\valu_i(C_i(Z)) \leq 8/5$. Considering $Z' = Z \setminus C_i(Z)$, we have   
$$\valu_i(X \cup Y \cup Z' ) \geq 8/5.$$
According to the fact that $\valu_i(X)<4/8$ and $\valu_i(Y)<4/5$, we have 
$$\valu_i(X \cup Z' ) \geq 4/5$$
and 
$$\valu_i(Y \cup Z' ) \geq 4/5.$$

\end{proof}

\begin{lemma}
\label{recurse}
Let $S= \{X,Y,Z\}$ be a set of three bundles of items, such that 

\begin{minipage}[t]{\linegoal}
\begin{enumerate}[leftmargin=*]
 \item $V_1(X) \ge 4/5 , V_1(Y) \ge 4/5, V_1(Z) \ge 4/5 $. 
\item $V_2(X \cup Y \cup Z) > 16/5$.
\item $V_3(X \cup Y \cup Z) \ge 3 $.
 
\end{enumerate}
\end{minipage}
\\

Then a $4/5$-$\MMS$ allocation of $X \cup Y \cup Z$ to the agents $\agent_1,\agent_2,\agent_3$ is possible.
\end{lemma}
\begin{proof}
If $\agent_2$ can be satisfied with two different bundles, then trivially $S$ is perfect. Otherwise, $\agent_2$ is satisfied with only one bundle, say $Z$. By Lemma ~\ref{3p}, $\agent_2$ can transfer some items from $Z$ to $Y$, such that both bundles satisfy him.  After moving the items, both $Y$ and $Z$ satisfy $\agent_2$, and bundle $X$ and $Y$, satisfy $\agent_1$. One the other hand,   since $V_3(X \cup Y \cup Z) \ge 3 $, $\agent_3$ is satisfied with at least one bundle. Its easy to observe that for any valuation of bundles $X,Y,Z$ for $\agent_3$, the set of bundles is perfect.
\end{proof}

\begin{lemma}
\label{core}
Let $S=\{X,Y,Z,T\}$ be a 4-partitioning of $\cal M$ and $\agent_i$ be an arbitrary agent. Then $\agent_i$ can select $3$ bundles and re-partition them into three new bundles in such a way that each bundle will be worth at least $4/5$ to $\agent_i$.
\end{lemma}
\begin{proof}
Consider bundles $X,Y,Z,T$. If more than two of them satisfy $\agent_i$, then the selection is trivial. Furthermore, if only one bundle satisfies $\agent_i$, then by Lemma ~\ref{3p}, we can move some items from the satisfying bundle to another bundle, such that both bundles satisfy $\agent_i$. Thus, without loss of generality, we assume that bundles $Z$ and $T$ satisfy $\agent_i$.

Let $Z' = Z \setminus C_i(Z)$ and $T'= T \setminus C_i(T)$. Without loss of generality, we assume $V_i(X) \geq V_i(Y)$ and let $X' = X \cup Z' \cup T'$. If $V_i(X') \ge 4/5$, then the proof is trivial. Thus, suppose that $V_i(X') < 4/5$.

Consider the value of bundles as $$V_i(X’)=4/5 - \epsilon_1,$$ $$V_i(Y)=4/5 - \epsilon_2,$$ $$V_i(C_i(Z))=6/5 + \epsilon_3,$$ and $$V_i(C_i(T))=6/5 + \epsilon_4$$ where $\epsilon_1 \leq \epsilon_2$ and $\epsilon_3 \leq \epsilon 4$. Note that $\epsilon_3$ can be negative. By the fact that total value of the items equals $4$ for all of the agents, we assume that $$\epsilon_1 + \epsilon_2 = \epsilon_3 + \epsilon_4$$ and hence $$\epsilon_1 \leq \epsilon_4.$$

Now, we explore the properties of the items in $C_i(T)$. Regarding Observation \ref{eps2}, every item in $C_i(T)$ is worth more than $2/5+\epsilon_4$ to $\agent_i$. Hence, $C_i(T)$ cannot contain more than 2 items, since value of every pair of items in $C_i(T)$ is more than $4/5$. Moreover, $C_i(T)$ cannot contain one item and hence, $C_i(T)$ contains exactly two items. Let $\ite_1$ and $\ite_2$ be these two items. Since $\valu_i(T) = 6/5 + \epsilon_4$, at least one of these two items, say $\ite_2$ is worth at least $3/5+\epsilon_4/2$ to $\agent_i$. Thus, in summary, $C_i(T)$ contains two items $\ite_1$ and $\ite_2$ with

$$\valu_i(\{\ite_1\}) >2/5 + \epsilon_4,$$
$$\valu_i(\{\ite_2\}) \geq 3/5+\epsilon_4/2.$$

Next, we characterise the items in $X'$. For Bundle $X'$, let $B$ be the set of items with a value less than $1/5 - \epsilon_4/2$. If $V_i(B)\geq 1/5 - \epsilon_4/2$, then Observation ~\ref{sum} states that there exists a subset $B'$ of $B$, such that:
$$1/5 - \epsilon_4/2 \leq V_i(B') < 2/5 - \epsilon_4.$$

Therefore, Bundles $B' \cup \{b_2\}$ and $(X' \setminus B') \cup \{b_1\}$ satisfy $\agent_i$. These two bundles together with  $C_i(Z)$ result in three bundles that satisfy $\agent_i$. Thus, $V_i(B) < 1/5 - \epsilon_4/2$. 

Finally, regarding the fact that $V_i(B) < 1/5 - \epsilon_4/2$, we have 
$$\valu_i(X' \setminus B)> 3/5  -\epsilon_1 + \epsilon_4/2 $$

For this case, we show that $X'\setminus B$ contains exactly one item. Otherwise, at least one of these items, say $\ite_3$, is worth less than $3/10 - \epsilon_1/2 + \epsilon_4/4$ and therefore, for the bundles $\{\ite_3\} \cup \{\ite_2\}$ and $(X' \setminus \{\ite_3\}) \cup \{\ite_1\}$  we have:
$$\valu_i(\{\ite_3\} \cup \{\ite_2\}) \geq 1/5 - \epsilon_4/2 + 3/5+\epsilon_4/2 \geq 4/5$$
and 
$$\valu_i((X' \setminus \{\ite_3\}) \cup \{\ite_1\}) \geq 4/5 - \epsilon_1 - 3/10 + \epsilon_1/2- \epsilon_4/4+ 2/5 + \epsilon_4 \geq 4/5$$
respectively. These two bundles along with $C_i(Z)$ form $3$ bundles that satisfy $\agent_i$. Therefore, we conclude that $X' \setminus B$ contains an item $\ite_3$ with a value more than $3/5-\epsilon_1  + \epsilon_4/2$ to $\agent_i$. The rest of the items in $X'$ belong to $B$ that are in total worth less than $1/5-\epsilon_4/2$.

Note that $X \subseteq X'$. Therefore, consider the $4-$partitioning of $\agent_i$ and remove the bundle containing $\ite_3$. Also, remove the items with value less than $1/5 - \epsilon_4/2$ to $\agent_i$ in $X$, from their corresponding bundles. Three bundles with value of each to $\agent_i$ more than $$1 - 1/5 + \epsilon_4/2 \geq 4/5,$$ with all of their items from $Y,Z$ and $T$ remain. Thus, $\agent_i$ can make three satisfying bundles with items in $Y,Z,T$. 
\end{proof}
Based on what we showed so far, we prove Theorem \ref{45main}.
\begin{theorem}
\label{45main}
A $4/5$-$\MMS$ allocation for $n=4$ is possible in the additive setting.
\end{theorem}
\begin{proof}
Consider the optimal $4$-partitioning of $\cal M$ with respect to $\agent_1$. Now,  ask $\agent_2$ to select $3$ bundles and re-partition them, such that he can be satisfied with all the three bundles. Based on Lemma \ref{core}, such a repartitioning is always possible. Due to the Pigeonhole principle, at least one of these three bundles still satisfies $\agent_1$. Let $S = \{X,Y,Z,T\}$ be the resulting bundles and without loss of generality, suppose that bundles $X,Y$ satisfy $\agent_1$ and bundles $Y,Z,T$ satisfy $\agent_2$.

Now, consider agents $\agent_3$ and $\agent_4$ and let $\phi$ be the set of bundles that satisfy $\agent_3$ or $\agent_4$. There are only two cases, in which $S$ is not perfect 
(recall definition \ref{perfect}):

\begin{enumerate}
\item $\phi \subseteq \{X,Y\}$ : $\agent_1$, selects three bundles $X,Y$ and one of $Z$ or $T$, say $Z$ and re-partitions them to three satisfying bundles. Now, give bundle $T$ to $\agent_2$. According to Lemma  ~\ref{recurse}, items of $X,Y,Z$ can satisfy the remaining three agents.

\item $|\phi|=1, \phi \notin \{X,Y\}$ : give $X$ to $\agent_1$ and allocate the items of $Y \cup Z \cup T$ to $\agent_2,\agent_3,\agent_4$, using Lemma ~\ref{recurse}.
\end{enumerate}
\end{proof}
\color{black}
\section{Omitted Proofs of Section \ref{additive:observations}}\label{additiveobservationsproof}

\begin{proof}[of Lemma \ref{remove1}]
The key idea is that given $\MMS_i \geq 1$ for an agent $\agent_i$, then for every item $\ite_j \in \items$ we have 
$\MMS_i^{n-1}(\items\setminus\ite_j) \geq 1$. This holds since removing an item from $\items$ will diminish the value of at most one partition in the optimal $n$ partitioning of the items. Therefore, at least $n-1$ partitions have a value of $1$ or more to $\agent_i$ and thus $\MMS_i^{n-1}(\items\setminus\ite_j) \geq 1$.
The rest of the proof follows from the definition of $\alpha$-irreducibility. If the valuation of an item $\ite_j$ to an agent $\agent_i$ is at least $\alpha$, then the problem is $\alpha$-reducible since if we allocate $\ite_j$ to $\agent_i$, we have 
$$\MMS_{\valu_k}^{n-1}(\items \setminus \{\ite_j\}) \geq 1$$
for every agent $\agent_k \neq \agent_i$. This contradicts with the $\alpha$-irreducibility assumption.
\end{proof}

\begin{proof}[of Lemma \ref{remove2}]
Suppose for the sake of contradiction that for every agent $\agent_{i'} \neq \agent_i$ we have $\valu_{i'}(\{\ite_j,\ite_k\}) \leq 1$. By this assumption, we show 
\begin{equation}
 \MMS_{i'}^{n-1}(\items\setminus \{\ite_j,\ite_k\}) \geq 1 \label{saeed1}
\end{equation} holds. This is true since removing two items $\ite_j$ and $\ite_k$ from $\items$ decreases the value of at most two partitions of the optimal partitioning of $\items$ for $\MMS_{i'}$. If $n-1$ partitions remain intact, then Inequality \eqref{saeed1} trivially holds. If not, merging the two partitions that initially contained $\ite_j$ and $\ite_k$ results in a partition with value at least $1$ to $\agent_i$. This partition together with the $n-2$ remaining partitions result in a desirable partitioning of $\items$ into $n-1$ partitions. Therefore, Inequality \eqref{saeed1} holds for any agent $\agent_{i'}$, and this implies that by allocating $S = \{\ite_j, \ite_k\}$ to $\agent_i$, not only does $\valu_i(S) \geq 3/4$ hold, but also for every $\agent_{i'} \neq \agent_i$ we have 
$$\MMS_{i'}^{n-1}(\items\setminus \{\ite_j,\ite_k\}) \geq 1$$
which means the problem is $3/4$-reducible, and it contradicts our assumption.
\end{proof}

\begin{proof}[of Lemma \ref{remove3}]
The proof for this lemma is obtained by applying Lemma \ref{remove2}, $|T|$ times. Consider an agent $\agent_i \notin T$. According to the argument in Lemma \ref{remove2}, if we assign $\ite_{{j_1}}$ and $\ite_{{j_2}}$ to $\agent_{i_1}$, $\agent_i$ can partition the items in $\items \setminus \{ \ite_{{j_1}}$ $\ite_{{j_2}}\}$ into $n-1$ partitions with value at least 1 to $\agent_i$, i.e.
$$\MMS_{i}^{n-1}(\items\setminus \{\ite_{j_1},\ite_{j_2}\}) \geq 1.$$
By the same deduction, after assigning $\ite_{{j_3}}$ and $\ite_{{j_4}}$ to $\agent_{i_2}$, we have
$$\MMS_{i}^{n-2}(\items\setminus \{\ite_{j_1},\ite_{j_2},\ite_{j_3},\ite_{j_4}\}) \geq 1.$$
By repeating above argument $|T|$ times, we have:
$$\MMS_{i}^{n-|T|}(\items\setminus S) \geq 1.$$

On the other hand, by condition $(II)$, every agent $\agent_{i_k}$ satisfies with items $\ite_{j_{2k-1}}$ and $\ite_{j_{2k}}$. This means that we can reduce the instance by satisfying the agents in $T$ by the items in $S$, which is a contradiction by the irreducibility assumption.

\end{proof}

\begin{proof}[of Lemma \ref{rem}]
We define $\parttwo_1$ as the set of vertices in $\parttwo$ that are not saturated by $M$, and $\parttwo_2$ as the set of vertices in $\parttwo$ that are connected to $\parttwo_1$ by an alternating path. Moreover, let $\partone_2 = M(\parttwo_2)$. By definition, $F_{H}(M,\partone) = \partone \setminus \partone_2$ (See Figure \ref{fig:FG}). As discussed before, all the vertices in $\partone_2$ are saturated by $M$. Consequently, all the vertices of $T$ are saturated by $M$ and $|N(T)| \geq |T|$. 

Let $M(T)$ be the set of vertices which are matched to the vertices of $T$ in $M$. We know that every vertex of $T$ is present in at least one of the alternating paths which connect $\parttwo_1$ to $\parttwo_2$. Let $$P = \langle \hat{y}_0, \hat{x}_1, \hat{y}_1, \hat{x}_2, \hat{y}_2, \ldots, \hat{x}_k, \hat{y}_k \rangle$$ be one of these paths that includes at least one of the vertices of $T$. Since $P$ is an alternating path which connects $\parttwo_1$ to $\parttwo_2$, $\hat{y}_0 \in \parttwo_1$ (see Figure \ref{fig:FG4}). In addition, according to the definition of alternating path, every edge $(\hat{x}_j,\hat{y}_j)$ of $P$ belongs to $M$ and every edge $(\hat{x}_j,\hat{y}_{j-1})$ does not belong to $M$. 

Let $\hat{x}_i$ be the first vertex of $T$ that appears in $P$. We know that the edge $(\hat{x}_i,\hat{y}_{i-1})$ does not belong to $M$. On the other hand, since $\hat{x}_i$ is the first vertex of $T$ in $M$, $\hat{x}_{i-1} \notin T$. Note that $\hat{y}_{i-1}$ does not belong to $M(T)$, since every vertex of $M(T)$ is matched with a vertex of $T$ in $M$ and $(\hat{x}_{i-1},\hat{y}_{i-1})$ is in $M$.  The fact that $\hat{y}_{i-1} \notin M(T)$ means $N(T)$ contains at least one vertex that is not in $M(T)$. Since all the vertices in $M(T)$ are in $N(T)$, $|N(T)|>|M(T)|$ and hence, $|N(T)|>|T|$.
\begin{figure}
\centering
\includegraphics[scale=0.6]{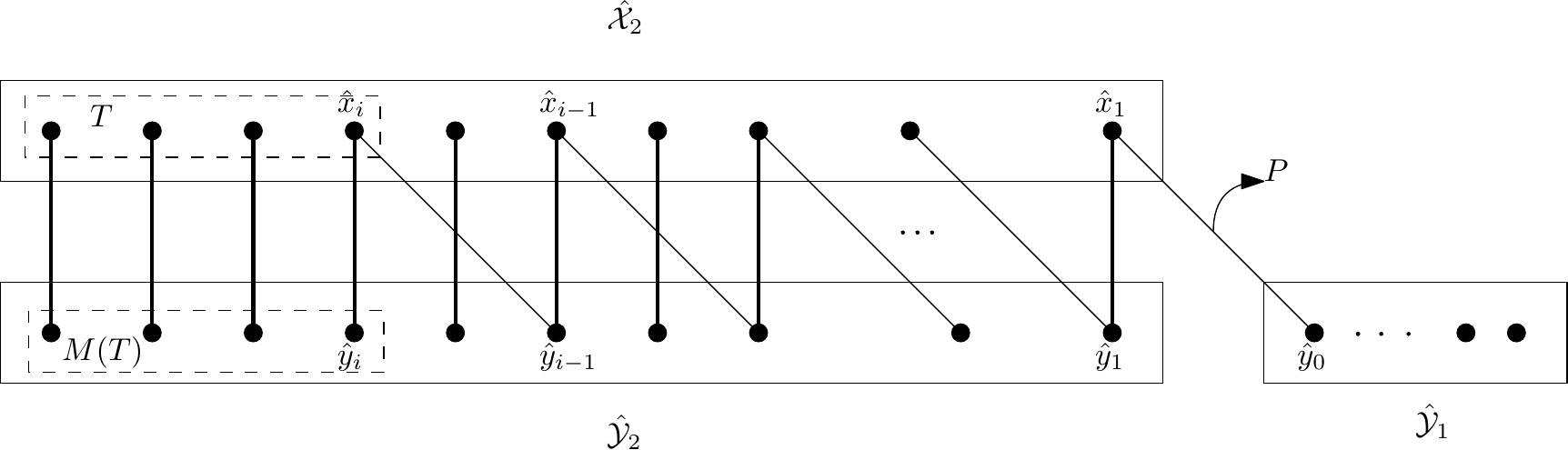}
\caption{Alternating path $P$ connects ${\hat{\cal Y}}_1$ to ${\hat{\cal Y}}_2$ and intersects $T$}
\label{fig:FG4}
\end{figure}
\end{proof}

\begin{proof}[of Lemma \ref{iff}]
If $F_H(M, \partone) = \emptyset$, according to Lemma \ref{rem}, $$\forall T \subseteq \partone \qquad |N(T)| > |T|.$$

On the other hand, suppose that for all $T \subseteq \partone$ we have $|N(T)| > |T|$. We show that $F_H(M, \partone) = \emptyset$. For the sake of contradiction, assume that $F_H(M, \partone) \neq \emptyset$ and let $T = F_H(M, \partone)$. Since there exists a matching from $T$ to $N(T)$ that saturates all the vertices of $N(T)$, we have $|T| \geq |N(T)|$, which is a contradiction. Hence, $F_H(M, \partone) = \emptyset$. 
\end{proof}

\begin{proof}[of Lemma \ref{dag}]
Consider a cycle $L$ in $G_C$. For each vertex $v_j \in L$, there is at least one vertex $v_i \in L$ such that $\agent_i$ envies $\agent_j$. Therefore, Considering $S$ as the set of agents with vertices in $L$, none of the agents of $S$ is a loser. By the same deduction, none of the agents of $S$ is a winner. But this contradicts the fact that the set $C$ is cycle-envy-free.
\end{proof}

\begin{proof}[of Lemma \ref{wm}] We describe the proof for the first condition in more details. The proof for the second condition is almost the same as the first condition. 

\textbf{The first condition}: Suppose that there exists no such vertex. Our goal is to find a new matching of $H$ with the same cardinality, but with more weight. To this end, we construct a directed graph $H'$ from $H$ as follows: for each $\vtwo_j \in T$ we consider a vertex $v_j$ in $V(H')$. Furthermore, there is a directed edge from $v_j$ to $v_i$ in $H'$, if and only if $w(\vone_j,\vtwo_{j}) < w(\vone_i,\vtwo_j)$ in $H$. 

If there exists a vertex $v_j$ with out-degree zero in $H'$, then $\vtwo_j$ is the desired winner in $T$, since
$$ \forall \vtwo_i \in H, w(\vone_j,\vtwo_{j}) \geq w(\vone_i,\vtwo_j).$$
 Otherwise, the out-degree of every vertex in $T$ is non-zero. Therefore, $H'$ has at least one cycle $L = \langle v_{l_1}, v_{l_2}, \ldots, v_{l_{|L|}}\rangle$. Now, if we change matching $M$ by removing the set of edges $$ \{(\vtwo_{l_1},\vone_{l_1}), (\vtwo_{l_2},\vone_{l_2}), \ldots, (\vtwo_{l_{|L|}},\vone_{l_{|L|}})\} $$
from $M$ and adding 
$$\{(\vtwo_{l_1},\vone_{l_2}), (\vtwo_{l_2},\vone_{l_3}),\ldots,(\vtwo_{l_{|L|}},\vone_{l_1})\}$$ 
to $M$, the weight of our matching will be increased. Note that by the definition of an edge in $H'$, we have $$w(\vone_{l_2}, \vtwo_{l_1}) > w(\vone_{l_1}, \vtwo_{l_1}),w(\vone_{l_3}, \vtwo_{l_2}) > w(\vone_{l_2}, \vtwo_{l_2}),\ldots,w(\vone_{l_1}, \vtwo_{l_{|L|}}) > w(\vone_{l_{|L|}}, \vtwo_{l_{|L|}}).$$ But this contradicts the fact that  $M$ was $\MCMWM$ of $H$.

\textbf{The second condition}:
Similar to the proof of the first condition, we construct a new directed graph $H'$ from $H$ where we have a vertex $v_j$ in $H'$ for each vertex $\vtwo_j$ in $T$. For every pair $\vtwo_i$ and $\vtwo_j$ which are members of $T$ we connect $v_i$ to $v_j$ with a directed edge in $H'$ if 
$$w(\vone_j,\vtwo_{i}) > w(\vone_i,\vtwo_i)$$ in $H$ and $(\vone_j,\vtwo_i) \in E(H)$. Note that if $H'$ contains a vertex $v_i$ with in-degree equal to zero, then $\vtwo_i$ is the desired loser in $T$. Thus, suppose that no vertex in $H'$ has in-degree zero and hence, $H'$ has a directed cycle.  Let $L = \langle \vtwo_{l_1}, \vtwo_{l_2}, \ldots, \vtwo_{l_{|L|}}\rangle$ be a directed cycle in $H'$. Similar to the proof of the previous condition, we leverage $L$ to alter $M$ to a new matching with more weight, which is a contradiction by the maximality of $M$.  

\textbf{The third condition}: If $w(\vone_i,\vtwo_i) < w(\vone_j,\vtwo_i)$, we can replace the edge between $\vone_i$ and $\vtwo_i$ by $(\vone_j, \vtwo_i)$ in $M$ which yields a matching with a greater weight. This contradicts the maximality of $M$.
\end{proof}

\section{Omitted Proofs of Section \ref{additive:clusters}}\label{clusteringappendix}

\begin{proof}[of Lemma \ref{forc2c3}]
By definition, there is no edge between the vertices of $F_{G_{1/2}}(M,\itemsv_{1/2})$ and $\agentsv_{1/2} \setminus N(F_{G_{1/2}}(M,\itemsv_{1/2}))$ in $G_{1/2}$. Furthermore, all the items are in worth less than $1/2$ for the agents corresponding to the vertices in $\agentsv \setminus \agentsv_{1/2}$. Thus, for every agent $\agent_i$ and every item $\ite_j$ with $\agentv_i \in \agentsv \setminus N(F_{G_{1/2}}(M,\itemsv_{1/2}))$ and $\itemv_j \in F_{G_{1/2}}(M,\itemsv_{1/2})$, we have $\valu_i(\ite_j)<1/2$. According to the fact that the agents that are not selected in the clustering of $\cone$ either belong to $\ctwo$ or $\cthree$, we have:
\[ \forall \agent_j \in \cone \qquad \valu_i(\firstset_j) < 1/2. \]
\end{proof}

\begin{proof}[of Lemma \ref{nicematch}]
First, we prove Lemma \ref{v1size}. This lemma ensures that there exists a matching in $G_1$ that saturates all the vertices in $W_1$. Lemma \ref{v1size} is a consequence of irreducibility. In fact, we show that if the condition in Lemma \ref{v1size} does not hold, the instance is reducible.
\begin{lemma}
\label{v1size}
For graph $G_1$, we have $$ \forall R \subseteq  W_1, \qquad |N(R)| > |R|.$$ 
\end{lemma}
\begin{proof}
Let $M_1$ a matching with the maximum number of edges in $G_1$ . Regarding Lemma \ref{iff}, it only suffices to show that $F_{G_1}(M_1,W_1)$ is empty. For the sake of contradiction, suppose that $F_{G_1}(M_1,W_1)$ is not empty. As mentioned before, there exists a matching between $F_{G_1}(M_1,W_1)$ and $N(F_{G_1}(M_1,W_1))$ that saturates all the vertices in $N(F_{G_1}(M_1,W_1))$. Let 
$$M_S = \{(\itemv_{j_1},\agentv_{i_1}),(\itemv_{j_2},\agentv_{i_2}),\ldots,(\itemv_{j_k},\agentv_{i_k})\}$$
be this matching. We show that the set of agents $$T= \{\agent_{i_1}, \agent_{i_2}, \ldots, \agent_{i_k}\}$$ and the set of items $$S = \{\firstset_{i_1}, \ite_{j_1},\firstset_{i_2},\ite_{j_2},\ldots, \firstset_{i_k},\ite_{j_k} \}.$$
have all three conditions in Lemma \ref{remove3} (Note that $\firstset_{i_l}$ contains exactly one item).  The first condition is trivial: $|S| = 2|T|$. Regarding the definition of an edge in $G_1$, we know that $\firstset_{i_l} \cup \{\ite_{j_l}\}$ satisfy $\agent_{i_l}$ and hence, the second condition is held as well.  
For the third condition, we should prove that for every agent $\agent_{i_l}$ in $T$, 
$$\valu_{i'}(\firstset_{i_l} \cup \{\ite_{j_l}\}) < 1 \qquad  \forall \agent_{i'} \notin T .$$
To show this, we consider two cases separately. First, if $\agent_{i'} \notin C_1$, by Lemma \ref{forc2c3}, $\valu_{i'}(\firstset_{i_l})<1/2$ and by Observation \ref{w1small}, $\valu_{i'}(\{\ite_{j_l}\})<1/2$, which means $\valu_{i'}(\firstset_{i_l} \cup \{\ite_{j_l}\}) < 1$.

Moreover, consider the case that $\agent_{i'} \in C_1$. Note that since $\agent_{i'} \notin T$, it's corresponding vertex $\agentv_{i'}$ is not in $N(F_{G_1}(M_1,W_1))$, which means:
$$\agentv_{i'} \in V_{C_1} \setminus N(F_{G_1}(M_1,W_1)).$$
By the definition of $N(F_{G_1}(M_1,W_1)$, there is no edge between $\agentv_{i'}$ and $\itemv_{j_l}$ and hence, $\valu_{i'}(\{\ite_{j_l}\})< \epsilon_{i'} \leq 1/4$. On the other hand, by the irreducibility assumption and the fact that $\firstset_{i_l}$ contains exactly one item,  $\valu_{i'}(\firstset_{i_l}) < 3/4$. Thus, $\valu_{i'}(\firstset_{i_l} \cup \{\ite_{j_l}\}) < 1$. 

As a result, $\valu_{i'}(\firstset_{i_l} \cup \{\ite_{j_l}\}) < 1$ for every agent $\agent_{i'} \notin T$ which means the third condition of Lemma \ref{remove3} is held as well. Thus, regarding Lemma \ref{remove3}, the instance is reducible. But this contradicts the irreducibility assumption.
\end{proof}

The rest of the proof of Lemma \ref{nicematch} is as follows. Since we used $\MCMWM$ to build cluster $\cone$, regarding Lemma \ref{wm}, $\cone$ is cycle-envy-free. Consider the topological ordering of $\cone$ and let $p_{a_i}$ be the position of $\agent_i$ in this ordering. More precisely, $p_{\agent_i} = k$ if $\agent_i$ is the $k$-th agent in the topological ordering of $\cone$.

According to Lemma \ref{v1size}, the condition of Hall’s Theorem holds for graph $G_1$ and as a result there exists a matching in $G_1$ that saturates all the vertices in $W_1$. Among all possible maximum matchings of $G_1$, let $M_1$ be a maximum matching that minimizes $$p_{M_1} = \sum_{\agentv_i \in M_1} p_{a_i}.$$ We claim that $M_1$ is the desired matching described in Lemma \ref{nicematch}. To prove our claim, we must show that for any edge $(\itemv_i, \agentv_j) \in M_1$ and any unsaturated vertex $\agentv_k \in N(\itemv_i)$, $\agent_j$ is a loser for the set $\{\agent_j, \agent_k\}$, which means $\agent_k$ does not envy $\agent_j$. Note that if $\agent_k$ envies $\agent_j$, $\agent_k$ appears before $\agent_j$ in the topological ordering of $\cone$ which means $p_{\agent_k} < p_{\agent_j}$. Therefore, if we replace $(\itemv_i, \agentv_j)$ by $(\itemv_i, \agentv_k)$ in $M_1$, $p_{M_1}$ will be decreased that contradicts the minimality of $p_{M_1}$.
\end{proof}

\begin{proof}[of Lemma \ref{gsmallc1r}]
Let $\ite_k$ be the item assigned to $\agent_j$ in the refinement of $\cone$. Since $\itemv_k \in W_1$, according to Observation \ref{w1small}, $\valu_i(\secondset_j)< 1/2$. 
\end{proof}

\begin{proof}[of Lemma \ref{forc2}]
Let $\agent_j$ be an agent in $\satagents_1^r$. First, note that $|\firstset_j| = |\secondset_j| = 1$. Lemma \ref{forc2c3} together with Observation \ref{w1small} state that $\valu_i(\firstset_j \cup \secondset_j)<1$. According to Inequality (\ref{saeed1}), we have  
\begin{equation}
\label{eq100}
\MMS_{\valu_i}^{|\agents \setminus \agent_j|} ( {\items} \setminus \firstset_j \cup \secondset_j) \geq 1.
\end{equation}
 Note that Equation (\ref{eq100}) holds for every agent in $\satagents_1^r$. Applying Equation (\ref{eq100}) to all the agents of $\satagents_1^r$ yields
 \[ \MMS_{\valu_i}^{|\agents \setminus \satagents_1^r|} ( {\items} \setminus \bigcup_{\agentv_i \in \satagents_1^r} \firstset_i \cup \secondset_i) \geq 1.\]

\end{proof}

\begin{proof}[of Lemma \ref{c1small2}]
According to Observation \ref{fsmallc1}, for any agent $a_k \in \cone$ and for every $\itemv_j \in \itemsv' \setminus \itemsv'_{1/2}$ we have $V_k(\{b_j\}) < \epsilon_k$. By additivity assumption, for any $\agent_k \in \cone$  we have 
$$ \forall {x_i, x_j \in \itemsv' \setminus \itemsv'_{1/2}} \qquad V_k(\{b_i, b_j\}) < 2\epsilon_k.$$
\end{proof}

\begin{proof}[of Lemma \ref{pairsmall}]
Suppose for the sake of contradiction that the problem is $3/4$-irreducible, and there exists a vertex $y_k \in \agentsv$ such that $V_k(\{b_i, b_j\}) \geq 3/4$. According to Lemma \ref{remove2} there exists an agent $a_{k'} \neq a_k$ such that $$V_{k'}(\{b_i, b_j\}) \geq 1.$$
Since the valuations are additive, we know that one of the inequalities $V_{k'}(\{b_i\}) \geq 1/2$ or $V_{k'}(\{b_j\}) \geq 1/2$ are held, which is contradiction, since we know both $\itemv_i$ and $\itemv_j$ belong to $\itemsv' \setminus \itemsv'_{1/2}$.
\end{proof}

\begin{proof}[of Lemma \ref{sizeeq}]
We prove Lemma \ref{sizeeq} in two steps. Firstly, we show that
\begin{equation}
\label{ineqhadi2}
|F_{G'_{1/2}}(M',\itemsv'_{1/2})| \leq |N(F_{G'_{1/2}}(M',\itemsv'_{1/2}))|.
\end{equation}
Furthermore, we prove 
\begin{equation}
\label{ineqhadi1}
|F_{G'_{1/2}}(M',\itemsv'_{1/2})| \geq |N(F_{G'_{1/2}}(M',\itemsv'_{1/2}))|.
\end{equation}
Inequalities \eqref{ineqhadi1} and \eqref{ineqhadi1} yields
\begin{equation}
\label{ineqhadi3}
|F_{G'_{1/2}}(M',\itemsv'_{1/2})| = |N(F_{G'_{1/2}}(M',\itemsv'_{1/2}))|.
\end{equation}

\textbf{To show Inequality \eqref{ineqhadi2},} argue that before Algorithm \ref{addvertex} starts, we have $$F_{G'_{1/2}}(M',\itemsv'_{1/2}) = \emptyset$$ and  $$N(F_{G'_{1/2}}(M',\itemsv'_{1/2})) = \emptyset$$ and all the vertices in $\itemsv'_{1/2}$ are saturated by $M'$. In each step of Algorithm \ref{addvertex}, we add a new vertex to $\itemsv'_{1/2}$, and the size of the maximum matching $M'$ is increased by one. Therefore, after each step of Algorithm \ref{addvertex}, all of the vertices in $\itemsv'_{1/2}$ remain saturated by $M'$. Since $F_{G'_{1/2}}(M',\itemsv'_{1/2}) \subseteq \itemsv'_{1/2}$, all the vertices of $F_{G'_{1/2}}(M',\itemsv'_{1/2})$ are also saturated by $M'$, which means
$$|F_{G'_{1/2}}(M',\itemsv'_{1/2})| \leq |N(F_{G'_{1/2}}(M',\itemsv'_{1/2}))|.$$

\textbf{To prove Inequality \eqref{ineqhadi1}}, note that by definition, $F_{G'_{1/2}}(M',\itemsv'_{1/2})$ has a property that there exists a matching from $F_{G'_{1/2}}(M',\itemsv'_{1/2})$ to $N(F_{G'_{1/2}}(M',\itemsv'_{1/2}))$ that saturates all the vertices of $N(F_{G'_{1/2}}(M',\itemsv'_{1/2}))$. Therefore, we have
$$
|F_{G'_{1/2}}(M',\itemsv'_{1/2})| \geq |N(F_{G'_{1/2}}(M',\itemsv'_{1/2}))|.
$$
This completes the proof.
\end{proof}

\begin{proof}[of Lemma \ref{forc3}] 
Firstly, we clarify what agents are in $\cthree$. Roughly speaking, the agents that are not selected for Clusters $\cone$ and $\ctwo$ are in $\cthree$. Thus, the agents in $\cthree$ correspond to the vertices in  
$$\agentsv' \setminus N(F_{G'_{1/2}}(M',\itemsv'_{1/2}))$$
$$=\big(\agentsv' \setminus \agentsv'_{1/2}\big) \cup  \big(\agentsv'_{1/2} \setminus N(F_{G'_{1/2}}(M',\itemsv'_{1/2}))\big).$$

\textbf{The term $ \agentsv' \setminus \agentsv'_{1/2} $ } refers to the vertices that are filtered in $G'_{1/2}$ which means no edge with weight at least $1/2$ is incident to any of these vertices.  
On the other hand, for every agent $\agent_j \in \ctwo$, $\firstset_j$ corresponds to a vertex in $F_{G'_{1/2}}(M',\itemsv'_{1/2})$.  Hence, for every agent $\agent_j \in \ctwo$  and every agent $\agent_i$ with corresponding vertex in $\agentsv' \setminus \agentsv'_{1/2}$ we have $\valu_i(f_j) <1/2$

\textbf{Next, consider the term $\agentsv'_{1/2} \setminus N(F_{G'_{1/2}}(M',\itemsv'_{1/2}))$.} By definition, the vertices of $F_{G'_{1/2}}(M',\itemsv'_{1/2})$ are only incident to the vertices of $N(F_{G'_{1/2}}(M',\itemsv'_{1/2}))$ in $G'_{1/2}$. Regarding the definition of an edge in $G'_{1/2}$, for every agent $\agent_j \in \ctwo$ and agent $\agent_i$ with $\agentv_i \in \agentsv'_{1/2} \setminus N(F_{G'_{1/2}}(M',\itemsv'_{1/2}))$ we have $\valu_i(f_j) <1/2$.

Therefore, for all $\agent_i \in \cthree$ we have: 
 $$\forall \agent_j \in \ctwo \qquad \valu_i(\firstset_j) < 1/2.$$
\end{proof}

\begin{proof}[of Lemma \ref{cr2smallc1}]
Regarding Observation \ref{fsmallc1}, after refinement of $\cone$, all the items with vertex in $\itemsv' \setminus \itemsv'_{1/2}$ are in worth less than $\epsilon_j$ for every agent $\agent_j \in \cone$. Furthermore, note that for every agent $\agent_i \in \satagents_2^r$, $\secondset_i$ is a single item with vertex in $\itemsv' \setminus \itemsv'_{1/2}$. Thus, $\valu_j(\secondset_i)< \epsilon_j$ for every agent $\agent_j \in \cone$.
\end{proof}

\begin{proof}[of Lemma \ref{cr2smallc3}]
According to Algorithm \ref{c2ref}, for any agent $\agent_i \in \satagents_2^r$, the corresponding vertex of the only member of $g_i$ is in $\itemsv' \setminus \itemsv'_{1/2}$. Therefore, for any agent $a_j \notin \cone \cup \ctwo$ we have $V_j(g_i) < 1/2$. Finally, note that the remaining agents that are not in $\cone$ and $\ctwo$ belong to $\cthree$.
\end{proof}

\begin{proof}[of Lemma \ref{lsmall_c3}] The algorithm \ref{addvertex} terminates when there is no desirable pair for the agents in $T = \agentsv' \setminus N(F_{G'_{1/2}}(M',\itemsv'_{1/2})).$ Furthermore, by definition, for every agent  $\agent_i \in \cthree$ we have  $$\agentv_i \in \agentsv' \setminus N(F_{G'_{1/2}}(M',\itemsv'_{1/2})).$$ But at the end of Algorithm \ref{addvertex}, no pair of vertices is desirable for $\agent_i$ which means for every $\itemv_j,\itemv_k \in \itemsv'' \setminus \itemsv''_{1/2}$, we have  $V_i(\{\ite_j,\ite_k\}) < {1/2}$ (note that $\itemsv'' \setminus \itemsv''_{1/2} \subseteq \itemsv' \setminus \itemsv'_{1/2}$).
\end{proof}

\section{Omitted Proofs of Section \ref{additive:allocation}}\label{clustering2appendix}
\begin{proof}[of Lemma \ref{general}]
At this point, for every agent $\agent_i \in \cone \cup \ctwo \cup \cthree^s$, $|\firstset_j| \leq 2$. If $|\firstset_i| = 1$ holds, then according to Lemma \ref{remove1}, value of the item in $\firstset_i$ is less than $3/4$ to all other agents. Moreover, if $|\firstset_i| = 2$, then $\firstset_i$ corresponds to a merged vertex. In this case, by Lemmas \ref{c1small2} and \ref{pairsmall}, value of $\firstset_i$ is less than $3/4$ to all other agents. 
\end{proof}

\begin{proof}[of Lemma \ref{c3fsmall}]
According to Lemma \ref{lsmall_c3}, value of every pair of items in $\fitems$ is less than $1/2$ to $\agent_i$. Therefore, $\firstset_i$ contains at least three items. Let $\ite_k$ be an arbitrary item in $\firstset_i$. Since $|\firstset_i| \geq 3$, $\firstset_i \setminus \{\ite_k\}$ is non-empty. On the other hand, $S$ is minimal and hence, none of the sets $\firstset_i \setminus \ite_k$ and $\{\ite_k\}$ is feasible for any agent. According to the definition of feasibility for the agents of $\cone \cup \ctwo \cup \cthree^s \cup \cthree^b$, we have
$$ \forall \agent_j \in \cone \cup \ctwo \cup \cthree^s \cup \cthree^b \qquad \valu_j(\firstset_i \setminus \{\ite_k\})< \epsilon_j $$ 
and 
$$\forall \agent_j \in \cone \cup \ctwo \cup \cthree^s \cup \cthree^b \qquad \valu_j(\{\ite_k\})< \epsilon_j$$
which means
$$ \forall \agent_j \in \cone \cup \ctwo \cup \cthree^s \cup \cthree^b \qquad \valu_j(\firstset_i)< 2\epsilon_j. $$ 

\end{proof}

\begin{proof}[of Lemma \ref{cef}]
The Lemma trivially holds for $\cone$ and $\ctwo$, since removing an agent from a cycle-envy-free set preserves this property. For $\cthree^s$, there may be multiple rounds that an agent is added to $\cthree^s$. We show that adding an agent to $\cthree^s$ preserves cycle-envy-freeness as well.

For the sake of contradiction, let $\mathbb{R}_z$ be the first round in which adding an agent $\agent_i$ to $\cthree^s$ results in a set, that is no longer cycle-envy-free. Since $\cthree^s \setminus \{\agent_i\}$ is cycle-envy-free, every subset of $\cthree^s \setminus \{\agent_i\}$ contains at least one winner and one loser. Moreover, by Lemma \ref{c3fsmall} we have:
\begin{equation}
\label{inec1}
\forall \agent_j \in \cthree^s, j \neq i, \qquad \valu_j( \firstset_i )< 2\epsilon_j.
\end{equation}

Note that $\agent_i$ previously belonged to $\cthree^f$.  By definition of $\cthree^f$  
\begin{equation}
\label{inec2}
\forall \agent_j \in \cthree^s , j \neq i, \qquad \valu_i(\firstset_j)< 1/2 .
\end{equation}

Inequalities (\ref{inec1}) and (\ref{inec2}) together imply that $\agent_i$ is both a winner and a loser for every subset of $\cthree^s$ that contains $\agent_i$. This means that every subset of $\cthree^s$ contains at least one winner and one loser, which is a contradiction.

\end{proof}

\begin{proof}[of Lemma \ref{prvalue}]

If $\agent_j \prec_{pr} \agent_i$, then $\secondset_i$ is not feasible for $\agent_j$, since the agent with the lowest 
priority is satisfied in each round of the second phase. Thus, $\valu_j(\secondset_i) < \epsilon_j$. For the case where $\agent_i \prec_{pr} \agent_j$, let $\ite_k$ be an arbitrary item of $\secondset_i$. According to the fact that $\secondset_i$ is minimal, $\secondset_i \setminus \{\ite_k\}$ is not feasible for any agent. Hence, $\valu_j(\secondset_i \setminus \{\ite_k\})< \epsilon_j$. On the other hand, by Observations \ref{fsmallc1} and \ref{fsmallc2}, $\valu_j(\{\ite_k\})<\epsilon_j $. Therefore, $\valu_j(\secondset_i)<2\epsilon_j$. 
\end{proof}

\begin{proof}[of Lemma \ref{m_1}]
Let $\mathbb{R}_z$ be the round, in which $\agent_i$ is satisfied. At that point, if $\agent_j \in \cthree^f$ then $\valu_j(\secondset_i) < {1/2}$ trivially holds. Since in round $\mathbb{R}_z$, $\agent_j \prec_{pr} \agent_i$ holds, $\secondset_i$ was not feasible for $\agent_j$ in the first place. Recall that in each round, the agent with lowest order in $\Phi(S)$ is selected. 

Furthermore, if in round $\mathbb{R}_z$, $\agent_j$ was in $\cthree^s \cup \cthree^b$, according to Observations \ref{fsmallc1} and \ref{fsmallc2}, $|S| \geq 2$, since no item alone can satisfy $\agent_i$. If $|S|=2$, then by Observation \ref{lsmall_c3}, $\valu_j(\secondset_i)<1/2$. For the case of $|S|>2$, let $\ite_k$ be the item in $S$ with the minimum value to $\agent_j$. According to Corollary \ref{small_c3}, $\valu_j(\{\ite_k\})<{1/4}$. Also, since $S$ is minimal, $S \setminus \{\ite_k\}$ is not feasible for any agent and hence, $\valu_j(S \setminus \{\ite_k\}) < \epsilon_j \leq {1/4}$. Thus, $\valu_j(S) < {1/2}$.
\end{proof}

\color{black}	 
\section{Omitted Proofs of Section \ref{additiveproofs}}\label{additiveproofappendix}
Before proceeding to the proof of Lemma \ref{c3null}, we show Lemmas (\ref{m_2}, \ref{c3bssmall} and \ref{c3sat}). 

\begin{lemma}
\label{m_2}
Let $\agent_i$ be an agent in $\satagents_3$ and let ${\mathbb R}_z$ be the round of the second phase in which $\agent_i$ is satisfied. Then, for any other agent $\agent_j$ that is  in $\cthree^f$ in ${\mathbb R}_z$, $\valu_j(\secondset_i) < 1/2$ holds.
\end{lemma}

\begin{proof}
In ${\mathbb R}_z$, $\agent_i$ either belongs to $\cthree^s$ or $\cthree^b$. Thus, $\agent_j \prec_{pr} \agent_i$, and thus $\secondset_i$ is not feasible for $\agent_j$ in that round. Therefore, $\valu_j(\secondset_i)< 1/2$.
\end{proof}

\begin{lemma}
\label{c3bssmall}
Let $\agent_i \in \satagents_3$ be a satisfied agent and let ${\mathbb R}_z$ be the round in which $\agent_i$ is satisfied. Then, for every other agent $\agent_j$ that belongs to $\cthree^s \cup \cthree^b$ in that round, either $\valu_j(\secondset_i) < \epsilon_j$ or $\valu_j(\firstset_i) \leq 3/4-\epsilon_j$.

\end{lemma}
\begin{proof}
If $\secondset_i$ is not feasible for $\agent_j$, then the condition trivially holds. Moreover, by the definition, the statement is correct for the agents of $\cthree^b$. Therefore, it only suffices  to  consider the case that $\agent_j \in \cthree^s$ and $\secondset_i$ is feasible for $\agent_j$. Due to the priority rules for satisfying the agents in the second phase, $\agent_i \prec_{pr} \agent_j$ and hence, $\agent_i$ cannot be in  $\cthree^b$. Thus, $\agent_i \in \cthree^s$. According to Observation \ref{epsofcluster} and the fact that $\prec_{pr}$ is equivalent to $\prec_{o}$ for the agents in $\cthree^s$, we have $\valu_j(\firstset_i) \leq 3/4 - \epsilon_j$.
\end{proof}

\begin{lemma}
\label{c3sat}
During the second phase, for any agent  $\agent_i$ in  $\cthree$, we have: $$\sum_{\agent_j \in \satagents_3} \valu_i(\firstset_j \cup \secondset_j)< |\satagents_3| + 1/4.$$ 
\end{lemma}

\begin{proof}
To show Lemma \ref{c3sat}, we show that for all the agents $\agent_j \in \satagents_3$ except at most one agent, $\valu_i(\firstset_j \cup \secondset_j)<1$ holds. To show this, let ${\mathbb R}_z$ be an arbitrary round of the second phase, in which an agent $\agent_j \in \cthree$ is satisfied. First, note that in ${\mathbb R}_z$, $\agent_j$ belongs to $\cthree^s \cup \cthree^b$. Also, in round ${\mathbb R}_z$, $\agent_i$ belongs to one of $\cthree^s, \cthree^b$, or $\cthree^f$.  
 
If $\agent_i \in \cthree^f$, then by Lemma \ref{m_2}, $\valu_i(\secondset_j)<1/2$ holds. On the other hand, by definition, $\valu_i(\firstset_j)<1/2$ and hence, $\valu_i(\firstset_j \cup \secondset_j)<1$. 

Now, consider the case, where $\agent_i \in \cthree^b \cup \cthree^s$. Note that by Lemma \ref{c3bssmall}, either $\valu_i(\firstset_j) \leq 3/4-\epsilon_i$ or $\valu_i(\secondset_j) < \epsilon_i$. If $\valu_i(\secondset_j) < \epsilon_i$, then by Lemmas \ref{general} and \ref{c3fsmall}, we know $\valu_i(\firstset_j) < 3/4$ and hence, $\valu_i(\firstset_j \cup \secondset_j)<3/4 + \epsilon_i < 1$. 

For the case where $\valu_i(\firstset_j) \leq 3/4-\epsilon_i$, let $\ite_l$ be the item in $\secondset_j$ with the maximum value to $\agent_i$. By minimality of $\secondset_j$, $\secondset_j \setminus \{\ite_l\}$ is not feasible for any agent, including  $\agent_i$ and thus, $\valu_i(\secondset_j\setminus \{\ite_l\}) < \epsilon_i$. Recall that by Corollary \ref{small_c3}, there is at most one item $\ite_k$ in $\fitems$, such that $\valu_i(\ite _k) \geq 1/4$. In addition to this, $\valu_i(\ite_k) < 1/2$ trivially holds, since $\ite_k$ is not assigned to any agent during the clustering phase. If $\ite_l \neq \ite_k$, $\valu_i(\secondset_j)< 1/4 + \epsilon_i$ holds and hence, $$\valu_i(\firstset_j \cup \secondset_j) < 3/4-\epsilon_i + 1/4 + \epsilon_i<1.$$ Moreover, If $\ite_l = \ite_k$, $\valu_i(\secondset_j)< 1/2 + \epsilon_i$ holds and thus, $\valu_i(\firstset_j \cup \secondset_j) < 3/4-\epsilon_i + 1/2 + \epsilon_i<5/4$. But, this can happen at most one round. Therefore, for all the agents $\agent_j \in \satagents_3$ except at most one, $\valu_i(\firstset_j \cup \secondset_j)<1$. Also, for at most one agent $\agent_j \in \satagents_3$, $\valu_i(\firstset_j \cup \secondset_j)<5/4$. Thus, 
$$\sum_{\agent_j \in \satagents_3} \valu_i(\firstset_j \cup \secondset_j)< |\satagents_3| + 1/4.$$   
\end{proof}

\begin{proof}[of Lemma \ref{c3null}]
Suppose for the sake of contradiction that $\cthree \neq \emptyset$.  Note that, by the definition of $\cthree^b$, if $\cthree^s = \emptyset$ holds, then consequently $\cthree^b = \emptyset$. Therefore, since we have $\cthree = \cthree^s \cup \cthree^b \cup \cthree^f$, if $\cthree$ is non-empty, at least either of the two sets $\cthree^s$ or $\cthree^f$ is non-empty. In case $\cthree^s$ is non-empty, let $\agent_i$ be a winner of $\cthree^s$, otherwise let $\agent_i$ be an arbitrary agent of $\cthree^f$.

According to Lemma \ref{m_1}, for every agent $\agent_j \in \satagents_1^s \cup \satagents_2^s$, $\valu_i(\secondset_j) < {1/2}$ holds. Also, by Lemmas \ref{gsmallc1r} and \ref{cr2smallc3}, for every agent  $\agent_j \in \satagents_1^r \cup \satagents_2^r$, we have $\valu_i(\secondset_j) < {1/2}$. Therefore, 
$$\forall \agent_j \in \satagents_1 \cup \satagents_2 \qquad \valu_i(\secondset_j) < {1/2}.$$

Also, by Lemmas \ref{forc2c3} and \ref{forc3} we know that $\valu_i(\firstset_j)<{1/2}$ for every $\agent_j \in \satagents_1 \cup \satagents_2$. Thus, for every satisfied agent $\agent_j \in \satagents_1 \cup \satagents_2$, $\valu_i(\firstset_j \cup \secondset_j) <1$ holds, and hence 
\begin{equation}\label{eq1}
\sum_{\agent_j \in \satagents_1 \cup \satagents_2} \valu_i (\firstset_j \cup \secondset_j) < |\satagents_1 \cup \satagents_2|.
\end{equation}

Moreover, by Lemma \ref{c3sat}, the total value of items assigned to the agents in $\satagents_3$ to $\agent_i$ is less than $|\satagents_3| + 1/4$. More precisely,
\begin{equation}\label{eq2}
\sum_{\agent_j \in \satagents_3} \valu_i(\firstset_j \cup \secondset_j) \leq |\satagents_3| + 1/4.
\end{equation}
Inequality \eqref{eq1} along with Inequality \eqref{eq2} implies: 
\begin{equation}
\begin{split}
\sum_{\agent_j \in \satagents} \valu_i (\firstset_j \cup \secondset_j) & = \sum_{\agent_j \in \satagents_1 \cup \satagents_2} \valu_i (\firstset_j \cup \secondset_j) + \sum_{\agent_j \in \satagents_3} \valu_i (\firstset_j \cup \secondset_j)\\
& < |\satagents_1 \cup \satagents_2| + |\satagents_3| + 1/4\\
 & = |\satagents |+1/4
\end{split}
\end{equation}

Recall that the total sum of the item values for $\agent_i$ is equal to $n$. In addition to this, since every agent belongs to either of the Clusters $\cone$, $\ctwo$, $\cthree$, or $\satagents$ we have $$|\satagents| + |\cone| + |\ctwo| + |\cthree| = n.$$ Furthermore, every item $\ite_j \in \items$ either belongs to $\fitems$ or one of the sets $\firstset_{j'}$ and $\secondset_{j'}$ for an agent $\agent_{j'}$. More precisely,
$$\fitems = \items \setminus \Big[\bigcup_{\agent_j \in \satagents \cup \cone \cup \ctwo \cup \cthree^s} \firstset_j \cup \bigcup_{\agent_j \in \satagents} \secondset_j\Big].$$ 
 Therefore
\begin{equation}\label{eq5}
\begin{split}
\sum_{\agent_j \in \cone} \valu_i(\firstset_j) + \sum_{\agent_j \in \ctwo} \valu_i(\firstset_j) + \sum_{\agent_j \in \cthree^s} \valu_i(\firstset_j) + \valu_i({\fitems}) & = \valu_i(\items) - \sum_{\agent_j \in \satagents} \valu_i(\firstset_j \cup \secondset_j)\\
&= n - \sum_{\agent_j \in \satagents} \valu_i(\firstset_j \cup \secondset_j)\\
&\geq n - (|\satagents| + 1/4)\\
&= |\cone| + |\ctwo| + |\cthree|-1/4
\end{split}
\end{equation}

According to Lemmas \ref{forc2c3} and \ref{forc2},  
\begin{equation}\label{eq5.1}
\sum_{\agent_j \in \cone} \valu_i(\firstset_j) < {1/2}|\cone|
\end{equation}
 and 
\begin{equation}\label{eq5.2}
\sum_{\agent_j \in \ctwo} \valu_i(\firstset_j)< {1/2} |\ctwo|
\end{equation}
hold. Inequalities \eqref{eq5}, \eqref{eq5.1}, and \eqref{eq5.2} together prove
\begin{equation}\label{eq6}
\begin{split}
\valu_i({\fitems}) &\geq |\cone| + |\ctwo| + |\cthree|-1/4 - \big[\sum_{\agent_j \in \cone} \valu_i(\firstset_j) + \sum_{\agent_j \in \ctwo} \valu_i(\firstset_j) + \sum_{\agent_j \in \cthree^s} \valu_i(\firstset_j)\big]\\
&\geq |\cone| + |\ctwo| + |\cthree|-1/4 - \big[1/2|\cone| + 1/2|\ctwo| + \sum_{\agent_j \in \cthree^s} \valu_i(\firstset_j)\big]\\
&\geq 1/2 |\cone| + 1/2 |\ctwo| + |\cthree| -1/4 - \sum_{\agent_j \in \cthree^s} \valu_i(\firstset_j).
\end{split}
\end{equation}
Now, we consider two cases separately: (i) $\agent_i \in \cthree^s$ and  (ii) $\agent_i \in \cthree^f$.

\textbf{In case $\agent_i \in \cthree^s$}, since $\agent_i$ is a winner of $\cthree^s$, we have 
\begin{equation}
\begin{split}
\sum_{\agent_j \in \cthree^s} \valu_i(\firstset_j) & \leq \sum_{\agent_j \in \cthree^s}  \valu_i(\firstset_i)\\
& = \sum_{\agent_j \in \cthree^s} 3/4 - \epsilon_i\\
& = ({3/4}-\epsilon_i) |\cthree^s|.
\end{split}
\end{equation}
This combined with Inequality \eqref{eq6} concludes
\begin{equation*}
\begin{split}
 \valu(\fitems) &\geq  1/2 |\cone| + 1/2 |\ctwo| + |\cthree| -1/4 - \sum_{\agent_j \in \cthree^s} \valu_i(\firstset_j)\\
 & \geq 1/2 |\cone| + 1/2 |\ctwo| + |\cthree| - 1/4 - ({3/4}-\epsilon_i) |\cthree^s|\\
 & \geq 1/2 |\cone| + 1/2 |\ctwo| + (1/4 + \epsilon) |\cthree| - 1/4.
\end{split}
\end{equation*}
On the other hand, since $\agent_i \in \cthree^s$, $|\cthree| \geq 1$ and hence, $\valu_i({\fitems}) \geq {1/4} + \epsilon_j - {1/4} = \epsilon_j$. This means that $\fitems$ is feasible for $\agent_i$, which contradicts the termination of the algorithm. 

\textbf{In case $\agent_i \in \cthree^f$}, by the definition of $\cthree^f$ we know that $\sum_{\agent_j \in \cthree^s} \valu_i(\firstset_j) < {1/2} |\cthree^s|$, which by Inequality \eqref{eq6} implies:

$$\valu_i({\fitems}) > {1/2}|\cthree^s| + |\cthree^b| + |\cthree^f| + {1/2}|\ctwo| + {1/2}|\cone|-1/4.$$

Since $\agent_i \in \cthree^f$, we have $|\cthree^f| \geq 1$ and hence, $\valu_i({\fitems}) > 3/4$. Again, this contradicts the termination of the algorithm since $\fitems$ is feasible for $\agent_i$.  
\end{proof}

\begin{proof}[of Lemma \ref{c1null}]
By Lemma \ref{c3null}, we already know $\cthree = \emptyset$. Now, let $\agent_i$ be a winner of the remaining agents in $\cone$. For convenience, we color the items in either blue or white. Intuitively, blue items may have a high value for $\agent_i$ whereas white items are always of lower value to $\agent_i$. Initially, all items are colored in white. For each $\agent_j  \in \agents$, if $|\firstset_j|=1$, then we color the item in $\firstset_j$ in blue. Moreover, for every $\agent_j \in \satagents$, if $|\secondset_j|=1$ and $\valu_i(\secondset_j) \geq \epsilon_i$, then we color the item in $\secondset_j$ in blue.

Now, let $\cal P$ $= \langle P_1, P_2, \ldots, P_n \rangle$ be the optimal $n$-partitioning of the items in $\items$ for $\agent_i$, that is, the value of every partition $P_k$ to $\agent_i$ is at least $1$. Based on the coloring procedure, we have three types of partitions in $\cal P$:
\begin{itemize}
    \item $B_2$: the set of partitions with at least two blue items
    \item $B_1$: the set of partitions with exactly one blue item
    \item $B_0$: the set of partitions without any blue items
\end{itemize}
Note that every partition in $\cal P$ belongs to one of $B_0,B_1$ or $B_2$. Hence,
\begin{equation}
\label{sumba}
|B_0| + |B_1| + |B_2| = n
\end{equation}
As declared, all the items in the partitions of $B_0$ are white. The total value of these items to $\agent_i$ is at least $|B_0| \geq 4 \epsilon_i |B_0|$, which is
\begin{equation}\label{pq1}
\sum_{P_k \in B_0}\sum_{\ite_j \in P_k} \valu_i(\ite_j) \geq 4 \epsilon_i |B_0|.
\end{equation}
 Also, each partition in $B_2$ has at least two blue items, each of which is singly assigned to another agent. We decompose the partitions of $B_1$ into two disjoint sets, namely $\hat{B_1}$ and $\tilde{B_1}$. More precisely, let $\hat{B_1}$ be the partitions in $B_1$, in which the blue item is worth more than $\valu_i(\firstset_i)$ to $\agent_i$ and $\tilde{B_1} = B_1 \setminus \hat{B_1}$. As such, for each partition $P_k \in \tilde{B_1}$, 
the white items in $P_k$ are worth at least 
\begin{equation*}
\begin{split}
1- \valu_i(\firstset_i) &= 1-(3/4-\epsilon_i)\\
&= {1/4}+ \epsilon_i\\
&\geq 2\epsilon_i
\end{split}
\end{equation*}
to $\agent_i$. Therefore,
\begin{equation}\label{pq2}
\sum_{P_k \in \tilde{B_1}}  \valu_i(\wcal(P_k)) \geq 2 |\tilde{B_1}|\epsilon_i
\end{equation}
where $\wcal(S)$ stands for the set of white items in a set $S$ of items.
 On the other hand, since the problem is $3/4$-irreducible, by Lemma \ref{remove1}, no item alone is of worth $3/4$ to $\agent_i$ and thus for each partition $P_k \in \hat{B_1}$, the white items in $P_k$ have a value of at least ${1/4} \geq \epsilon_i$ to $\agent_i$. This implies that
\begin{equation}\label{pq4}
\sum_{P_k \in \hat{B_1}} \valu_i(\wcal(P_k)) \geq |\hat{B_1}|  \epsilon_i.
\end{equation}
By Inequalities \eqref{sumba},\eqref{pq1}, \eqref{pq2}, and \eqref{pq4} we have
\begin{equation}\label{enough}
\begin{split}
\valu_i(\wcal(\items)) &= \sum_{P_j \in B_0} \valu_i(\wcal(P_j)) + \sum_{P_j \in B_1} \valu_i(\wcal(P_j)) + \sum_{P_j \in B_2} \valu_i(\wcal(P_j)) \\
&\geq \sum_{P_j \in B_0} \valu_i(\wcal(P_j)) + \sum_{P_j \in B_1} \valu_i(\wcal(P_j)) \\
&\geq \sum_{P_j \in B_0} \valu_i(\wcal(P_j)) + \sum_{P_j \in \hat{B_1}} \valu_i(\wcal(P_j)) + \sum_{P_j \in \tilde{B_1}} \valu_i(\wcal(P_j)) \\
&\geq |B_0|  4\epsilon_i + |\hat{B_1}|  \epsilon_i + |\tilde{B_1}|  2\epsilon_i \\
&\geq |B_0|  4\epsilon_i + |B_1| 2\epsilon_i - |\hat{B_1}|  \epsilon_i  \\ 
&\geq |B_0|  4\epsilon_i + |B_1| 4\epsilon_i + |B_2|4\epsilon_i - |B_1| 2\epsilon_i - |B_2|4\epsilon_i- |\hat{B_1}|  \epsilon_i \\
&= (2n-2|B_2| - |B_1| -  |\hat{B_1}|)2\epsilon_i + (|\hat{B_1}|)\epsilon_i
\end{split}
\end{equation}
Note that the total value of white items that are assigned to the agents during the algorithm is equal to $\valu_i(\wcal(\items \setminus \fitems))$. The rest of the white items are still in $\fitems$. Thus, we have 
\begin{equation}
\label{fsum}
\valu_i(\wcal(\items)) = \valu_i(\wcal(\items \setminus \fitems)) + \valu_i(\fitems)
\end{equation}
Now, we provide an upper bound on the value of $\valu_i(\wcal(\items \setminus \fitems))$. As a warm up, one can trivially prove an upper bound of $2\epsilon_i (2n - 1 - |B_1| - 2|B_2|)$ on $\valu_i(\wcal(\items \setminus \fitems))$. This follows from the fact that two sets of items are assigned to any agent and hence we have a total of $2n$ disjoint sets. Among these $2n$ sets, at least one of them is empty (since $\secondset_i = \emptyset$) and at least $|B_1| + 2|B_2|$ of the sets contain a single blue item. On the other hand, by Lemmas \ref{c1small2}, \ref{cr2smallc1}, \ref{c3fsmall} and \ref{prvalue} every set with white items is of worth at most $2\epsilon_i$ to $\agent_i$. Therefore, the total value of the white items in $\items \setminus \fitems$ to $\agent_i$ is less than $2\epsilon_i (2n - 1 - |B_1| - 2|B_2|)$ and thus $$\valu_i(\wcal(\items \setminus \fitems)) \leq 2\epsilon_i (2n - 1 - |B_1| - 2|B_2|).$$ 

However, in order to complete the proof, we need a stronger upper bound on $\valu_i(\wcal(\items \setminus \fitems))$. To this end, we provide the following auxiliary lemma.
 
\begin{lemma}
\label{eps}
Let $\agent_j$ be an agent such that $|\firstset_j| = 1$ and $\valu_i(\firstset_j) > \valu_i(\firstset_i)$. Then, $\valu_i(\secondset_j) < \epsilon_i$.
\end{lemma} 
\begin{proof}
First, note that if $\agent_j$ is not satisfied yet, then $\secondset_j = \emptyset$ and therefore $\valu_i(\secondset_j) < \epsilon_i$. Otherwise, we argue that agent $\agent_j$ is either satisfied in the second phase, or in the refinement phases of $\cone$ or $\ctwo$.

Consider the case that $\agent_j$ is satisfied in the second phase. If $\agent_j \in \satagents_2^s \cup \satagents_3$, then by Lemma \ref{prvalue}, $\valu_i(\secondset_j) < \epsilon$ holds. Also, if $\agent_j \in \satagents_1^s$, considering the fact that $\agent_i$ envies $\agent_j$, $\agent_i \prec_{pr} \agent_j$. Thus, by Lemma \ref{prvalue}, we have $\valu_i(\secondset_j)< \epsilon_i$.

Next, consider the case that $\agent_j$ is in $\satagents_1^r \cup \satagents_2^r$. Note that the matching of the refinement phase of $\cone$ preserves the property described in Lemma \ref{nicematch}. Hence, if $\agent_j$ belongs to $\satagents_1^r$, then either $\agent_j \prec_{pr} \agent_i$ or there is no edge between $\agentv_i$ and $M_1(\agentv_j)$ in $G_1$, where $M_1(\agentv_j)$ is the vertex matched with $\agentv_j$ in $M_1$. If $\agent_j \prec_{pr} \agent_i$, according to Observation \ref{epsofcluster}, $\valu_i(\firstset_j) \leq 3/4-\epsilon_i$ holds. On the other hand, by the definition, if no edge exists between $\agentv_i$ and  $M_1(\agentv_j)$ in $G_1$, $\valu_i(\secondset_j)<\epsilon_i$. 
In addition to this, if $\agent_j$ belongs to $\satagents_2^r$, according to Lemma \ref{cr2smallc1}, $\valu_i(\secondset_j)< \epsilon_i$ holds. 
Therefore, Lemma \ref{eps} holds for the agents in $\satagents_1^r \cup \satagents_2^r$.

\end{proof}

Note that since matching $M$ of $G_{1/2}$ for building Cluster $\cone$ is $\MCMWM$ according to condition (iii) of Lemma \ref{wm}, there exists no agent $\agent_k$, such that $|\secondset_k| = 1$ and $\valu_i(\secondset_k) > 3/4 - \epsilon_i$. Otherwise, by assigning the item in $\secondset_j$ to $\agent_i$ instead of the item in $\firstset_i$, we can increase the total weight of the matching, that contradicts the maximality of $M$.

According to Lemma \ref{eps}, for all the agents $\agent_j$ with the property that $f_j$ is a blue item that belongs to a partition in $\hat{B_1}$, $\valu_i(\secondset_j)< \epsilon_i$ holds. 
The number of such agents is at least $|\hat{B_1}|$. Therefore, the total value of $\valu_i(\wcal(\items \setminus \fitems))$ is less than $2\epsilon_i \cdot (2n - 1 - |B_1| - 2|B_2|- |\hat{B_1}|) + \epsilon_i \cdot |\hat{B_1}|$. Combining the bounds obtained in Observation \ref{enough} and Lemma \ref{eps} by Inequality \eqref{fsum}, we have:
$$ \valu_i({\fitems}) \geq 2\epsilon_i \cdot(2n -2|B_2| - |B_1| -  |\hat{B_1}|) + \epsilon_i \cdot(|\hat{B_1}|) -  2\epsilon_i \cdot (2n - 1 - |B_1| - 2|B_2|- |\hat{B_1}|) - \epsilon_i \cdot |\hat{B_1}|$$
That is:
$$ \valu_i({\fitems}) \geq 2\epsilon_i$$

This contradicts the fact that the set $\fitems$ is not feasible for $\agent_i$.

\end{proof}

\begin{proof}[of Lemma \ref{c2null}]
Lemmas \ref{c3null} and \ref{c1null} state that at the end of the algorithm, $\cone = \cthree = \emptyset$. Now, let $\agent_i$ be a winner of $\ctwo$. We consider two cases separately: $\epsilon_i \geq {1/8}$ and $\epsilon_i < {1/8}$.

If $\epsilon_i \geq{1/8}$, the proof follows from a similar argument we used to prove Lemma \ref{c1null}. 

\begin{lemma}
\label{c2rem}
If $\epsilon_i \geq 1/8$, then the following inequality holds:
$$ \sum_{\agent_j \in \satagents} \valu_i( \firstset_j \cup \secondset_j) \leq |\satagents| + 1/8.$$
\end{lemma}
\begin{proof}
We know $\satagents = \satagents_1 \cup \satagents_2 \cup \satagents_3$. For every agent $\agent_j$ in $\satagents_3$, by Lemmas \ref{general} and \ref{c3fsmall}, we know $\valu_i(\firstset_j)<3/4$. Also, according to Lemma \ref{prvalue}, $\valu_i(\secondset_j)< \epsilon_i \leq 1/4$. Therefore,
\begin{equation}
\label{bow}
\sum_{\agent_j \in \satagents_3} \valu_i(\firstset_j \cup \secondset_j) \leq \sum_{\agent_j \in \satagents_3} (3/4 + 1/4) = |\satagents_3|.
\end{equation} 
Now, consider an agent $\agent_j \in \satagents_1$. Note that by Lemma \ref{forc2c3}, $\valu_i(\firstset_j)< 1/2$. Also, remark that either $\agent_j \in \satagents_1^r$  or $\agent_j \in \satagents_1^s$ . If $\agent_j \in \satagents_1^r$ then according to Lemma \ref{gsmallc1r}, $\valu_i(\secondset_j)< 1/2$ holds and hence $\valu_i(\firstset_j \cup \secondset_j)< 1$. Also, If $\agent_j \in \satagents_1^s$, then according to Lemma \ref{prvalue}, $\valu_i(\secondset_j)<2\epsilon_i<1/2$. 
Thus, in both cases, $\valu_i(\firstset_j \cup \secondset_j)< 1$ and hence:
\begin{equation}
\label{hogh}
\sum_{\agent_j \in \satagents_1} \valu_i(\firstset_j \cup \secondset_j) \leq \sum_{\agent_j \in \satagents_1} 1 = |\satagents_1|.
\end{equation} 

Finally consider a satisfied agent $\agent_j \in \satagents_2$. Again, remark that either $\agent_j \in \satagents_2^r$  or $\agent_j \in \satagents_2^s$ holds. 

Consider the case that $\agent_j \in \satagents_2^s$. If $\agent_j \prec_{pr} \agent_i$, then by Observation \ref{epsofcluster}, $\valu_i(\firstset_j) \leq 3/4 - \epsilon_i$ and by Lemma \ref{prvalue}, $\valu_i(\secondset_j) < 2\epsilon_i \leq 1/4 + \epsilon_i$ which means $\valu_i(\firstset_j \cup \secondset_j) < 1$. Moreover, if $\agent_i \prec_{pr} \agent_j$, according to Lemmas \ref{general} and \ref{prvalue}, $\valu_i(\firstset_j \cup \secondset_j) < 3/4 + \epsilon_i \leq 1 $. Thus, we have:

\begin{equation}
\label{f}
\sum_{\agent_j \in \satagents_2^s} \valu_i(\firstset_j \cup \secondset_j) \leq \sum_{\agent_j \in \satagents_2^s} 1 = |\satagents_1|
\end{equation}
It only remains to investigate the case where $\agent_j \in \satagents_2^r$. Note that since $\agent_i$ is not satisfied in the refinement phase of $\ctwo$, if $\agent_i \prec_{pr} \agent_j$, then $\valu_i(\secondset_j)< \epsilon_i \leq 1/4$. Otherwise, we could assign the item in $\secondset_j$ to $\agent_i$ in the refinement phase of $\ctwo$. Also, by Lemma \ref{general}, $\valu_i(\firstset_j)<3/4$ holds which yields $\valu_i(\firstset_j \cup \secondset_j)<1$.  

Finally, if $\agent_j \prec_{pr} \agent_i$, by Observation \ref{epsofcluster} $\valu_i(\firstset_j) \leq 3/4 - \epsilon_i$ holds. Corollary \ref{forc2small} states that there is at most one item $\ite_k$ with $\items_k \in \itemsv' \setminus \itemsv'_{1/2}$ and $\valu_i(\ite_k) \geq 3/8$. Also, note that since $\ite_k$ belongs to  $\itemsv' \setminus \itemsv'_{1/2}$, $\valu_i(\{\ite_k\})<1/2$ holds. For agent $\agent_j$, let $\ite_l$ be the item that is assigned to $\agent_j$ in the refinement of $\ctwo$, i.e., $\secondset_j = \{\ite_l\}$. We have $$\valu_i(\firstset_j \cup \secondset_j) \leq 3/4 - \epsilon_i + \valu_i(\{\ite_l\}).$$  
If $\ite_l \neq \ite_k$, $\valu_i(\firstset_j \cup \secondset_j) \leq 3/4 - \epsilon_i + 3/8$ holds which by the fact that $\epsilon_i \geq{1/8}$, implies $\valu_i(\firstset_j \cup \secondset_j) \leq 3/4 -1/8 + 3/8 \leq 1$. In addition to this, If $\ite_l = \ite_k$, 
$\valu_i(\firstset_j \cup \secondset_j) \leq 3/4 - 1/8 + 4/8 \leq 1+1/8$. But this can happen for at most one agent. Thus, for every agent $\agent_j$ in $\satagents_2^r$, $\valu_i(\firstset_j \cup \secondset_j) \leq 1$ holds and for at most one agent $\agent_j \in \satagents_2^r$, $\valu_i(\firstset_j \cup \secondset_j) \leq 1+1/8$. Thus, we have

\begin{equation}
\label{u}
\sum_{\agent_j \in \satagents_2^r} \valu_i(\firstset_j \cup \secondset_j) \leq |\satagents_2^r| + 1/8.
\end{equation}
Inequality (\ref{u}) together with Inequality (\ref{f}) yields
\begin{equation}
\label{oh}
\sum_{\agent_j \in \satagents_2} \valu_i(\firstset_j \cup \secondset_j) \leq |\satagents_2| + 1/8.
\end{equation}
Furthermore, by Inequalities (\ref{bow}), (\ref{hogh}) and (\ref{oh}) we have
\begin{equation}
\begin{split}
\sum_{\agent_j \in \satagents} \valu_i( \firstset_j \cup \secondset_j)  & = \sum_{\agent_j \in \satagents_1} \valu_i( \firstset_j \cup \secondset_j) + \sum_{\agent_j \in \satagents_2} \valu_i( \firstset_j \cup \secondset_j) + \sum_{\agent_j \in \satagents_3} \valu_i( \firstset_j \cup \secondset_j)\\
& \leq |\satagents_1| + |\satagents_2|+ 1/8 + |\satagents_3| \\
& \leq |\satagents|+1/8. \\
\end{split}
\end{equation}
 
\end{proof}

By Lemma \ref{c2rem}, value of agent $\agent_i$ for the items assigned to the satisfied agents is less than $|\satagents| + 1/8$. Recall that $\ctwo = \cthree = \emptyset$ and hence $|\satagents| = n - |\ctwo|$. Therefore,
\begin{equation}
\sum_{\agent_j \in \satagents} \valu_i( \firstset_j \cup \secondset_j) \leq n - |\ctwo| + 1/8.
\end{equation}
Since $\agent_i$ is a winner of $\ctwo$, for all $\agent_j \in \ctwo$, we have  $\valu_i(\firstset_j)\leq \valu_i(\firstset_i)$. On the other hand, since the total value of all items for $\agent_i$ is equal to $n$ we have

\begin{equation}\label{yy1}
\begin{split}
\valu_i({\fitems}) & = \valu_i(\items) - \sum_{\agent_j \in \ctwo} {\valu_i(\firstset_j)} - \sum_{\agent_j \in \satagents}\valu_i(\firstset_j \cup \secondset_j)\\
& = n - \sum_{\agent_j \in \ctwo} {\valu_i(\firstset_j)} - \sum_{\agent_j \in \satagents}\valu_i(\firstset_j \cup \secondset_j)\\
& \geq n - \sum_{\agent_j \in \ctwo} {\valu_i(\firstset_j)} - \big[ n - |\ctwo| + 1/8 \big] \\
& = |\ctwo| - 1/8 - \sum_{\agent_j \in \ctwo} {\valu_i(\firstset_j)}.
\end{split}
\end{equation}

Also, $\valu_i(\firstset_i) = {3/4} - \epsilon_i$ holds and $\valu_i(\firstset_j) \leq \valu_i(\firstset_i)$ for any $\agent_j \in \ctwo$ follows from the fact that $\agent_i$ is a winner of $\ctwo$. Therefore by Inequality (\ref{yy1}) we have
\begin{equation*}
\begin{split}
\valu_i({\fitems}) & \geq |\ctwo| - 1/8 - \sum_{\agent_j \in \ctwo} {\valu_i(\firstset_j)}\\
& \geq |\ctwo| - 1/8 - \sum_{\agent_j \in \ctwo} {\valu_i(\firstset_i)}\\
& = |\ctwo| - 1/8 - |\ctwo| {\valu_i(\firstset_i)}\\
& = |\ctwo| - 1/8 - |\ctwo| ({3/4} - \epsilon_i)\\
& = |\ctwo| ({1/4} + \epsilon_i) - 1/8.
\end{split}
\end{equation*}
Recall that by the assumption $\epsilon_i \geq 1/8$ holds. Moreover, $\epsilon_i\leq 1/4$, and thus
\begin{equation*}
\begin{split} 
\valu_i(\fitems) & \geq |\ctwo| ({1/4} + \epsilon_i) - 1/8\\
& \geq |\ctwo| 2\epsilon_i - 1/8\\
& \geq |\ctwo|  2\epsilon_i - \epsilon_i
\end{split}
\end{equation*}
and since $|\ctwo| \geq 1$, 
\begin{equation*}
\begin{split}
\valu_i(\fitems) &\geq |\ctwo|  2\epsilon_i - \epsilon_i\\
& \geq 2\epsilon_i - \epsilon_i\\
& \geq \epsilon_i
\end{split}
\end{equation*}
 and thus $\fitems$ is feasible for $\agent_i$. This contradicts the termination of the algorithm. 

Next, we investigate the case where $\epsilon < {1/8}$. Our proof for this case is similar to the one for $\cone$. Let $\satagents^r_{1}$ be the agents in $\satagents_1$ that are satisfied in the refinement phase and let $$\items^r_1 = \bigcup_{\agent_j \in \satagents^r_1} \firstset_j \cup \secondset_j.$$ Lemma \ref{forc2} states that the maxmin value of the agents in $\ctwo \cup \cthree$ for the items in $\items' = \items \setminus \items^r_1$ is at least $1$. More precisely for every $\agent_j \in \ctwo$:
\begin{equation}
\label{c1refine}
 \MMS_{j}^{n- |\satagents^r_{1}|} ( {\items} \setminus \items^r_{1}) \geq 1 
 \end{equation}  

We color the items of $\items'$ in one of four colors blue, red, green, or white. Initially, all the items are colored in white. For each agent $\agent_j \in \agents \setminus \satagents^r_1$, if $|\firstset_j|=1$, then we color the item in $\firstset_j$ in blue. Also, if $|\firstset_j|=2$ (which means $\firstset_j$ is corresponding to a merged vertex), color both the elements of $\firstset_j$ in red. In addition to this, if $|\secondset_j|=1$ then color the item in $\secondset_j$ in green. For any set $S \subseteq \items$, we denote the subset of blue, red, green, and white items in $S$ by $\mathcal{B}(S)$,$\mathcal{R}(S)$, $\mathcal{G}(S)$, and $\mathcal{W}(S)$, respectively. Recall that by Lemma \ref{pairsmall}, every pair of items in red or green are worth less that $3/4$ in total to $\agent_i$. In other words,
\begin{equation*}
\valu_i(\{\ite_j, \ite_k\}) \leq 3/4.
\end{equation*}
for any two different items $\ite_j, \ite_k \in \mathcal{B}(\items) \cup \mathcal{G}(\items)$.
  Also, according to Lemmas \ref{prvalue} and \ref{c3fsmall}, every set including white items is worth less than $2\epsilon_i < 1/4$ to $\agent_i$. 

Now, let $n' = n - |\satagents^r_1|$. Let $\cal P$ $= \langle P_1, P_2, \ldots, P_{n'} \rangle$ be the optimal $n'-$partitioning of $\items' $ for $\agent_i$. Recall that by Inequality \eqref{c1refine} the value of every partition in $\cal P$ is at least $1$ for $\agent_i$. Based on the number of blue and red items in every partition, we define three sets of partitions:
\begin{itemize}
    \item $B_{00}:$ Partitions with no red or blue items.
    \item $B_{10}:$ Partitions with blue items, but without any red items.
    \item $B_{01}:$ Partitions that contain at least one red item.
\end{itemize}

Next we prove Lemmas \ref{lowerwhite1} and \ref{lowerwhite2} to be used later in the proof. 
\begin{lemma}
\label{lowerwhite1}
Let $|{\mathcal{G}}(B_{00})|$ be the number of green items in the partitions of $B_{00}$. Then, $$\valu_i(\wcal(B_{00})) \geq (3|B_{00}| - |{\mathcal{G}}(B_{00})|)\cdot 1/4.$$
\end{lemma}
\begin{proof}
Let $B_{00}^j$ be the set of partitions in $B_{00}$ that contain exactly $j$ green items. We have:
\begin{equation}
\label{gbound}
|\mathcal{G}(B_{00})| = \sum_{1 \leq j < \infty} j  |B_{00}^j| \geq |B_{00}^1| + 2|B_{00}^2| + \sum_{3 \leq j < \infty} 3  |B_{00}^j| 
\end{equation}
Also, we have:
\begin{equation}
\label{bbound}
3 |B_{00}| = \sum_{0 \leq j < \infty} 3  |B_{00}^j| = 3  |B_{00}^0| + 3  |B_{00}^1| + 3  |B_{00}^2| + \sum_{3 \leq j < \infty} 3  |B_{00}^j|
\end{equation}
Finally, we argue that the value of white items in $B_{00}$ is at least $|B_{00}^0| + |B_{00}^1|\cdot 1/2 + |B_{00}^3|\cdot 1/4$. This follows from the fact that every green item in $P_k \in B_{00}^1$ has a value less than $1/2$ and by Lemma \ref{pairsmall}, every pair of green items in $P_k \in B_{00}^2$ are worth less than $3/4$ to $\agent_j$. According to the fact that the value of every partition $P_k$ is at least 1, we have:

\begin{equation}
\label{wbound}
\valu_i({\cal{W}}(B_{00})) \geq |B_{00}^0| + |B_{00}^1|\cdot 1/2 + |B_{00}^3|\cdot 1/4 = \big( 4|B_{00}^0| + 2|B_{00}^1| + |B_{00}^2| \big)\cdot 1/4
\end{equation}  

According to Equations \eqref{gbound} and \eqref{bbound}, we have:
\begin{equation}
\label{tbound}
3 |B_{00}| - |\mathcal{G}(B_{00})| \leq 3|B_{00}^0| + 2|B_{00}^1| + |B_{00}^2| \leq 4|B_{00}^0| + 2|B_{00}^1| + |B_{00}^2|
\end{equation}

Next we combine Equations \eqref{wbound} and \eqref{tbound} to obtain: 
\begin{equation}
\valu_i({\cal{W}}(B_{00})) \geq \big( 3 |B_{00}| - |\mathcal{G}(B_{00})| \big)\cdot 1/4
\end{equation}

\end{proof}

\begin{lemma}
\label{lowerwhite2}
$\valu_i(\wcal(B_{10})) \geq (2|B_{10}| - |{\mathcal{B}}(B_{10}) | - |{\mathcal{G}}(B_{10})| )\cdot 1/4$
\end{lemma}
\begin{proof}
First, note that every partition in $B_{10}$ contains at least one blue item. Let $B_{10}^{w}$ be the partitions in $B_{10}$ that contains exactly one blue item and no green item. The other items in each partition of $B_{10}^{w}$, are white. Since the problem is $3/4$-irreducible, the value of every blue item to $\agent_i$ is less than $3/4$ and therefore we have:
$$\valu_i(\wcal(B_{10})) \geq |B_{10}^{w}|\cdot 1/4$$
or 
\begin{equation}
\label{wbound2}
4\valu_i(\wcal(B_{10})) \geq |B_{10}^{w}|.
\end{equation}
Moreover, let $B_{10}^{\bar{w}} = B_{10} \setminus B_{10}^w$. Since every partition in $B_{10}$ contains at least one blue item, every partition in $B_{10}^{\bar{w}}$ contains at least two items with colors blue or green. Thus, we have:
\begin{equation}
\label{gbound2}
 |{\mathcal{G}}(B_{10}^{\bar{w}})| + |{\mathcal{B}}(B_{10}^{\bar{w}})| \geq 2|B_{10}^{\bar{w}}|
\end{equation}
Summing up Equations \eqref{wbound2} and \eqref{gbound2} results in 
$$
4\valu_i(\wcal(B_{10})) +  |{\mathcal{G}}(B_{10}^{\bar{w}})| + |{\mathcal{B}}(B_{10}^{\bar{w}})| \geq 2|B_{10}^{\bar{w}}| + |B_{10}^{w}|
$$
which means:
\begin{equation}
\label{fbound}
4\valu_i(\wcal(B_{10})) \geq 2|B_{10}^{\bar{w}}| -  |{\mathcal{G}}(B_{10}^{\bar{w}})| - |{\mathcal{B}}(B_{10}^{\bar{w}})|  + |B_{10}^{w}|.
\end{equation}

Morover, we have $|{\mathcal{B}}(B_{10})| = |{\mathcal{B}}(B_{10}^{w})|+|{\mathcal{B}}(B_{10}^{\bar{w}})|$. According to the fact that every partition in $B_{10}^w$ contains exactly one blue item, $|{\mathcal{B}}(B_{10}^{w})| = |B_{10}^w|$ and hence, $|{\mathcal{B}}(B_{10})| = |B_{10}^w|+|{\mathcal{B}}(B_{10}^{\bar{w}})|$. By Equation \eqref{fbound}, we have:
$$ 4\valu_i(\wcal(B_{10})) \geq 2|B_{10}^{\bar{w}}| -  |{\mathcal{G}}(B_{10}^{\bar{w}})| -|{\mathcal{B}}(B_{10})|+ |B_{10}^w|  + |B_{10}^{w}|. $$
Finally by the fact that $2|B_{10}^{w}| + 2|B_{10}^{\bar{w}}| = 2|B_{10}|$, we have:
$$ 4\valu_i(\wcal(B_{10})) \geq 2|B_{10}| -  |{\mathcal{G}}(B_{10}^{\bar{w}})| -|{\mathcal{B}}(B_{10})| $$
which is:
$$ \valu_i(\wcal(B_{10})) \geq \big(2|B_{10}| -|{\mathcal{B}}(B_{10})|-  |{\mathcal{G}}(B_{10}^{\bar{w}})| \big) \cdot 1/4$$
\end{proof}

For the partitions in $B_{01}$, we construct a graph $G_{01} \langle V_{01},E_{01} \rangle$, where every vertex $v_j \in V_{01}$ corresponds to a partition $P_j \in B_{01}$. Consider an agent $\agent_j$ such that $\firstset_j$ consists of a pair of red items $\ite_k,\ite_{k'}$ and let $\ite_k \in P_l$ and $\ite_{k'} \in P_{l'}$. We add an edge $(v_l,v_{l'})$ to $E_{01}$. By the definition of $B_{01}$, $P_l,P_{l'} \in B_{01}$ holds. Note that $\ite_k$ and $\ite_{k'}$ might belong to the same partition, i.e., $P_l = P_{l'}$. In this case, we add a loop to $G_{01}$. Furthermore, for every item $\ite_k \in {\mathcal B}(B_{01})$, we add a loop to the vertex $v_l$, where $\ite_k \in P_l$. 

Next, define $R_j$ as the set of partitions in $B_{01}$, such that the degree of their corresponding vertices in $V_{01}$ are equal to $j$. In other words:
$$P_k \in R_j \iff d(v_k)=j$$

Next we prove Lemma \ref{lowerwhite3}.

\begin{lemma}
\label{lowerwhite3}
For $R_1$, we have: $$\valu_i(\wcal(R_1)) \geq  (2|R_1| - |{\mathcal{G}}(R_1)|)\cdot 1/4 $$
\end{lemma}
\begin{proof}
Consider a partition $P_j \in R_1$. Since $d(v_j)=1$, $P_j$ contains exactly one red item and no blue item. Thus, other items in $P_j$ are either green or white. We show that 
\begin{equation}
\label{newway}
|{\mathcal G}(P_j)| + 4.\valu_i(\wcal(P_j)) \geq 2.
\end{equation}
First, argue that if $|{\mathcal G}(P_j)| \geq 2$, then Inequality \eqref{newway} holds. Also, if $|{\mathcal G}(P_j)|=0$, then $\valu_i(\wcal(P_j)) \geq 1/2$, because the value of the red item in $P_j$ is less than $1/2$ (recall that all the red items correspond to the vertices in $\itemsv' \setminus \itemsv'_{1/2}$). This immediately implies the fact that $4.\valu_i(\wcal(P_j)) \geq 2$. Finally, if $|{\mathcal G}(P_j)|=1$, then by Lemma \ref{pairsmall}, the total value of the green and red items in $P_j$ is less than $3/4$ and hence, $\valu_i(\wcal(P_j)) \geq 1/4$ which means $|{\mathcal G}(P_j)| + 4.\valu_i(\wcal(P_j)) \geq 2$.

Since Inequality \eqref{newway} holds for every partition $P_j \in R_1$, we have:

$$\sum_{P_j \in R_1} \big(|{\mathcal G}(P_j)| + 4.\valu_i(\wcal(P_j))\big) \geq 2 |R_1|$$
Therefore,
 $$|{\mathcal G}(R_1)| + 4.\valu_i(\wcal(R_1)) \geq 2|R_1|$$ and hence, $$\valu_i(\wcal(R_1)) \geq  (2|R_1| - |{\mathcal{G}}(R_1)|)\cdot 1/4 $$
\end{proof}

\begin{lemma}
\label{lowerwhite4}
For $R_2$, we have: $$ \valu_i(\wcal(R_2)) \geq (|R_2| - |{\mathcal{G}}(R_2)|)\cdot 1/4 $$
\end{lemma}
\begin{proof}
Let $P_j$ be a partition in $R_2$. First, we show the following inequality holds:
\begin{equation}
\label{gwbound2}
4\valu_i(\wcal(P_j)) + |{\mathcal G}(P_j)| \geq 1
\end{equation}

By the definition of $R_2$, degree of $v_j$ is $2$. Therefore, $P_j$ contains two red items. Note that the degree of the partitions in $B_{01}$ that contain blue items is at least $3$. Thus, $P_j$ contains no blue items. By Lemma \ref{pairsmall}, the total value of the red items in $P_j$ is less than $3/4$. The rest of the items in $P_j$ are either green or white. If $P_j$ contains a green item, then Inequality \eqref{gwbound2} holds. On the other hand, if $P_j$ contains no green items, then $\valu_i(\wcal(P_j)) \geq 1/4$ and hence, $4\valu_i(\wcal(P_j)) \geq 1$. Therefore, Inequality \eqref{gwbound2} holds in both cases. 

Summing up Inequality \eqref{gwbound2} for all the partitions in $R_2$, we have:

$$
\sum_{P_j \in R_2} 4\valu_i(\wcal(P_j)) + |{\mathcal G}(P_j)| \geq |R_2|
$$
which means:
$$
4\valu_i(\wcal(R_2)) + |{\mathcal G}(R_2)| \geq |R_2|
$$
That is:
$$
\valu_i(\wcal(R_2)) \geq \big(|R_2| - |{\mathcal G}(R_2)|\big)\cdot 1/4  
$$
\end{proof}

Putting together Lemmas \ref{lowerwhite1},\ref{lowerwhite2},\ref{lowerwhite3}, and \ref{lowerwhite4} we obtain the following lower bound on the valuation of $\agent_i$ for all white items:
\begin{equation}\label{w1}
\begin{split}
\valu_i(\wcal(\items')) & = \valu_i(\wcal(B_{00})) + \valu_i(\wcal(B_{01})) + \valu_i(\wcal(B_{10}))\\
& \geq \bigg(3|B_{00}| - |{\mathcal{G}}(B_{00})|\bigg)\cdot 1/4 + \bigg(2|B_{10}| - |{\mathcal{B}}(B_{10})| - |{\mathcal{G}}(B_{10})|\bigg)\cdot 1/4 \\ & \hspace{10pt}+ \bigg(2|R_1| - |{\mathcal{G}}(R_1)|\bigg)\cdot 1/4 + \bigg(|R_2| - |{\mathcal{G}}(R_2)|\bigg)\cdot 1/4\\
&= \bigg((3|B_{00}| + 2|B_{10}| + 2|R_1| + |R_2| - |{\mathcal{B}}(B_{10})|) - \big(|{\mathcal{G}}(B_{00})| + |{\mathcal{G}}(B_{10})| + |{\mathcal{G}}(R_1)| + |{\mathcal{G}}(R_2)|\big) \bigg)\cdot 1/4\\
& \geq \bigg((3|B_{00}| + 2|B_{10}| + 2|R_1| + |R_2| )- |{\mathcal{B}}(B_{10})| -  |{\mathcal{G}}(\items')| \bigg)\cdot 1/4
\end{split}
\end{equation}
where $|{\mathcal{G}}(\items')|$ is the total number of green items.

The items in $\wcal(\items')$ are either allocated to an agent during the second phase, or are still in $\fitems$. Let $\wcal_2$ be the white items that are allocated to an agent during the second phase. We have: 
\begin{equation}
\valu_i(\wcal(\items')) = \valu_i(\wcal_2) + V_i(\fitems)
\end{equation}
Now, we present an upper bound on the value of $\valu_i(\wcal_2)$. First, note that the number of agents in $\satagents \setminus \satagents^r_1$ is $n'$. Each of these $n'$ agents has two sets $\firstset_j$ and $\secondset_j$, that leaves us $2n'$ sets. Since $\secondset_i = \emptyset$ we know that at least one of these sets is empty. Moreover, of all these $|{\mathcal{G}}(\items')|$ sets contain a single green item and $|{\mathcal{B}}(B_{10})| + |E_{01}|$ of the sets contain either a single blue item, or a pair of red items (recall that each edge of $G_{01}$ refers to a blue item or a pair of red items). Therefore, the number of the sets that contain only white items is at most:
$$2n' - 1 - |{\mathcal{G}}(\items')| - |{\mathcal{B}}(B_{10})| - |E_{01}|$$

By Lemmas \ref{prvalue} and \ref{c3fsmall}, the value of every set with white items to $\agent_i$ is less than $2\epsilon_i<1/4$ and hence:
\begin{equation}
\label{w_2}
\valu_i(\wcal_2) \leq (2n' - 1 - |{\mathcal{G}}(\items')| - |{\mathcal{B}}(B_{10})| - |E_{01}|)\cdot 1/4
\end{equation}
Subtracting the lower bound obtained for $\valu_i(\wcal(\items'))$ in \eqref{w1} from the upper bound for $\valu_i(\wcal_2)$ in $\eqref{w_2}$ gives us a lower bound on the value of $\fitems$:
\begin{equation}\label{bachekhoshgel}
\begin{split}
\valu_i(\fitems) &= \valu_i(\wcal(\items')) - \valu_i(\wcal_2)\\
 &\geq \bigg((3|B_{00}| + 2|B_{10}| -|{\mathcal{B}}(B_{10})| + 2|R_1| + |R_2|) -  |{\mathcal{G}}(\items')| \bigg)\cdot 1/4 - \valu_i(\wcal_2)\\
 &\geq \bigg((3|B_{00}| + 2|B_{10}| - |{\mathcal{B}}(B_{10})| + 2|R_1| + |R_2|) -  |{\mathcal{G}}(\items')| \bigg)\cdot 1/4\\
 & \qquad - \bigg(2n' - 1 - |{\mathcal{G}}(\items')| - |{\mathcal{B}}(B_{10})| - |E_{01}|\bigg)\cdot 1/4 \\
&= \bigg(3|B_{00}|+2|B_{10}| + 2|R_1| + |R_2| -2n' + 1 + |E_{01}| \bigg)\cdot 1/4 \\
&= \bigg(2|B_{00}|+2|B_{10}| + |B_{00}| + |E_{01}| + 2|R_1| + |R_2| - 2n' + 1\bigg) \cdot 1/4 
\end{split}
\end{equation} 
Next we provide Lemmas \ref{ebound}, \ref{B00size}, and \ref{Esize} to complete the proof.
\begin{lemma}
\label{ebound}
$|B_{00}| \geq |E_{01}| - |B_{01}|$ 
\end{lemma}
\begin{proof}
First, note that $|B_{00}| + |B_{10}| + |B_{01}|=n'$. Moreover we have $|{\mathcal B}(B_{10})| + |E_{01}| \leq n'$. To show this Lemma, note that each edge in $G_{01}$ corresponds to the first set of an agent in $\satagents \setminus \satagents_1^r$. Also, every blue item in $B_{10}$ corresponds to the first set of an agent in $\satagents \setminus \satagents_1^r$. Therefore, the total number of the agents must be more than this number. By the definition of $B_{10}$, we know that $|{\mathcal B}(B_{10})| \geq |B_{10}|$. Therefore, we have: 

\begin{equation}
\begin{split}
|B_{00}| + |B_{10}| + |B_{01}| &\geq  |{\mathcal B}(B_{10})| + |E_{01}|\\
 &\geq |(B_{10})| + |E_{01}|\\
\end{split}
\end{equation} 
This means:
$$ |B_{00}| \geq |E_{01}| - |B_{01}| $$
\end{proof}

\begin{lemma}
\label{Esize}
$|E_{01}| \geq 3/2 \sum_{j \geq 3}|R_j| + |R_2| + |R_1|/2$ 
\end{lemma}
\begin{proof}
$|E_{01}| = \frac{\sum_{v_j \in V_{01}} d(v_j)}{2} = \frac{\sum_j j|R_j|}{2} \geq 3/2 \sum_{j \geq 3}|R_j| + |R_2| + |R_1|/2.$
\end{proof}

\begin{lemma}
\label{B00size}
$|B_{00}| \geq \frac{\sum_{j \geq 3}|R_j| - |R_1|}{2}$
\end{lemma}
\begin{proof}
By Lemma \ref{ebound}, $|B_{00}| \geq |E_{01}| - |B_{01}| $. Furthermore, by Lemma \ref{Esize}, $$|E_{01}| \geq 3/2 \sum_{j \geq 3}|R_j| + |R_2| + |R_1|/2.$$ By these two inequalities, we have:
\begin{equation}
\label{Ebd}
|B_{00}| \geq 3/2 \sum_{j \geq 3}|R_j| + |R_2| + |R_1|/2 - |B_{01}| 
\end{equation}
Also, since there is a one-to-one correspondence between $B_{01}$ and $V_{01}$, $|B_{01}| = |V_{01}|$ holds. By the definition of $R_j$, we have: 
\begin{equation}
\label{Vbd}
|V_{01}| = \sum_j |R_j| 
\end{equation}

By replacing the value obtained for $B_{01}$ from \eqref{Vbd} into Inequality \eqref{Ebd}, we have:
\begin{equation}
\begin{split}
|B_{00}| &\geq 1/2 \sum_{j \geq 3}|R_j| - |R_1|/2 \\
& = \frac{\sum_{j \geq 3}|R_j| - |R_1|}{2}.
\end{split} 
\end{equation}
  
\end{proof}

By applying Lemmas \ref{B00size} and \ref{Esize} to Inequality \eqref{bachekhoshgel} we have:
\begin{equation*}
\begin{split}
\valu_i(\fitems) & = \bigg(2|B_{00}|+2|B_{10}| + |B_{00}| + |E_{01}| + 2|R_1| + |R_2| - 2n' + 1\bigg) \cdot 1/4\\
&\geq \bigg(2|B_{00}|+2|B_{10}| + \frac{\sum_{j \geq 3}|R_j| - |R_1|}{2} + 3/2 \sum_{j \geq 3}|R_j| + |R_2| + |R_1|/2 + 2|R_1| + |R_2| - 2n' + 1\bigg)\cdot 1/4 \\
&= \bigg(2|B_{00}|+2|B_{10}| + \sum_{j \geq 3} 2|R_j| + 2|R_2| + 2|R_1| - 2n' + 1 \bigg)\cdot 1/4
\end{split}
\end{equation*} 
Finally, note that $\sum_{j \geq 3} 2|R_j| + 2|R_2| + 2|R_1| = 2|V_{01}| = 2|B_{01}|$. This, together with the fact that $|B_{00}| + |B_{01}| + |B_{10}| = n'$, yields $\valu_i(\fitems) \geq (2n' - 2n' + 1)\cdot 1/4$. This means $\valu_i(\fitems) \geq 1/4$ which is a contradiction since $\fitems$ is feasible for $\agent_i$.
\end{proof}

\section{Omitted Proofs of Section \ref{submodular}}\label{submodular-appendix}

\begin{observation}\label{obs_E1}
$f^x(S) \leq x$ for every given $S$.
\end{observation}

\begin{observation}\label{obs_E2}
$f^x(S) \leq f(S)$ for every given $S$.
\end{observation}

\begin{proof}[Of Lemma \ref{ceilingfunctions}]
\\
\textbf{First Claim:} By definition of submodular functions, for given sets $A$ and $B$ we have:
$$ f(A \cup B) \leq f(A) + f(B) - f(A \cap B)  $$
We prove that $f^x(.)$ is a submodular function in three different cases:\\

First Case: Let both $f(A)$ and $f(B)$ be at least $x$. According to Observation \ref{obs_E1}, $f^x(A \cup B)$ and $f^x(A \cap B)$ are bounded by $x$. Therefore, $f^x(A \cup B) + f^x(A \cap B) \leq 2x$, which yields: 
$$f^x(A \cup B) + f^x(A \cap B) \leq f^x(A) + f^x(B)$$

Second Case: In this case one of $f(A)$ and $f(B)$ is at least $x$. We have $f(A \cup B) \geq x$ and $f(A \cap B)$ is no more than max $\{f(A), f(B)\}$. As a result $f^x(A \cup B)$ and one of $f^x(A)$ or $f^x(B)$ are equal to $x$ which yields:
$$f^x(A \cup B) + f^x(A \cap B) \leq f^x(A) + f^x(B)$$

Third Case: In this case both $f(A)$ and $f(B)$ are less than $x$, and $f(A \cap B)$ is less than $x$ too. Since $f^x(A) = f(A)$, $f^x(B) = f(B)$, $f^x(A \cap B) = f(A \cap B)$, according to Observation \ref{obs_E2}, $f^x(A \cup B) \leq f(A \cup B)$ holds. Since $f(.)$ is a submodular function, we conclude that:
$$f^x(A \cup B) \leq f^x(A) + f^x(B) - f^x(A \cap B).$$\\
\textbf{Second Claim:} Since $f(.)$ is an XOS set function, by definition, there exists a finite set of additive functions  $\{f_1, f_2, \ldots, f_{\alpha}\}$ such that $$f(S) = \max_{i=1}^{\alpha} f_i(S)$$ for any set $S \subseteq \domp(f)$. With that in hand, for a given real number $x$, we define an XOS set function $g(.)$, and show $g(.)$ is equal to $f^{x}(.)$.

We define $g(.)$ on the same domain as $f(.)$. Moreover, based on $\{f_1, f_2, \ldots, f_{\alpha}\}$, we define a finite set of additive functions $\{g_1, g_2, \ldots, g_{\beta}\}$ that describe $g$. More precisely, for each set $S$ in domain of $f(.)$ we define a new additive function like $g_{\gamma}$ in $g(.)$ as follows: Without loss of generality let $f_{\delta}$ be the function which maximizes $f(S)$. For each $b_i \notin S$ let $g_{\gamma}(b_i) = 0$. Furthermore, for each $b_i \in S$ if $f(S) \leq x$ let $g_{\gamma}(b_i) = f_{\delta}(b_i)$, and otherwise let $g_{\gamma}(b_i) = \frac{x}{f(S)} f_{\delta}(b_i)$. 

We claim that $g(.)$ is equivalent to $f^x(.)$, which implies $f^x(.)$ is an XOS function. $g(.)$ and $f^x(.)$ are two functions which have equal domains. First, we prove that $g(S) \leq f(S)$ for any given set $S$. According to construction of $g(.)$, for each additive function in $g(.)$ such $g_{\gamma}$, there is at least one additive function in $f(.)$ such $f_{\delta}$ where $g_{\gamma}(b_i) \leq f_{\delta}(b_i)$ for each $b_i \in \items$. Therefore, for any given set $S$ we have:
 \begin{equation}\label{hineq1} g(S) \leq f(S)  \end{equation} 
Now, according to the construction of $g(.)$, for any given set $S$ where $f(S) \leq x$, we have a function $g_{\gamma}(S) = f(S)$, and where $f(S) > x$, we have a function $g_{\gamma}(S) = x$. Therefore, we can conclude that: 
 \begin{equation}\label{hineq2} g(S) \geq f^x(S)  \end{equation} 
 
For any given set $S$ where $f(S) \leq x$, according to the definition of $f^x(.)$, $f(S) = f^x(S)$, and using Inequalities \eqref{hineq1} and \eqref{hineq2} we argue that $f^x(S) = g(S)$. Moreover, according to the construction of $g(.)$, $g(S) \leq x$ for any given set $S$. Therefore, for any given set $S$ where $f(S) > x$, according to the definition of $f^x(.)$ and Inequality \eqref{hineq2}, $f^x(S) = g(S) = x$. As a result, by considering these two cases we argue that $f^x(.)$ and $g(.)$ are equivalent, which shows $f^x(.)$ is an XOS function.\\
\textbf{Third Claim:} In this claim, we use a similar argument to the first claim. By definition of subadditive functions for any given sets $A$ and $B$, we have:
$$f(A \cup B) \leq f(A) + f(B)$$
We prove that $f^x(.)$ meets the definition of subadditive functions by considering two different cases. In the first case at least one of $f(A)$ and $f(B)$ is at least $x$, and in the second case both $f(A)$ and $f(B)$ is less than $x$.\\

First Case: In this case $f^x(A) + f^x(B)$ is at least $x$, and since $f^x(S) \leq x$ for any given set $S$, $f^x(A \cup B) \leq x$. Therefore, 
$$f^x(A \cup B) \leq f^x(A) + f^x(B)$$

Second Case: Since $f^x(A \cup B) \leq f(A \cup B)$, $f(A \cup B) \leq f(A) + f(B)$, $f(A) = f^x(A)$, and $f(B) = f^x(B)$, we have:
$$f^x(A \cup B) \leq f^x(A) + f^x(B)$$
 
\end{proof}

\begin{proof}[Of Lemma \ref{submodularaval}]
Since $f(.)$ is submodular, according to the definition of submodular functions, for every given sets $X$ and $Y$ in domain of $f(.)$ with $X \subseteq Y$ and every $x \in \items \setminus Y$ we have:
\begin{equation} \label{hineq5} f(X \cup \{x\}) - f(X) \geq f(Y \cup \{x\}) - f(Y) \end{equation}

Let $S_i = \{e_1, e_2, \ldots, e_{\alpha}\}$, $T_0 = \emptyset$, and $T_j = \{e_1, e_2, \ldots, e_j\}$, for every $1 \leq j \leq \alpha$. Since $T_j \subseteq S_i$ for each $0 \leq j \leq \alpha$ and $f_i$ is a submodular function, according to Inequality \eqref{hineq5} we have:
\begin{equation} \label{hineq6} \sum_{1 \leq j \leq \alpha} f_i(S_i \setminus T_{j-1}) - f_i(S_i \setminus T_j) \geq \sum_{1 \leq j \leq \alpha} f_i(S_i) - f_i(S_i - e_j) \end{equation}

Since $f_i(S_i) = \sum_{1 \leq j \leq \alpha} f_i(S_i \setminus T_{j-1}) - f_i(S_i \setminus T_j)$, we can rewrite Inequality \eqref{hineq6} for every $1 \leq i \leq k$ as follows:
\begin{equation} \label{hineq7} f_i(S_i) \geq \sum_{e \in S_i} f_i(S_i) - f_i(S_i - e) \end{equation}

For every $1 \leq i \leq k$ we can rewrite Inequality \eqref{hineq7} as follows:
\begin{equation} \label{hineq8} \sum_{e \in s_i} f_i(S_i-e) \geq (|S_i| - 1) f_i(S_i) \end{equation}

By adding $(|\bigcup S_i| - |S_i|) f_i(S_i)$ to the both sides of Inequality \eqref{hineq8}, we have:
\begin{equation} 
\label{hineq9} 
 \begin{split}
 (|\bigcup S_i| - |S_i|) f_i(S_i) + \sum_{e \in S_i} f_i(S_i - e) &= \sum_{e \in \bigcup S_i} f_i(S_i \setminus \{e\}) \\
 &\geq (|\bigcup S_i| - 1) f_i(S_i)
 \end{split}
\end{equation}

Since Inequality \eqref{hineq9} holds for every $1 \leq i \leq k$, we can sum up both sides of Inequality \eqref{hineq9} as follows:

\begin{equation} \label{hineq10} \sum_{1 \leq i \leq k}\sum_{e \in \bigcup S_i} f_i(S_i - e) \geq \sum_{1 \leq i \leq k} (|\bigcup S_i| - 1) f_i(S_i) \end{equation}

By dividing both sides of Inequality \eqref{hineq10} over $1/|\bigcup S_i|$ we obtain:

\begin{equation}
\label{hineq11} 
\begin{split}
    \frac{1}{|\bigcup S_i|}(\sum_{e \in \bigcup S_i} \sum_{1 \leq i \leq k} f_i(S_i - e)) &= \mathbb{E}[\sum_{1 \leq i \leq k} f_i(S^*_i)]\\
    &\geq \sum_{1 \leq i \leq k} f_i(S_i)\frac{|\bigcup S_i| -1}{|\bigcup S_i|}.
\end{split}
\end{equation}

\end{proof}

\begin{proof}[Of Lemma \ref{submodulardovom}]
Similar to the proof of Lemma \ref{submodularaval}, we use Inequality \eqref{hineq5} as a definition of submodular functions. Let $S'_i = S_i \setminus S = \{e_1, e_2, \ldots, e_{\alpha}\}$, $T_0 = S$, and $T_j = S \cup \{e_1, e_2, \ldots, e_j\}$ for $1 \leq j \leq \alpha$. According to $f(S) < 1/3$, $f(S \cup S'_i) \geq 1$, and Inequality \eqref{hineq5} as a definition of sub-modular functions, we have:

\begin{equation}
\label{hineq12} 
\begin{split}
    2/3 &< f(S \cup S') - f(S)\\
    &= \sum_{1 \leq j \leq \alpha} f(T_{j-1} \cup \{e_j\}) - f(T_{j-1})\\
    &\leq \sum_{e \in S'_i} f(S \cup \{e\}) - f(S)
\end{split}
\end{equation}

Similar to Inequality \eqref{hineq10}, we can rewrite Inequality \eqref{hineq12} with a summation, since Inequality \eqref{hineq12} holds for any $1 \leq i \leq k$.

\begin{equation}
    \label{hineq13}
    2k/3 < \sum_{1 \leq i \leq k} \sum_{e \in S'_i} f(S \cup \{e\}) - f(S)
\end{equation}

By dividing both sides of Inequality \eqref{hineq13} over $1/ |\bigcup S_i \setminus S|$ we have:

\begin{equation}
    \label{hineq14}
    \begin{split}
    \frac{2k/3}{|\bigcup S_i \setminus S|} &< \frac{1}{|\bigcup S_i \setminus S|}(\sum_{1 \leq i \leq k} \sum_{e \in S'_i} f(S \cup \{e\}) - f(S))\\
    &= \mathbb{E}[f(S \cup \{e\}) - f(S)]
    \end{split}
\end{equation}
\end{proof}
\section{Omitted Proofs of Section \ref{xos}}\label{xosappendix}

\begin{proof}[of Lemma \ref{xos2lemma}]
According to the definition of XOS function, $f(.)$ is an XOS function with a finite set of additive functions $\{g_1, g_2, \ldots, g_{\alpha}\}$ where $f(S) = \max_{i=1}^{\alpha} g_i(S)$ for any set $S \in \domp(f)$. Let $g_j(.)$ be the additive function which maximizes $S$. Let $g_j(S_1) = \alpha_1, g_j(S_2) = \alpha_2, \ldots, g_j(S_k) = \alpha_k$, which yields $\beta = \sum \alpha_i$. Since $g_j(S_i) = \alpha_i$, $f(S \setminus S_i) \geq \beta - \alpha_i$. Therefore, we have:

\begin{equation}
    \label{hineq15}
    \begin{split}
    \sum f(S) - f(S \setminus S_i) &\leq \sum \beta - (\beta - \alpha_i)\\
    &= \beta\\
    &= f(S)
    \end{split}
\end{equation}

\end{proof}

\begin{proof}[of Lemma \ref{2nsets}]
According to the definition of $\MMS$, we know that $a_{i}$ can divide items to $n$ sets ${\cal P} = \langle P_1, P_2, \ldots, P_n \rangle$ such that $V_i(P_j) \geq 1$ for any $P_j$. The catch is that $a_i$ can divide each of these $n$ sets to two disjoint sets such that the value of each of these new sets be at least $2/5$ to him. Let $T = \{b_1, b_2, \ldots, b_\gamma\}$ be one of these $n$ sets, and $g_j(.)$ be an additive function which maximizes $V_i(T)$. Let $T_k = \{b_1, b_2, \ldots, b_k\}$ for any $1 \leq k \leq \gamma$. According to Lemma \ref{remove1}, since the problem is $1/5$-irreducible, the value of any item is less than $1/5$ to $a_i$. Therefore, there is a set $T_k$ among $T_1$ to $T_\gamma$ where $2/5 \leq g_j(T_k) < 3/5$. Since $g_j(.)$ is one of additive functions of XOS function $V_i$, we have $V_i(T_k) \geq 2/5$. Moreover, since $g_j(T_k) < 3/5$, $g_j(T \setminus T_k) \geq 2/5$, which yields $V_i(T \setminus T_k) \geq 2/5$. As a conclusion, we can divide each of $n$ sets to two disjoint sets with at least $2/5$ value to $a_i$.

\end{proof}


\end{document}